\documentclass[ejsv2,authoryear]{imsart}

\RequirePackage{natbib}
\RequirePackage[colorlinks,citecolor=blue,urlcolor=blue]{hyperref}
\RequirePackage{graphicx}

\usepackage{algorithm}
\usepackage{algorithmic}
\usepackage{array}
\usepackage{bm}
\usepackage{bbm}
\usepackage{color}
\usepackage{float}
\usepackage{booktabs}
\usepackage{bigints}
\usepackage{extarrows}
\usepackage{multirow}
\usepackage{subcaption}
\usepackage[bbgreekl]{mathbbol}
\usepackage{mathtools}
\usepackage[section]{placeins}
\usepackage{verbatim}
\usepackage{enumitem}
\usepackage{xcolor}
\usepackage{optidef}
\usepackage{graphicx}
\usepackage{rotating}

\startlocaldefs
\theoremstyle{plain}
\newtheorem{theorem}{Theorem}[section]
\newtheorem{lemma}[theorem]{Lemma}
\newtheorem{remark}{Remark}[section]

\theoremstyle{definition}
\newtheorem{definition}[theorem]{Definition}

\theoremstyle{remark}

\numberwithin{equation}{section}
\numberwithin{theorem}{section}
\numberwithin{proposition}{section}
\numberwithin{corollary}{section}
\numberwithin{remark}{section}
\numberwithin{table}{section}
\numberwithin{figure}{section}

\endlocaldefs

\usepackage[most]{tcolorbox}

\newtcolorbox{highlightremark}{
    colback=white,
    colframe=black,
    boxrule=0.5pt,
    arc=1mm,
    left=2mm,
    right=2mm,
    top=1mm,
    bottom=1mm
}

\newcommand{\mE}{\mathbb{E}}
\newcommand{\mP}{\mathbb{P}}
\newcommand{\bh}{\mathbf{h}}
\newcommand{\bA}{\mathbf{A}}
\newcommand{\bB}{\mathbf{B}}

\newcommand{\bD}{\mathbf{D}}

\newcommand{\bG}{\mathbf{G}}

\newcommand{\bJ}{\mathbf{J}}

\newcommand{\bS}{\mathbf{S}}

\newcommand{\bX}{\mathbf{X}}
\newcommand{\bY}{\mathbf{Y}}
\newcommand{\bZ}{\mathbf{Z}}
\newcommand{\bz}{\mathbf{z}}

\newcommand{\calA}{\mathcal{A}}

\newcommand{\calE}{\mathcal{E}}

\newcommand{\calG}{\mathcal{G}}
\newcommand{\calH}{\mathcal{H}}

\newcommand{\calJ}{\mathcal{J}}
\newcommand{\calK}{\mathcal{K}}

\newcommand{\calP}{\mathcal{P}}

\newcommand{\calS}{\mathcal{S}}
\newcommand{\calT}{\mathcal{T}}
\newcommand{\calU}{\mathcal{U}}
\newcommand{\calV}{\mathcal{V}}

\newcommand{\calX}{\mathcal{X}}

\newcommand{\cov}{\operatorname{Cov}}
\newcommand{\SC}{\calV_{\rm sc}}
\newcommand{\MC}{\calV_{\rm mc}}
\newcommand{\DISC}{\calV_{\rm disc}}

\newcommand{\bOmega}{\mathbf{\Omega}}
\newcommand{\tr}{{\rm tr} }
\newcommand{\bu}{\widetilde{\bh}}


\begin{document}

\begin{frontmatter}
\title{Detecting Structural Changes in High-Dimensional Multivariate Regression Models}
\runtitle{Change Point Detection in Multivariate Regression}

\begin{aug}
\author[A]{\fnms{Haoran}~\snm{Li}\ead[label=e1]{hzl0152@auburn.edu}}
\address[A]{{Department of Mathematics and Statistics}, {Auburn University}\\ {{305 W Samford Avenue}, {36849}, {AL}, {USA}}\printead[presep={,\ }]{e1}}
\runauthor{H. Li}
\end{aug}

\begin{abstract}
We study structural-change testing in multivariate linear regression when the response dimension is proportional to the sample size and the number of predictors is fixed. Despite its relevance to applications across a broad range of fields, this problem remains underexplored. The alternatives of interest allow multiple unknown changes in a prescribed linear contrast of the predictor effects. We construct a least-squares-based Wald statistic standardized by the residual covariance estimator and scan it over candidate change-point segmentations. We propose flexible segment scans for both single and multiple change points, together with a discretized multiscale scan that reduces the computational cost. When the response dimension and sample size diverge proportionally, we establish weak convergence of the normalized statistic process to a centered Gaussian process, yielding implementable critical values for all three scans. We further characterize their asymptotic power under local alternatives, explicitly describing how the signal, design, and high-dimensional aspect ratio determine the limiting power. The finite-sample performance is examined through simulation studies. We apply the proposed methods to filtered and standardized returns from U.S. equity portfolios to investigate structural changes in their exposures to the Fama--French factors.

\end{abstract}
\end{frontmatter}

\section{Introduction}\label{sec:intro}

Detecting structural changes in regression models is a longstanding problem in statistics and econometrics. Early work focused on the testing and estimation of structural breaks in univariate linear regression models; see, for example, \citet{kim1989likelihood,horvath1995detecting,Bai1997,BaiPerron1998,BaiPerron2003}. These methods provide useful tools for assessing the stability of regression relationships in economics, finance, environmental science, and engineering. \citet{qu2007estimating} extended this framework to multivariate regression models with multiple structural breaks. The classical theory underlying these procedures assumes that both the response and predictor dimensions remain fixed as the sample size increases.

Many modern applications involve a high-dimensional response and comparatively few explanatory variables, placing them outside this fixed-dimensional framework. Examples include neuroimaging studies, in which brain activity is measured simultaneously across hundreds or thousands of regions under a limited number of experimental conditions; genomic studies, in which thousands of gene expression levels are related to a small number of clinical or environmental variables; and financial studies, in which returns on a large portfolio of assets are explained by a small number of economic factors. A central scientific question in such applications is whether the relationships between the predictors and the multivariate response remain stable over time. We address this question through structural-change testing in multivariate regression models with high-dimensional responses.

For $j=1,\ldots,n$, let $\bY_j\in\mathbb{R}^p$ denote a multivariate response and $\bX_j\in\mathbb{R}^M$ the corresponding predictor vector. Throughout, the predictors are treated as nonrandom and their dimension $M$ is fixed, whereas the response dimension $p$ may diverge with the sample size $n$. The responses follow the multivariate linear regression model
$$
\bY_j=\beta_j^T\bX_j+\Sigma_p^{1/2}\bz_j,\qquad j=1,\ldots,n,
$$
where $\beta_j\in\mathbb{R}^{M\times p}$ is the regression coefficient matrix, $\Sigma_p$ is an unknown positive definite covariance matrix, and $\{\bz_j\}_{j=1}^n$ are independent and identically distributed random vectors whose entries are independent with mean zero and variance one. The covariance matrix $\Sigma_p$ is assumed to remain constant over time. Under the alternative, the coefficient matrix is piecewise constant with unknown break points $1<B_1<B_2<\cdots<B_s<n$, where the number of breaks $s$ is fixed but unknown. For notational convenience, set $B_0=0$ and $B_{s+1}=n$.

Our primary inferential target is a prescribed linear contrast of the predictor effects. For a user-specified nonzero vector $\theta\in\mathbb{R}^M$, the contrast of interest is $\beta_j^T\theta\in\mathbb{R}^p$. For example, choosing $\theta=(1,0,\ldots,0)^T$ targets changes in the effect of the first predictor while treating the remaining predictors as nuisance variables. This formulation allows us to assess the stability of a particular predictor effect, or a linear combination of effects, while including the remaining predictors in the regression. 



Formally, we consider the null hypothesis
\[
H_0:\quad \beta_j\equiv\beta_{(0)}, \qquad j=1,\ldots,n.
\]
We aim to detect alternatives under which the coefficient matrix is piecewise constant, and at least one change point induces a change in the prescribed contrast $\beta_j^T\theta$:
\[
H_a:\quad \beta_j=\beta_{(i)}, \qquad B_{i-1}<j\le B_i, \quad i=1,\ldots,s+1,
\]
where the segment-specific coefficient matrices
$\beta_{(1)}$, $\beta_{(2)}$, $\ldots$, $\beta_{(s+1)}$ satisfy
\[
\beta_{(i)}^T\theta \neq \beta_{(i+1)}^T\theta
\quad \text{for at least one } i \in \{1,\ldots, s\}.
\]

High-dimensional change-point analysis has developed along several directions, two of which are particularly relevant here. The first concerns changes in the marginal distributions of high-dimensional observations, particularly shifts in the mean vector; representative examples include \citet{shao2010self,wang2022inference,jirak2015change,wang2018high,shao2010testing,matteson2014nonparametric}. The second concerns regression models with a scalar response $(p=1)$ and a high-dimensional predictor vector ($M\asymp n$ or $M\gg n$), for which sparsity assumptions on the regression coefficients are typically imposed; representative examples include \citet{baek2024estimation,bai2023unified,lee2016lasso,rinaldo2021localizing,xu2024change}. Further developments in high-dimensional change-point detection include \citet{kirch2015detection,cho2016change,liu2020unified,zhang2022adaptive,dornemann2024linear,dornemann2024sequentialeigenvaluestatisticschangepoint,bastian2026change}. 

In this paper, we study a complementary high-dimensional regime in which the predictor dimension $M$ is fixed, whereas the response dimension $p$ is smaller than, but of the same order as, the sample size $n$. Specifically, we assume
$$
p/n \longrightarrow \gamma\in(0,1)
\qquad\text{as } p,n\to\infty.
$$
To the best of our knowledge, structural-change testing in this regime has received little attention in the existing literature. In this setting, the regression coefficient matrix can be estimated by ordinary least squares without imposing sparsity, although its aggregate estimation error need not vanish when $p/n$ converges to a nonzero limit. The principal challenge is to account for dependence among a growing number of response variables. When the response dimension is proportional to the sample size, estimation of the residual covariance affects the limiting distribution of the test statistic, making classical fixed-dimensional critical values inappropriate. Procedures developed for scalar-response regression with many predictors address a different dimensional regime and therefore do not directly resolve this difficulty.

To address this problem, we construct a least-squares-based Wald statistic that compares a prescribed coefficient contrast across two adjacent segments. We standardize the statistic using the residual covariance estimator under the null model and derive its limiting distribution that accounts for the high-dimensional response. Maximizing the normalized statistic over suitable sets of candidate segmentations yields tests for a single change point and for multiple change points at unknown locations. A discretized multiscale scan provides a computationally more efficient alternative to the full scan.

The contributions of this work are threefold. First, we develop tests for changes in prescribed regression contrasts when the response dimension grows proportionally to the sample size and the design is allowed to be general subject to mild regularity conditions. The tests are invariant under nonsingular linear transformations of the responses and require no sparsity or other structural restrictions on the positive definite error covariance matrix. Second, we introduce a flexible scan over adjacent segments of varying locations and lengths, accommodating multiple changes without specifying their number or locations in advance, subject only to a minimum segment-length constraint. Third, we establish weak convergence of the normalized statistic process over the full scanning set to a centered Gaussian process. This functional limit provides a common theoretical basis for approximating the distributions of the single-change-point, multiple-change-point, and discretized scan statistics. In particular, it yields approximated critical values at any significance level using only the design matrix. We further characterize the asymptotic power of the proposed tests under local alternatives, describing how the signal, design, and high-dimensional aspect ratio determine the limiting drift and power.

The remainder of the paper is organized as follows. Section~\ref{sec:method} introduces the proposed scan statistics for single and multiple structural changes. Section~\ref{sec:theory} establishes their limiting null distributions and studies their asymptotic power under local alternatives. Section~\ref{sec:simulation} reports simulation results on finite-sample size and power. Section~\ref{sec:real_data} presents a real data application. The Appendix contains additional numerical results and proofs of the theoretical results.

\section{Methodology}\label{sec:method}
In this section, we introduce the proposed change-point detection procedures. For $r\in[0,1]$, define 
\[k(r)=\lfloor nr\rfloor+1.\]
Let $\varepsilon>0$ be a user-specified trimming parameter satisfying $M< k(\varepsilon)<n/2$. We keep $\varepsilon$ fixed throughout the asymptotic analysis; for practical implementation, we recommend $\varepsilon= \max\{0.1,(M+1)/n\}$. 

Define the set of admissible ordered triples by
\[
\calT_{\rm mc}(\varepsilon)
=
\left\{
(t_1,t_2,t_3)\in[0,1]^3:
t_2-t_1\geq \varepsilon,\quad t_3-t_2\geq \varepsilon
\right\}.
\]
For any matrix $\bA=[A_1,\ldots,A_n]$ with $n$ columns, denote by $\bA_{t_1:t_2}$ the submatrix consisting of the columns whose indices lie in $[k(t_1),k(t_2))$. Namely,
\[ \bA_{t_1:t_2} =  \big[ A_{k(t_1)}, A_{k(t_1)+1},  \cdots, A_{k(t_2)-1}\big]. \]
Here, we require $k(t_2)>k(t_1)$ so that the segment is not empty.

Collect all response vectors and the predictor vectors into $\bY = [\bY_1, \cdots, \bY_n]$ and $\bX = [\bX_1, \cdots, \bX_n]$, respectively. For $t=(t_1,t_2,t_3)\in\calT_{\rm mc}(\varepsilon)$, we restrict the predictors to the two adjacent segments determined by $t$ and define the augmented design matrix 
\[
\calX_t=
\begin{pmatrix}
\bX_{t_1:t_2} & 0 \\
0 & \bX_{t_2:t_3}
\end{pmatrix}
\in \mathbb{R}^{2M\times (k(t_3)-k(t_1))}.
\]
The trimming condition on $\varepsilon$ guarantees that each segment contains at least $M$ observations.
To estimate the contrast between two adjacent segments, define
\begin{equation}\label{eq:def_h}
\bu_t=\calX_t^T(\calX_t\calX_t^T)^{-1}\tilde{\theta}, \quad \text{with}\quad \tilde{\theta}=(\theta^T,-\theta^T)^T\in\mathbb{R}^{2M}.
\end{equation}
We then embed $\bu_t$ into $\mathbb{R}^n$ by padding it with zeros outside the interval $[k(t_1), k(t_3))$: 
\[\bh_t = [0_{k(t_1)-1}^T, (\bu_t)^T, 0_{n-k(t_3)+1}^T]^T\in \mathbb{R}^n,\]
where, $0_m$ means a vector of zeros of length $m$.

Throughout this section, the matrix inverses appearing in the definitions are assumed to exist. Condition \ref{enum:full_rank} in Section~\ref{sec:theory} provides sufficient conditions for their existence.

\paragraph*{Motivation.}
To motivate the proposed procedures, we first consider an oracle setting with at most one change point, whose location is known to be $B_1=k(r)-1$ for some $r\in(0,1)$. Let $t=(0,r,1)$. The full model can then be written as
\[
\bY=(\beta_{(1)}^T,\beta_{(2)}^T)\calX_t+\Sigma_p^{1/2}\bZ,
\qquad
\bZ=(\bz_1,\ldots,\bz_n).
\]
Set $\mu=(\beta_{(1)}^T,\beta_{(2)}^T)\tilde{\theta}$. In this oracle setting, detecting a change in the prescribed contrast is equivalent to testing $H_0:\mu=0$
{against} $H_a:\mu\neq0$.

The ordinary least squares estimator of $\mu$ is $\hat{\mu}=\bY\bh_t$. A direct calculation gives
\[
\cov(\hat{\mu})
=
\left\{\tilde{\theta}^T(\calX_t\calX_t^T)^{-1}\tilde{\theta}\right\}\Sigma_p = \|\bh_t\|_2^2 \Sigma_p.
\]
If $\Sigma_p$ were known, whitening the estimator by its covariance would yield the Wald statistic
\[\frac{\bh_t^T \bY^T \Sigma_p^{-1} \bY \bh_t}{\|\bh_t\|_2^2}.\]
Under Gaussian errors, this statistic coincides with the likelihood ratio statistic.

When $\Sigma_p$ is unknown, a natural plug-in estimator is the residual covariance matrix under the null model,
\[
\bS=
\frac{1}{n}
\bY
\left(
I_n-\bX^T(\bX\bX^T)^{-1}\bX
\right)
\bY^T.
\]
Replacing $\Sigma_p$ in the oracle Wald statistic with $\bS$ yields the empirically normalized statistic
\begin{equation}
\label{eq:def_stat_t}
V(t)=\|\bh_t\|_2^{-2}  \bh_t^T \bY^T \bS^{-1} \bY \bh_t.
\end{equation}
Under the classical asymptotic regime in which $p$ is fixed and $n\to\infty$, $V(t)$ is asymptotically $\chi^2_p$ distributed, and the null hypothesis $\mu =0$ is rejected for sufficiently large values of $V(t)$. When $p$ diverges with $n$, however, this classical $\chi^2$ approximation is no longer valid. A high-dimensional calibration is therefore required.

\paragraph*{Proposed Methodology.} Define the normalized statistic
\begin{equation}
\label{eq:def_stat_normalized}
D(t)=\frac{\sqrt{V(t)} - \sqrt{p}}{ \sqrt{1/2} \sqrt{1 -p/n } }.
\end{equation}
Here, the square-root transform is taken to alleviate the finite-sample skewness of $V(t)$. Section~\ref{sec:theory} establishes the null asymptotic behavior of this statistic. Under mild regularity conditions, when $p/n\to\gamma \in(0,1)$, $D(t)$ is asymptotically centered with unit variance for each fixed $t\in\calT_{\rm mc}(\varepsilon)$. Consequently, the normalized statistics are asymptotically comparable across candidate segmentations specified by $t$ under the null.

\textbf{Single break scan.} 
In applications where at most one change point is expected $(s=1)$, we scan the normalized statistic $D(t)$ over all binary segmentations for which both segments contain at least an $\varepsilon$ fraction of the observations. Motivated by the classical CUSUM procedure for single-change-point detection; see, for example, \citet{csorgo1997limit}, define
\[
\calT_{\rm sc}(\varepsilon)
=
\left\{
(0,t_2,1):
\varepsilon\le t_2\le1-\varepsilon
\right\}.
\]
Note that $\calT_{\rm sc}(\varepsilon)\subset \calT_{\rm mc}(\varepsilon)$. The proposed single-change-point statistic is
\begin{equation}\label{eq:single_change_CUSUM}
\SC(\varepsilon)=\sup_{t\in\calT_{\rm sc}(\varepsilon)}D(t).
\end{equation}

\textbf{Multiple breaks scan.}
When multiple change points may be present, the single-change-point scan can lose power because a fitted segment may span more than one regime. In that case, the least squares estimator averages coefficients across regimes, and changes in $\beta_j^T\theta$ with opposite signs or directions may partially or completely cancel. To retain sensitivity without requiring prior knowledge of the number of change points, we scan adjacent local segments over multiple scales. The proposed multiple-change-point statistic is
\[
\MC(\varepsilon)=\sup_{t\in\calT_{\rm mc}(\varepsilon)}D(t).
\]

\textbf{Discretized scan.} The scanning set $\calT_{\rm mc}(\varepsilon)$ accommodates a broad range of break configurations, provided the relevant changes are separated from one another and from the endpoints by at least $\lfloor n\varepsilon\rfloor$ observations. This flexibility, however, comes at a substantial computational cost. Evaluating $\MC(\varepsilon)$ requires scanning over $O(n^3)$ candidate triples. Given the transformed data $\bS^{-1/2}\bY$, evaluating $D(t)$ for a single $t$ has computational complexity $O(pnM)$, resulting in an overall complexity of $O(pn^4M)$ for the full scan. 

To reduce the computational burden, we also consider the discretized scanning set
\begin{align*}
\calT_{\rm disc}(\varepsilon)
=\calT_{\rm mc}(\varepsilon)\cap{\calG(\varepsilon)}^3,
\qquad
\calG(\varepsilon)
=\{a\varepsilon:a=0,1,2,\ldots\}.
\end{align*}
That is, we restrict the coordinates $t_j$ to multiples of $\varepsilon$. The corresponding discretized scan statistic is
\begin{equation}
\DISC(\varepsilon)
=\sup_{t\in\calT_{\rm disc}(\varepsilon)}D(t).
\end{equation}
This restriction reduces the number of candidate triples from $O(n^3)$ to $O(\varepsilon^{-3})$.



Section~\ref{sec:theory} derives the asymptotic distributions of the proposed statistics and provides practical critical values. We conclude this section with three remarks on the proposed procedures.

\begin{remark}\label{remark:invariance}
The proposed tests are invariant under any nonsingular linear transformation of the responses; that is, they are unchanged under $\bY\mapsto A\bY$ for any nonsingular matrix $A$. Consequently, provided that $\Sigma_p$ is nonsingular, the null distribution of the test statistics does not depend on the population covariance structure. In particular, no sparsity or other structural assumptions are imposed on $\Sigma_p$. Moreover, the proposed procedures do not require sparsity of the regression coefficient matrices $\beta_j$.
\end{remark}

\begin{remark}\label{remark_effect_varepsilon}
The trimming parameter $\varepsilon$ ensures that all scanned segments contain at least $\lfloor n\varepsilon\rfloor$ observations, thereby excluding contrasts based on very small subsamples. Such trimming is standard in change-point analysis; see, for example, \citet{csorgo1997limit,aue2013structural,wang2022inference}. By construction, the proposed statistics target structural changes in the prescribed contrast direction $\theta$ that are at least $\lfloor n\varepsilon\rfloor$ observations away from one another and from the endpoints. Larger values of $\varepsilon$ generally yield more stable inference but may reduce power for changes near the boundaries or for closely spaced change points. Conversely, smaller values permit detection over shorter time scales at the cost of greater scanning multiplicity, larger critical values, and a higher computational burden. From a theoretical perspective, keeping $\varepsilon$ bounded away from zero ensures the asymptotic normality of $D(t)$ on $\calT_{\rm mc}(\varepsilon)$ and the required asymptotic equicontinuity. The moving-domain regime in which $\varepsilon=\varepsilon_n\downarrow0$ is not covered by the present theory.
\end{remark}

\begin{remark}\label{remark:usage_of_residual_cov}
The statistic $V(t)$ standardizes the ordinary least squares estimator using the residual covariance estimator $\bS$ computed under the null model. A natural alternative is the residual covariance estimator under the two-segment model,
$n^{-1} \bY_t
\left(
I-\calX_t^T(\calX_t\calX_t^T)^{-1}\calX_t
\right)
\bY^T_t$. Under the null hypothesis, the null-model estimator uses more residual degrees of freedom and is based on the full sample, which can be advantageous when the segment length $|t_3-t_1|$ is relatively small. By contrast, under the alternative, the two-segment estimator may better reflect the covariance structure after accounting for a change in the regression coefficients. We adopt $\bS$ primarily because it leads to more stable behavior under the null, a substantially simpler asymptotic theory, and, consequently, a more transparent methodology.
\end{remark}

\section{Asymptotic Theory} \label{sec:theory}
In this section, we establish the asymptotic properties of the proposed tests. We first derive their limiting null distributions and obtain practical approximations to the critical values. We then characterize their asymptotic power under local alternatives.

Throughout, we impose the following conditions.
\begin{enumerate}[label={\bf C\arabic*}]
\item\label{enum:moment_conditions} \emph{Moment condition.} The collection of entries $\{z_{ji}:1\leq j\leq n,\ 1\leq i\leq p\}$ consists of independent and identically distributed random variables with mean zero, variance one, and finite moments of all orders.

\item\label{enum:high_dimensional_regime} \emph{High-dimensional regime.} While the predictor dimension $M$ and number of breaks $s$ remain fixed,  we assume that $p,n\to\infty$ with $p/n\to\gamma\in(0,1)$. 
If $s>0$, we assume that ${B_i}/{n}\to \tilde{B}_i$, $i=1,\ldots,s$, for constants $0<\tilde{B}_1<\tilde{B}_2<\cdots<\tilde{B}_s<1$.

\item\label{enum:full_rank} \emph{Segmentwise regular design.} Assume that
\[\limsup_{n}\|\bX\|_{\max} = \limsup_n \max_{i,j} |X_{ij}| <\infty.\] 
In addition, there exist constants
$c_\varepsilon>0$ and $N_\varepsilon>0$ such that, whenever $n\ge N_\varepsilon$, 
\[\inf_{t\in\calT_{\rm mc}(\varepsilon)} \lambda_{\min}\Big( n^{-1} \calX_t \calX_t^T\Big)\geq c_\varepsilon\qquad\text{and}\qquad \inf_{t\in\calT_{\rm mc}(\varepsilon)}\tilde{\theta}^T (n^{-1}\mathcal X_t\mathcal X_t^T)^{-1} \tilde{\theta}\ge c_\varepsilon. \]
Here, $\lambda_{\min}(\cdot)$ denotes the smallest eigenvalue of a matrix. 

\item\label{enum:stable_cov_kernel} \emph{Limiting correlation kernel.} There exists a positive-semidefinite kernel $\calK:\calT_{\rm mc}(\varepsilon)\times \calT_{\rm mc}(\varepsilon)\to \mathbb{R}$ such that
\[\sup_{t,t' \in \calT_{\rm mc}(\varepsilon)} \left| \frac{|\langle\bh_t,\bh_{t'}\rangle|^2}{\|\bh_t\|_2^2\,\|\bh_{t'}\|_2^2} - \calK(t,t')\right| \longrightarrow 0,  \quad \text{as }n\to\infty.\]
Moreover, there exist constants $C>0$ and $a>0$ such that
\[1-\calK(t,t')\le C\|t-t'\|_2^{a},\qquad t,t'\in\calT_{\rm mc}(\varepsilon).\]
\end{enumerate}

Condition~\ref{enum:moment_conditions} specifies the moment requirements on the errors, while Condition~\ref{enum:high_dimensional_regime} requires the response dimension $p$ and sample size $n$ to grow proportionally. Under the alternative, the normalized break points converge to distinct interior limits and are therefore asymptotically separated from one another and from the endpoints. Condition~\ref{enum:full_rank} ensures that the segmentwise least-squares estimators are well defined and that the prescribed contrast $\beta_j^T\theta$ remains estimable over every admissible segment. Finally, Condition~\ref{enum:stable_cov_kernel} specifies the limiting dependence structure of the normalized scan statistics. Its H\"{o}lder-type increment bound ensures that the limiting Gaussian process admits a continuous modification, as required for the supremum limits in Theorem~\ref{thm:main_null}.

\subsection{Asymptotic Null Distribution}\label{subsec:asymptotic_null}

We derive the asymptotic distributions of the proposed statistics under the null hypothesis $H_0$. Unless stated otherwise, $\rightsquigarrow$ denotes weak convergence in $\ell^\infty(\calT_{\rm mc}(\varepsilon))$, whereas $\stackrel{D}{\longrightarrow}$ denotes weak convergence of random variables.

Let $\{G(t):t\in\calT_{\rm mc}(\varepsilon)\}$ be a centered Gaussian process with covariance kernel 
\[\mE G(t)G(t')  = \calK(t, t'), \quad t,t' \in \calT_{\rm mc}(\varepsilon).\]
By Condition~\ref{enum:stable_cov_kernel} and Kolmogorov's continuity theorem, this process admits an almost surely continuous modification.

\begin{theorem}\label{thm:main_null}
Suppose that Conditions~\ref{enum:moment_conditions}--\ref{enum:stable_cov_kernel}
hold. Under the null hypothesis $H_0$, for every fixed trimming parameter
$\varepsilon\in(0, 0.5)$, as $n\to\infty$,
\[\bigl\{D(t):t\in\calT_{\rm mc}(\varepsilon)\bigr\}\rightsquigarrow\bigl\{G(t):t\in\calT_{\rm mc}(\varepsilon)\bigr\}.\]
Consequently, by the continuous mapping theorem,
\begin{align*}
\MC(\varepsilon)\stackrel{D}{\longrightarrow} \sup_{t\in \calT_{\rm mc}(\varepsilon)} G(t),   \quad
\SC(\varepsilon)\stackrel{D}{\longrightarrow} \sup_{t\in \calT_{\rm sc}(\varepsilon)} G(t), \quad
\DISC(\varepsilon)\stackrel{D}{\longrightarrow} \sup_{t\in \calT_{\rm disc}(\varepsilon)} G(t).
\end{align*}
\end{theorem}

\begin{remark}\label{remark:moment_condition}
Theorem~\ref{thm:main_null} establishes weak convergence of the standardized process $D(t)$ over $\calT_{\rm mc}(\varepsilon)$ under the finite-moments-of-all-orders assumption. If weak convergence is required only over the single-change-point scanning set $\calT_{\rm sc}(\varepsilon)$, this condition can be relaxed to the existence of a finite eighth moment. If convergence is required only over the discretized scanning set $\calT_{\rm disc}(\varepsilon)$, it can be further relaxed to the existence of a finite fourth moment.
\end{remark}

In practice, critical values at any prescribed significance level can be approximated as follows. For a scanning set $\calA =\calT_{\rm mc}(\varepsilon)$ or $\calT_{\rm sc}(\varepsilon)$, choose a sufficiently fine grid $\{\tau_1,\tau_2,\ldots,\tau_Q\}\subset\calA$ and form the $Q\times Q$ matrix $\widehat{\calK}$ whose $(i,j)$th entry is
$$
\widehat{\mathcal K}(\tau_i,\tau_j)
\coloneqq
\frac{\left|\langle\bh_{\tau_i},\bh_{\tau_j}\rangle\right|^2}
{\|\bh_{\tau_i}\|_2^2\,\|\bh_{\tau_j}\|_2^2},
\qquad
i,j=1,\ldots,Q.
$$
For $\calA=\calT_{\rm disc}(\varepsilon)$, set the grid equal to the entire scanning set.

The limiting distribution of the process $\{D(t),\, t\in \calA\}$ is approximated by simulating a centered Gaussian random vector with covariance matrix $\widehat{\calK}$. Correspondingly, the distribution of $\SC(\varepsilon)$, $\MC(\varepsilon)$, or $\DISC(\varepsilon)$ is approximated by the maxima of the Gaussian vector for the three choices of $\calA$. Denote the resulting $(1-\alpha)$-quantiles by $\xi_{\rm sc}(1-\alpha)$, $\xi_{\rm mc}(1-\alpha)$, and $\xi_{\rm disc}(1-\alpha)$, respectively. At asymptotic significance level $\alpha$, the corresponding test rejects the null hypothesis whenever
\[
\SC(\varepsilon)>\xi_{\rm sc}(1-\alpha),\qquad
\MC(\varepsilon)>\xi_{\rm mc}(1-\alpha),\qquad
\DISC(\varepsilon)>\xi_{\rm disc}(1-\alpha),
\]
respectively.

\subsection{Asymptotic Power Under Local Alternatives}
\label{subsec:asymptotic_alternative}

We next study the asymptotic power of the proposed tests under local alternatives. To provide intuition, we first consider the single-change-point case $(s=1)$ and then extend the analysis to an arbitrary but fixed number of change points.

\subsubsection{Power Analysis for the Single-Change-Point Case}
\label{subsec:power_single}
Suppose first that there is a single change point ($s=1$). Under Condition~\ref{enum:high_dimensional_regime}, its normalized location satisfies $B_1/n\to \tilde{B}_1$ for some $\tilde{B}_1\in(0,1)$. Let $\widetilde{\mathcal X}$ denote the design matrix corresponding to the limiting true segmentation. In this case,
\[ \widetilde{\calX} = \calX_t \in \mathbb{R}^{2M\times n} , \qquad t=(0,\tilde{B}_1,1).\]
Recall that $\beta_{(1)}$ and $\beta_{(2)}$ denote the regression coefficient matrices before and after the change point, respectively.

Define the power of the single-change-point test based on $\SC(\varepsilon)$ at significance level $\alpha$ by
\[\calP_{\rm sc} (\varepsilon, \alpha) = \mP\left( \SC(\varepsilon) > \tilde{\xi}_{\rm sc}(1-\alpha)\, \middle|\, H_a\right), \]
where $\tilde{\xi}_{\rm sc}(1-\alpha)$ is the $(1-\alpha)$-quantile of $\sup_{t\in \calT_{\rm sc}(\varepsilon)} G(t)$.

Define the drift process under the alternative by
\begin{equation}\label{eq:def_U}
\calS_n(t) =  \frac{\left\|n^{-1/2}
\Sigma_p^{-1/2}
(\beta_{(1)}^T,\beta_{(2)}^T)
\widetilde{\mathcal X}\bh_t
\right\|_2^2
}{
\|\bh_t\|_2^2
},
\qquad
t\in\mathcal T_{\rm sc}(\varepsilon).
\end{equation}

\begin{remark}\label{remark:change_in_mean}
To illustrate $\calS_n(t)$, consider the classical change-in-mean setting. In this case, $M=1$, $\mE \bY_j = \beta_{(1)}^T$ for $j \leq B_1$, and $\mE \bY_j = \beta_{(2)}^T$ for $j > B_1$.
The corresponding true design matrix is
\[
\widetilde{\mathcal X}
=
\left(
\begin{array}{c|c}
\underbrace{\begin{matrix}
1 & \cdots & 1\\
0 & \cdots & 0
\end{matrix}}_{B_1}
&
\underbrace{\begin{matrix}
0 & \cdots & 0\\
1 & \cdots & 1
\end{matrix}}_{\,n-B_1\,}
\end{array}
\right).
\]
The mean shift of interest is $\beta_{(2)}-\beta_{(1)}$, corresponding to $\theta=1$. For $t=(0,t_2,1)$, direct calculation gives
\[
\calS_n(t)
= (\eta_{\rm sc}(t)+o(1))\,
\big\|\Sigma_p^{-1/2}(\beta_{(1)}^T-\beta_{(2)}^T)
\big\|_2^2,\quad \text{where}
\]
\[
\eta_{\rm sc}(t)
=
\begin{cases}
\dfrac{t_2}{1-t_2}(1-\tilde{B}_1)^2,
&
t_2\le \tilde{B}_1,\\[10pt]
\dfrac{1-t_2}{t_2}\tilde{B}_1^2,
&
t_2>\tilde{B}_1.
\end{cases}
\]
Thus, $\calS_n(t)$ is asymptotically proportional to the squared Mahalanobis distance of the mean shift $\beta_{(2)}-\beta_{(1)}$. The factor
$\eta_{\rm sc}(t)$ quantifies the signal attenuation when the candidate split $t_2$ differs from the limiting true break location $\tilde{B}_1$. In particular, $\eta_{\rm sc}(t)$ is uniquely maximized at $t_2=\tilde{B}_1$, recovering the familiar population drift of a classical CUSUM scan; see, for example, \citet{csorgo1997limit} and \citet{li2026adaptable}. Equation~\eqref{eq:def_U} extends this drift structure to regression models with a general design matrix.
\end{remark}

\begin{theorem}\label{thm:power_single}
Suppose that Conditions~\ref{enum:moment_conditions}--\ref{enum:stable_cov_kernel} hold and $s=1$. Assume that there exist constants $\overline{K}>0$ and $N\in\mathbb N$ such that, for all $n\ge N$,
\begin{equation}
\label{eq:local_alternative_sc}
\sqrt{n} \left\| \Sigma_p^{-1/2} (\beta_{(2)}^T-\beta_{(1)}^T) \right\|_{2}^{2}\le \overline{K}.
\end{equation}
Then, for every fixed $\varepsilon\in(0,1/2)$ and $\alpha\in(0,1)$, as $n\to\infty$,
\begin{align*}
&\calP_{\rm sc}(\varepsilon,\alpha)
-
\mP \bigg(
\sup_{t\in\calT_{\rm sc}(\varepsilon)}
\bigg[
G(t)
 + \sqrt{\frac{n(1-\gamma_n)}{2\gamma_n} }  \calS_n(t)
\bigg]>
\tilde{\xi}_{\rm sc}(1-\alpha)
\bigg)
\longrightarrow0,
\end{align*}
where $\gamma_n=p/n$.
\end{theorem}

\begin{remark}\label{remark:nontrivial_power}
Theorem~\ref{thm:power_single} indicates that the proposed test $\SC(\varepsilon)$ has nontrivial power when the maximum scaled drift
$\sup_{t\in\calT_{\rm sc}(\varepsilon)} \sqrt{n}\calS_n(t)$ is sufficiently large. The magnitude of this drift depends on the interaction among the true coefficients $(\beta_{(1)}^T,\beta_{(2)}^T)$, the design matrix $\widetilde{\calX}$, and the prescribed contrast $\theta$. In particular, if $\tilde{B}_1\in(\varepsilon,1-\varepsilon)$, then
\begin{equation*}
\begin{split}
\sup_{t\in\calT_{\rm sc}(\varepsilon)}
\sqrt{n}\calS_n(t)
&\ge
\sqrt{n}\calS_n\bigl((0,\tilde{B}_1,1)\bigr) =
\frac{
\sqrt{n}
\big\|
\Sigma_p^{-1/2}
(\beta_{(1)}^T-\beta_{(2)}^T)\theta
\big\|_2^2
}{
\tilde{\theta}^T
\bigl(n^{-1}\widetilde{\calX}\widetilde{\calX}^T\bigr)^{-1}
\tilde{\theta}
}\\
&\ge
c_\varepsilon\|\tilde{\theta}\|_2^{-2}
\sqrt{n}
\big\|
\Sigma_p^{-1/2}
(\beta_{(1)}^T-\beta_{(2)}^T)\theta
\big\|_2^2.
\end{split}
\end{equation*}
Thus, under a regular design and provided that the true break is at least $\lfloor n\varepsilon\rfloor$ observations away from the endpoints, the proposed test $\SC(\varepsilon)$ attains nontrivial power whenever the scaled squared Mahalanobis norm of the contrast shift,
\[\sqrt{n}\big\|\Sigma_p^{-1/2}(\beta_{(1)}^T-\beta_{(2)}^T)\theta\big\|_2^2,\]
is sufficiently large. 
\end{remark}

\subsubsection{Power Analysis for the Multiple-Change-Point Case}
\label{subsec:power_multiple}

We now allow an arbitrary but fixed number of break points. For conciseness, we focus on the test based on
$\MC(\varepsilon)$. The same arguments apply to the discretized test $\DISC(\varepsilon)$ after replacing the scan set $\calT_{\rm mc}(\varepsilon)$ with $\calT_{\rm disc}(\varepsilon)$.
Let
\[
\widetilde{\calX}=
\operatorname{BlockDiag}
\left(
\bX_{0:\tilde{B}_1},
\bX_{\tilde{B}_1:\tilde{B}_2},
\ldots,
\bX_{\tilde{B}_s:1}
\right)
\]
denote the augmented design matrix for the limiting true segmentation.  This construction extends the definition of $\widetilde{\calX}$ in the single-change-point setting in Section~\ref{subsec:power_single}.

Define the power of the test based on $\MC(\varepsilon)$ at significance level $\alpha$ by
\[ \calP_{\rm mc}(\varepsilon,\alpha) =\mP\left( \MC(\varepsilon)> \tilde{\xi}_{\rm mc}(1-\alpha) \,\middle|\, H_a \right),\]
where $\tilde{\xi}_{\rm mc}(1-\alpha)$ is the $(1-\alpha)$-quantile of $\sup_{t\in \calT_{\rm mc}(\varepsilon)} G(t)$.

The drift process $\calS_n(t)$ extends to this setting as follows:
\begin{equation}\label{eq:def_U_multiple}
\calS_n(t) =  \frac{\left\|n^{-1/2}
\Sigma_p^{-1/2} (\beta_{(1)}^T,\cdots,\beta_{(s+1)}^T) \widetilde{\calX} \bh_t \right\|_2^2 }{ \|\bh_t\|_2^2},
\qquad
t\in\calT_{\rm mc}(\varepsilon).
\end{equation}

\begin{theorem}\label{thm:power_multiple}
Suppose that Conditions~\ref{enum:moment_conditions}--\ref{enum:stable_cov_kernel} hold. Assume that there exist constants $\overline{K}>0$ and $N\in\mathbb N$ such that, for all $n\ge N$,
\begin{equation}
\label{eq:local_alternative_mc}
\sqrt{n} \left\| \Sigma_p^{-1/2}
\bigl(\beta_{(2)}^T-\beta_{(1)}^T,\ldots,\beta_{(s+1)}^T-\beta_{(1)}^T\bigr)
\right\|_{2}^{2}\le \overline{K}.
\end{equation}
Then, for every fixed $\varepsilon\in(0,1/2)$ and $\alpha\in(0,1)$, as $n\to\infty$,
\begin{align*}
&\calP_{\rm mc}(\varepsilon,\alpha)
-
\mP \bigg(
\sup_{t\in\calT_{\rm mc}(\varepsilon)}
\bigg[
G(t)
 + \sqrt{\frac{n(1-\gamma_n)}{2\gamma_n}}\calS_n(t)
\bigg]>\tilde{\xi}_{\rm mc}(1-\alpha)
\bigg)
\to0,
\end{align*}
where $\gamma_n=p/n$.
\end{theorem}


The drift depends on both the coefficient changes and their configuration
relative to the candidate segments. If a change point $\tilde{B}_i$ has an unchanged
regime of length greater than $\varepsilon$ on each side, the full scan
contains an adjacent-segment comparison, say $[\tilde{B}_i - \varepsilon, \tilde{B}_i)$ and $[\tilde{B}_i, \tilde{B}_i+\varepsilon)$, that isolates this change.
At such a comparison, the population contrast equals the difference
between the two neighboring coefficient contrasts, yielding a lower
bound analogous to Remark~\ref{remark:nontrivial_power}. Comparisons
that span several regimes may instead attenuate the signal through
cancellation. The discretized scan has the same power representation,
but its drift depends additionally on the alignment of the grid with
the change points.

\subsection{Further Discussion on Change-Point Localization}\label{subsec:further_discussion}
After a structural change has been detected, a natural subsequent task is to estimate its location. In the single-change-point setting, one might consider
\[
\widehat{\tau}_{\rm scan}
\in\mathop{\rm arg\,max}_{r\in[\varepsilon,1-\varepsilon]}D((0,r,1)).
\]

The change-in-mean setting of Remark~\ref{remark:change_in_mean} provides a useful benchmark. In this special case, $\eta_{\rm sc}(t)$, and hence the population drift, is uniquely maximized at the true break fraction. If the data-generating sequence satisfies $\sqrt{n}\big\|(\beta_{(2)}-\beta_{(1)})\Sigma_p^{-1/2}\big\|_2^2\longrightarrow\infty$, the deterministic drift dominates the stochastic fluctuations. Subject to a suitable uniform extension of the stochastic expansion beyond the local regime, the argmax estimator can therefore consistently recover the true break fraction.

This conclusion does not extend to a general design. Even when the signal is sufficiently strong for the stochastic fluctuations to be negligible, the population drift $\calS_n(t)$ need not be maximized at the true break fraction. It is because that in the definition of $V(t)$ in \eqref{eq:def_stat_t}, both the design-weighted contrast $\bY\bh_t$ and the normalization by $\|\bh_t\|_2^2$ depend on the candidate split. Their joint effect can make the standardized population signal at a misspecified split larger than that at the true split. Appendix \ref{sec:counterexample_maxima_location} provides an explicit counterexample. Thus, the proposed scan statistic does not furnish a generally valid localization criterion.

For general designs, a separate localization criterion is therefore needed. Quasi-maximum-likelihood methods are commonly used to estimate break locations by comparing model fit across candidate segmentations; see, for example, \citet{qu2007estimating}. Extending this approach to the present high-dimensional regime requires substantial additional analysis and is beyond the scope of this paper.

\section{Simulation Studies}\label{sec:simulation}

In this section, we assess the proposed tests through Monte Carlo simulations. The study has two objectives. First, we validate the empirical null distributions and power of the proposed methods under general regression designs. Second, we consider the special change-in-mean setting described in Remark~\ref{remark:change_in_mean}, which allows comparison with existing methods specifically designed for this problem.

\subsection{Validation of the Proposed Methods under General Designs}
Because the procedures are invariant to $\Sigma_p$, we set $\Sigma_p=I_p$. We also set $M=3$, $\varepsilon=0.1$, and $\theta=(0,1,0)^T$. We vary the following factors.
\begin{enumerate}
    \item[(a)] We consider sample sizes and dimensions as $n\in\{200,400\}$ and $p/n\in\{0.2,0.3,0.4\}$.
    \item[(b)] We generate the error entries $z_{ji}$ from: (i) $N(0,1)$; (ii) $t(6)$; and (iii) $\operatorname{Poisson}(10)$. In each case, the variables are standardized to have mean zero and variance one. Note that the $t(6)$ setting violates Condition \ref{enum:moment_conditions} and therefore should be viewed as a robustness experiment with heavy tail outside the scope of the stated theory.
    \item[(c)] We consider two designs for the design matrix $\bX$.
    \begin{enumerate}
        \item[(i)] Unstructured design: The entries of $\bX$ are generated independently from $N(0,1)$.

        \item[(ii)] Structured design: The first row of $\bX$ consists of ones and represents the intercept, the second row follows the linear trend $X_{2j}\propto (j/n-1/2)$, and the third row follows a sinusoidal pattern with period $\lfloor n/5\rfloor$.
    \end{enumerate}
    For both designs, each row of $\bX$ is normalized to have Euclidean norm $\sqrt{n}$. For each combination of $n$ and design type, the same realized design matrix is used in all Monte Carlo replications.
\end{enumerate}

\begin{table}[htb]
\centering
\scriptsize
\setlength{\tabcolsep}{2.8pt}
\renewcommand{\arraystretch}{1.08}
\caption{Empirical sizes of $\SC$, $\MC$, and $\DISC$ at nominal level $\alpha=0.05$ and the trimming parameter $\varepsilon=0.1$ or $0.05$.}
\label{tab:empirical_size}
\resizebox{\linewidth}{!}{%
\begin{tabular}{@{}clcccccc@{\quad}clcccccc@{}}
\toprule
\multicolumn{8}{c}{Unstructured design}
&
\multicolumn{8}{c}{Structured design}
\\
\cmidrule(lr){1-8}\cmidrule(l){9-16}
& &
\multicolumn{3}{c}{$n=200$}
&
\multicolumn{3}{c}{$n=400$}
&
& &
\multicolumn{3}{c}{$n=200$}
&
\multicolumn{3}{c}{$n=400$}
\\
\cmidrule(lr){3-5}\cmidrule(lr){6-8}
\cmidrule(lr){11-13}\cmidrule(l){14-16}
Error & Scan
& $.2$ & $.3$ & $.4$
& $.2$ & $.3$ & $.4$
& Error & Scan
& $.2$ & $.3$ & $.4$
& $.2$ & $.3$ & $.4$
\\
\midrule

\multicolumn{16}{c}{{$\varepsilon=0.1$}}\\
\addlinespace[2pt]

\multirow{3}{*}{$N(0,1)$}
& $\SC$   & 0.045 & 0.047 & 0.043 & 0.048 & 0.045 & 0.041 & \multirow{3}{*}{$N(0,1)$} & $\SC$   & 0.050 & 0.049 & 0.042 & 0.046 & 0.047 & 0.045 \\
& $\DISC$ & 0.047 & 0.047 & 0.036 & 0.046 & 0.044 & 0.043 & & $\DISC$ & 0.050 & 0.046 & 0.037 & 0.049 & 0.047 & 0.041 \\
& $\MC$   & 0.041 & 0.034 & 0.024 & 0.045 & 0.040 & 0.032& & $\MC$   & 0.043 & 0.030 & 0.018 & 0.043 & 0.035 & 0.028 \\
\cmidrule(lr){1-8}\cmidrule(l){9-16}

\multirow{3}{*}{Poisson}
& $\SC$   & 0.048 & 0.043 & 0.042 & 0.044 & 0.045 & 0.041 & \multirow{3}{*}{Poisson} & $\SC$   & 0.052 & 0.045 & 0.043 & 0.048 & 0.045 & 0.042 \\
& $\DISC$ & 0.046 & 0.045 & 0.040 & 0.047 & 0.045 & 0.040 &                          & $\DISC$ & 0.050 & 0.046 & 0.038 & 0.051 & 0.048 & 0.042 \\
& $\MC$   & 0.042 & 0.033 & 0.026 & 0.043 & 0.039 & 0.034 &                          & $\MC$   & 0.044 & 0.031 & 0.022 & 0.046 & 0.035 & 0.026 \\
\cmidrule(lr){1-8}\cmidrule(l){9-16}

\multirow{3}{*}{$t(6)$} 
& $\SC$   & 0.049 & 0.046 & 0.040 & 0.048 & 0.046 & 0.047 & \multirow{3}{*}{$t(6)$}  & $\SC$   & 0.060 & 0.059 & 0.047 & 0.051 & 0.051 & 0.044 \\
& $\DISC$ & 0.056 & 0.051 & 0.041 & 0.049 & 0.048 & 0.046 &                          & $\DISC$ & 0.061 & 0.054 & 0.044 & 0.058 & 0.052 & 0.044 \\
& $\MC$   & 0.055 & 0.043 & 0.029 & 0.054 & 0.047 & 0.041 &                          & $\MC$   & 0.111 & 0.065 & 0.036 & 0.072 & 0.048 & 0.036 \\
\midrule

\multicolumn{16}{c}{{$\varepsilon=0.05$}}\\ \addlinespace[2pt]
 \multirow{3}{*}{$N(0,1)$}
& $\SC$   & 0.047 & 0.047 & 0.045 & 0.049 & 0.047 & 0.044 & \multirow{3}{*}{$N(0,1)$} & $\SC$   & 0.047 & 0.044 & 0.037 & 0.045 & 0.045 & 0.042 \\
& $\DISC$ & 0.044 & 0.038 & 0.028 & 0.044 & 0.036 & 0.036 &                           & $\DISC$ & 0.045 & 0.039 & 0.024 & 0.045 & 0.042 & 0.032 \\
& $\MC$   & 0.038 & 0.025 & 0.019 & 0.044 & 0.034 & 0.028 &                           & $\MC$   & 0.035 & 0.026 & 0.013 & 0.041 & 0.032 & 0.021 \\
\cmidrule(lr){1-8}\cmidrule(l){9-16}

\multirow{3}{*}{Poisson}
& $\SC$   & 0.049 & 0.047 & 0.043 & 0.049 & 0.048 & 0.043 & \multirow{3}{*}{Poisson}  & $\SC$   & 0.050 & 0.042 & 0.036 & 0.045 & 0.045 & 0.042 \\
& $\DISC$ & 0.043 & 0.034 & 0.028 & 0.047 & 0.040 & 0.034 &                           & $\DISC$ & 0.048 & 0.039 & 0.029 & 0.048 & 0.045 & 0.034 \\
& $\MC$   & 0.037 & 0.025 & 0.018 & 0.044 & 0.037 & 0.029 &                           & $\MC$   & 0.038 & 0.023 & 0.015 & 0.042 & 0.034 & 0.021 \\

\cmidrule(lr){1-8}\cmidrule(l){9-16}

\multirow{3}{*}{$t(6)$}
& $\SC$   & 0.058 & 0.053 & 0.045 & 0.054 & 0.050 & 0.051 & \multirow{3}{*}{$t(6)$}   & $\SC$   & 0.092 & 0.072 & 0.056 & 0.074 & 0.060 & 0.051 \\
& $\DISC$ & 0.078 & 0.058 & 0.040 & 0.060 & 0.050 & 0.044 &                           & $\DISC$ & 0.087 & 0.063 & 0.042 & 0.061 & 0.050 & 0.040 \\
& $\MC$   & 0.093 & 0.062 & 0.034 & 0.066 & 0.053 & 0.040 &                           & $\MC$   & 0.139 & 0.074 & 0.036 & 0.102 & 0.061 & 0.036 \\
\bottomrule
\end{tabular}
}
\end{table}
\subsubsection{Empirical Null Distribution Under General Designs}\label{subsec:null_validation}
\begin{figure}[htb]
\centering
\includegraphics[width=0.8\linewidth,height =0.4\linewidth]{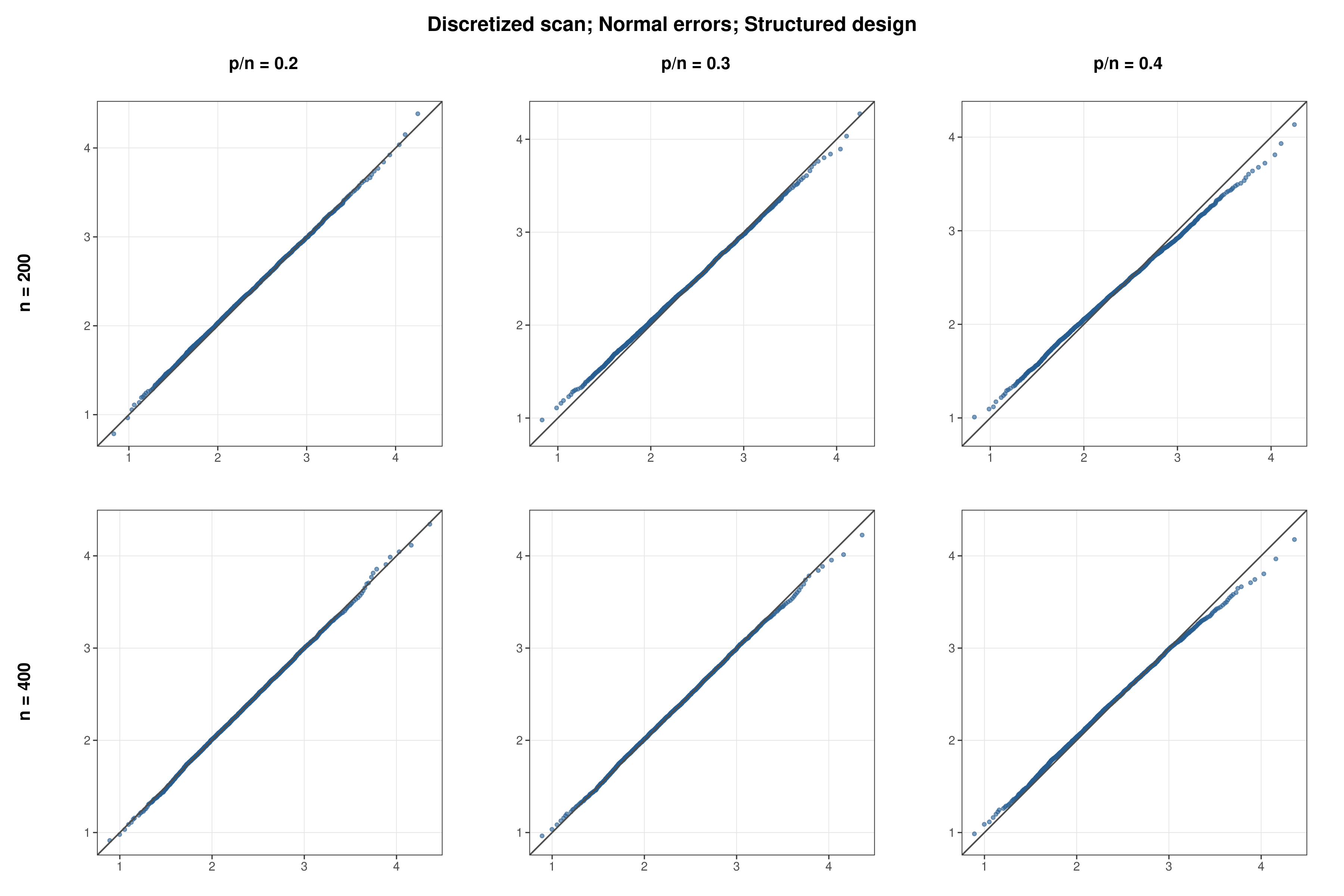}
\caption{Q--Q plot of $\DISC$ under the structured design with normal errors and $\varepsilon=0.1$.}
\label{fig:qqplot_example}
\end{figure}

Under the null hypothesis, the coefficient matrix $\beta_j$ is constant over time, with all entries set to one. For each sample size, design, and scan, the critical value is computed from 10,000 simulations of the corresponding limiting Gaussian-process maximum using the proposed approximation. Table~\ref{tab:empirical_size} reports empirical rejection probabilities at the nominal $5\%$ level based on 10,000 null replications. As a representative illustration, Figure~\ref{fig:qqplot_example} presents the Q--Q plot of $\DISC$ under the structured design with normal errors and $\varepsilon=0.1$.

We summarize the main findings as follows. For $\SC$ and $\DISC$, when $\varepsilon=0.1$, the empirical sizes are generally well controlled, mostly falling between approximately $4\%$ and $6\%$ across the settings considered. This supports the accuracy of the proposed asymptotic calibration in finite samples. When $\varepsilon=0.05$, some size inflation appears under the $t(6)$ errors, particularly when $p/n$ is relatively small, whereas the tests tend to become slightly conservative when $p/n=0.4$. With the smaller trimming parameter, the scanning sets include segments containing only $\lfloor n\varepsilon\rfloor=10$ or $20$ observations for $n=200$ and $400$, respectively. The resulting finite-sample distribution of $D(t)$ can deviate more substantially from its asymptotic Gaussian approximation on such short segments, especially under heavy-tailed errors. This provides a plausible explanation for the larger size distortions observed in the $t(6)$ setting.

For $\MC$, the empirical sizes become increasingly conservative as $p/n$ increases when $\varepsilon=0.1$, while noticeable size inflation occurs under $t(6)$ errors in some settings. Several factors may contribute to these deviations. First, the multiple-change-point scan includes a much larger collection of candidate segments, including relatively short segments for which the finite-sample approximation of $D(t)$ can be less accurate. Second, because evaluating the limiting Gaussian process over the entire scanning set is computationally prohibitive, its maximum is approximated using a Gaussian vector on a reduced grid, introducing an additional discretization error. Third, the $t(6)$ distribution violates the moment condition imposed in Condition~\ref{enum:moment_conditions}, which can further deteriorate the finite-sample approximation. When $\varepsilon=0.05$, the scan includes even shorter segments and a larger effective scanning set, and these discrepancies become more pronounced, particularly under the heavy-tailed $t(6)$ errors.

Overall, the results support the accuracy of the proposed limiting approximation for the critical values across a broad range of settings. In practice, we recommend avoiding excessively short scanning segments; for example, requiring $\lfloor n\varepsilon\rfloor\gtrsim 30$ appears to provide more reliable finite-sample calibration in our simulations. The results also indicate that pronounced heavy tails can lead to noticeable size distortions and therefore warrant additional caution. In such settings, the proposed procedures may be combined with appropriate transformations or preprocessing of the responses to mitigate the effect of heavy tails.


\subsubsection{Empirical Power Under General Designs}\label{subsec:empirical_power}
\begin{figure}[htb]
    \centering
    \includegraphics[width=0.49\linewidth,height=0.45\linewidth]{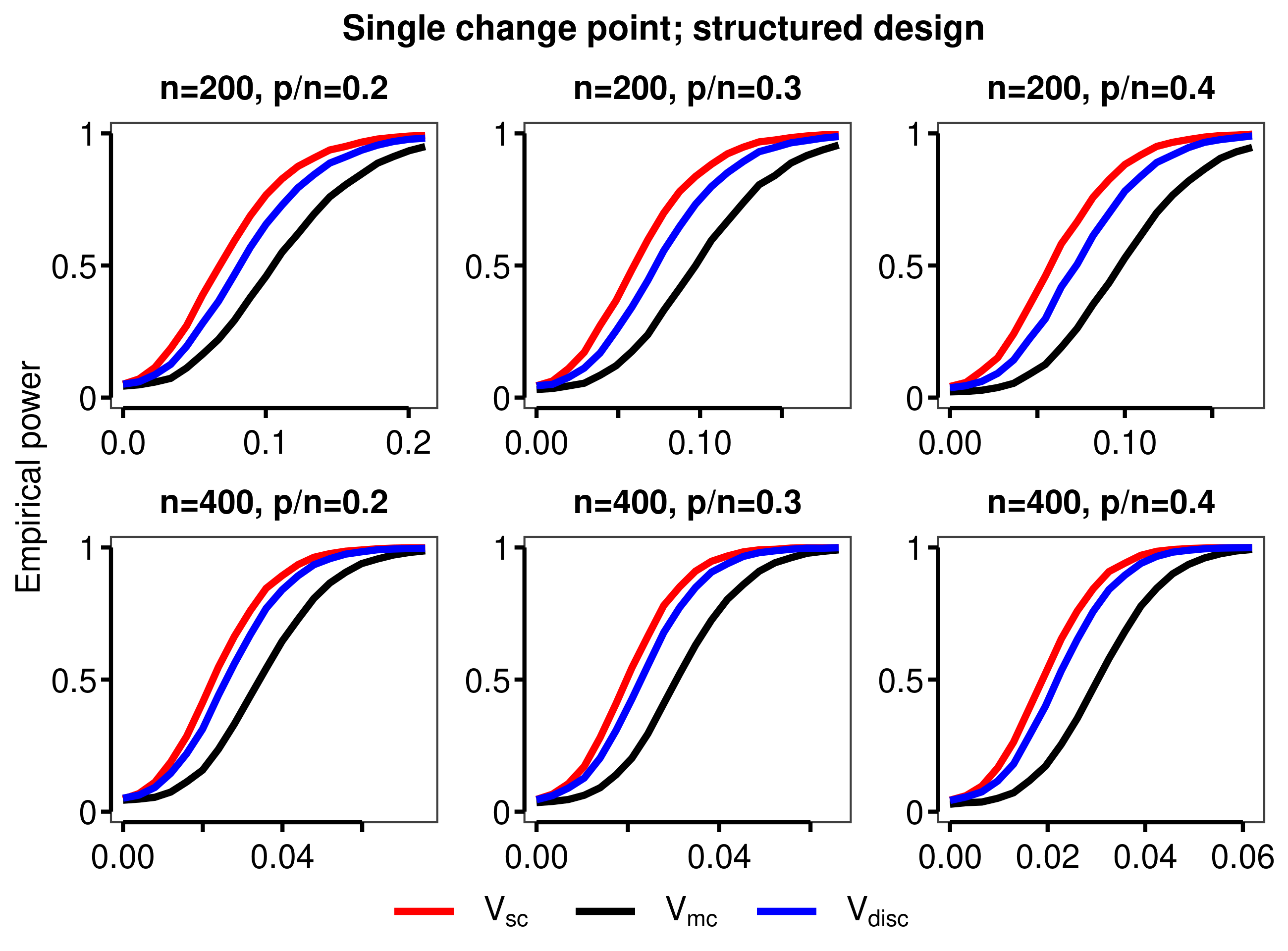}
    \includegraphics[width=0.49\linewidth,height=0.45\linewidth]{figures/eps010_power_single_normal_structured.png}
    \caption{Empirical power under single-change alternatives at $\varepsilon=0.1$ under the structured design setting and normal errors.}
    \label{fig:power_structured_normal_single_eps010}
\end{figure}

\begin{figure}[htb]
    \centering
    \includegraphics[width=0.49\linewidth,height=0.45\linewidth]{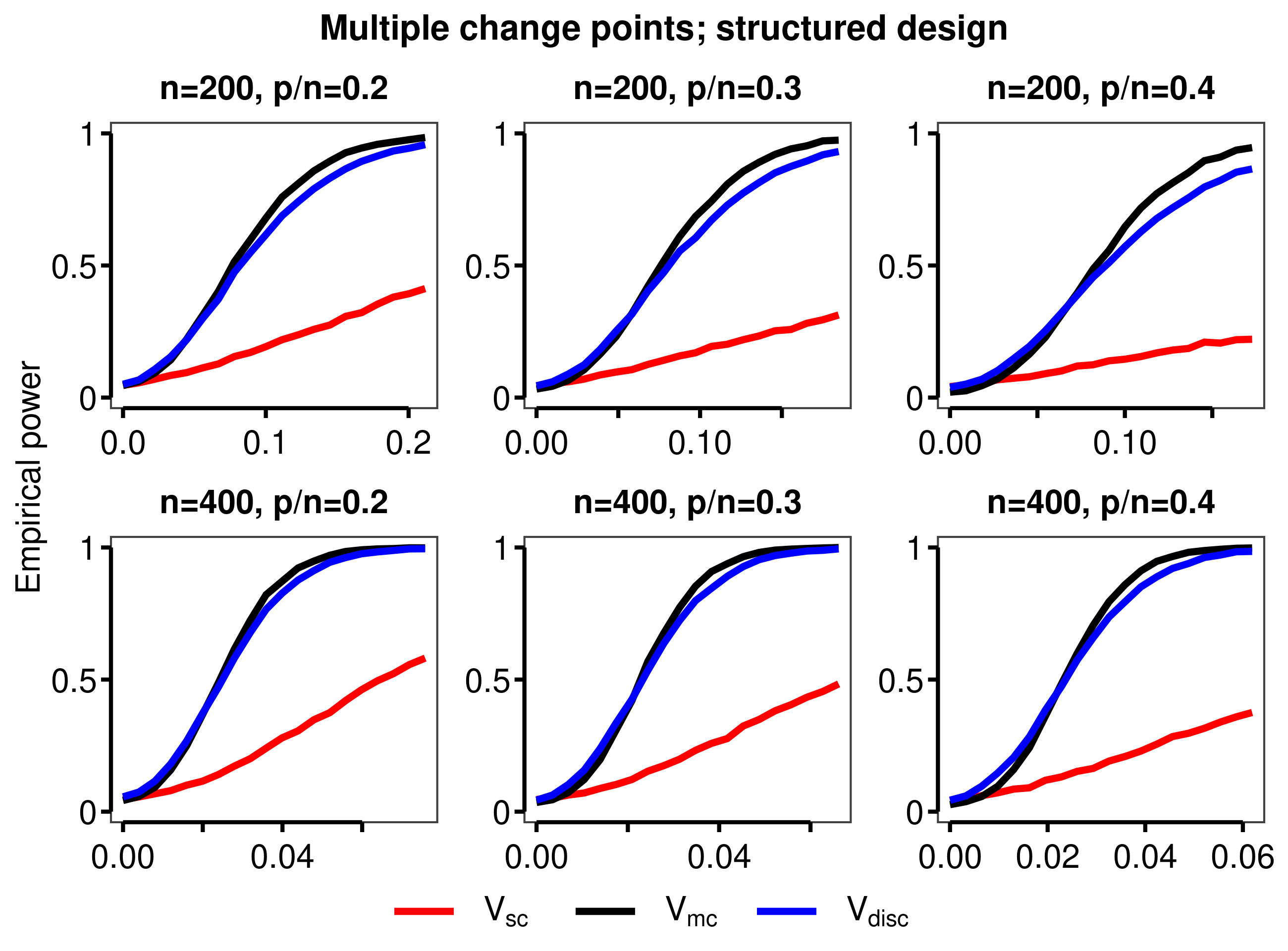}
    \includegraphics[width=0.49\linewidth,height=0.45\linewidth]{figures/eps010_power_epidemic_normal_structured.png}
    \caption{Empirical power under epidemic alternatives at $\varepsilon=0.1$ under the structured design setting and normal errors.}
    \label{fig:power_structured_normal_epidemic_eps010}
\end{figure}

To examine empirical power, we consider a single-change alternative and an epidemic alternative with two change points. In particular, in each Monte Carlo replication, we independently generate a shift vector $\delta_p\sim N(0,cI_p)$, independently of the errors, where $c>0$ controls the signal strength. Under this parameterization, $p^{-1}\mE\|\delta_p\|_2^2=c$, so the expected squared magnitude of the shift increases with $c$. We set $\Delta_p=\theta\delta_p^T\in\mathbb{R}^{M\times p}$ with $\theta=(0,1,0)^T$ and consider the following two change-point configurations:
\begin{enumerate}
    \item[(i)] \emph{Single-change alternative.} A persistent shift occurs after $B_1=\lfloor n/2\rfloor$. That is, $\beta_j = 0_{M\times p}$, for $j\leq B_1$; and $\beta_j = \Delta_p$, for $j >B_1$.
    \item[(ii)] \emph{Epidemic alternative.} A temporary shift occurs after $B_1=\lfloor n/4\rfloor$ and ends after $B_2=\lfloor 3n/4\rfloor$. That is, $\beta_j = 0_{M\times p}$, for $j\leq B_1$ or $j>B_2$; and $\beta_j = \Delta_p$, for $B_1 <j \leq B_2$. 
\end{enumerate}
The epidemic alternative represents a departure from the baseline regression relationship followed by a return to that baseline. We test the prescribed contrast $\beta_j^T\theta$. The corresponding contrast shift under either alternative is $\Delta_p^T\theta=\delta_p$. 

Although the experiments are conducted under all three error distributions, we report only the results for standard normal errors, as the other two settings yield similar patterns. Figures~\ref{fig:power_structured_normal_single_eps010} and~\ref{fig:power_structured_normal_epidemic_eps010} display the empirical power curves at the nominal $5\%$ level as functions of the signal parameter $c$, for $\varepsilon=0.1$ under the structured design. Results for the remaining settings are provided in Appendix~\ref{sec:additional_simulation}.

Several patterns emerge. All three scans exhibit increasing power as the signal strength grows. Under the single-change alternative, the $\SC$ test is more powerful than the $\MC$ and $\DISC$ tests. By contrast, under the epidemic alternative, the $\MC$ and $\DISC$ scans are substantially more powerful than $\SC$. This contrast illustrates the advantage of local segment comparisons when changes in opposite directions can partially cancel in a full-sample comparison. The difference between $\MC$ and $\DISC$ is relatively small in the settings considered, because the true change-point configurations are contained in both scanning sets. Differences across design settings are also consistent with the role of the design matrix in determining the drift process in the local-power analysis. Overall, the results support the effectiveness of the proposed tests across the dimensional aspect ratios considered.

\subsection{Power Comparison for Detecting Changes in Means}

We now focus on the special change-in-mean setting, which is a reduced case of the regression model considered in Remark~\ref{remark:change_in_mean}. We select two representative competing methods specifically designed for detecting changes in high-dimensional means:
\begin{itemize}

\item \textbf{U-Stat.}
The unnormalized U-statistic procedure of \citet{wang2022inference}. This method is based on a U-statistic equivalent to the squared $\ell_2$-norm of the unnormalized mean shift $\|\bar{\bY}_{t_2:t_3}- \bar{\bY}_{t_1:t_2}\|_2^2$ but with the squared terms of the form $\bY_j^T \bY_j$ removed, where $\bar{\bY}_{t_1:t_2}$ denotes the sample mean of the observations in the segment $[k(t_1), k(t_2))$. The procedure does not employ empirical covariance normalization, so its power can depend on the population covariance matrix $\Sigma_p$ and on the orientation of the mean shift relative to its eigenspace.
\item \textbf{$\ell_\infty$-Sparse.}
The $\ell_\infty$-norm-based procedures of \citet{jirak2012change,yu2021high}. These methods use statistics related to $\big\|\bar{\bY}_{t_2:t_3}-\bar{\bY}_{t_1:t_2}\big\|_\infty$ to measure departures from the null hypothesis. They are designed to detect sparse mean shifts in which only a relatively small number of coordinates of $\beta_j$ change, and are therefore included as a sparse-signal benchmark.
\end{itemize}

While the proposed method is invariant to $\Sigma_p$, the performance of the competing methods can depend on the spectrum of $\Sigma_p$. We consider three covariance models:
(i) \emph{Identity}: $\Sigma_p=I_p$;
(ii) \emph{Poly Decay}: $\tau_j=0.1+(1 - (j-1)/p)^2$; and
(iii) \emph{Exp Decay}: $\tau_j=\exp(-3j/p)$,
where $\tau_j$, $j=1,\ldots,p$, denote the eigenvalues of $\Sigma_p$. In all cases, $\Sigma_p$ is rescaled so that $\tr(\Sigma_p)=p$. The eigenvectors are randomly generated from the Haar measure on the orthogonal group. The two decay models introduce spectral heterogeneity of different forms. 

Because the competing methods use different and more restrictive scanning strategies in multiple-change-point settings, we restrict the comparison to the single-change-point setup considered in Section~\ref{subsec:empirical_power}. Specifically, we consider a mean shift $\delta_p$ at $B_1=\lfloor n/2\rfloor$, so that
$$ \mE\bY_j=0, \quad j\le\lfloor n/2\rfloor\qquad\text{and} \qquad  \mE\bY_j=\delta_p, \quad j>\lfloor n/2\rfloor.$$

\begin{figure}[htb]
    \centering
    \includegraphics[width=0.49\linewidth,height=0.45\linewidth]{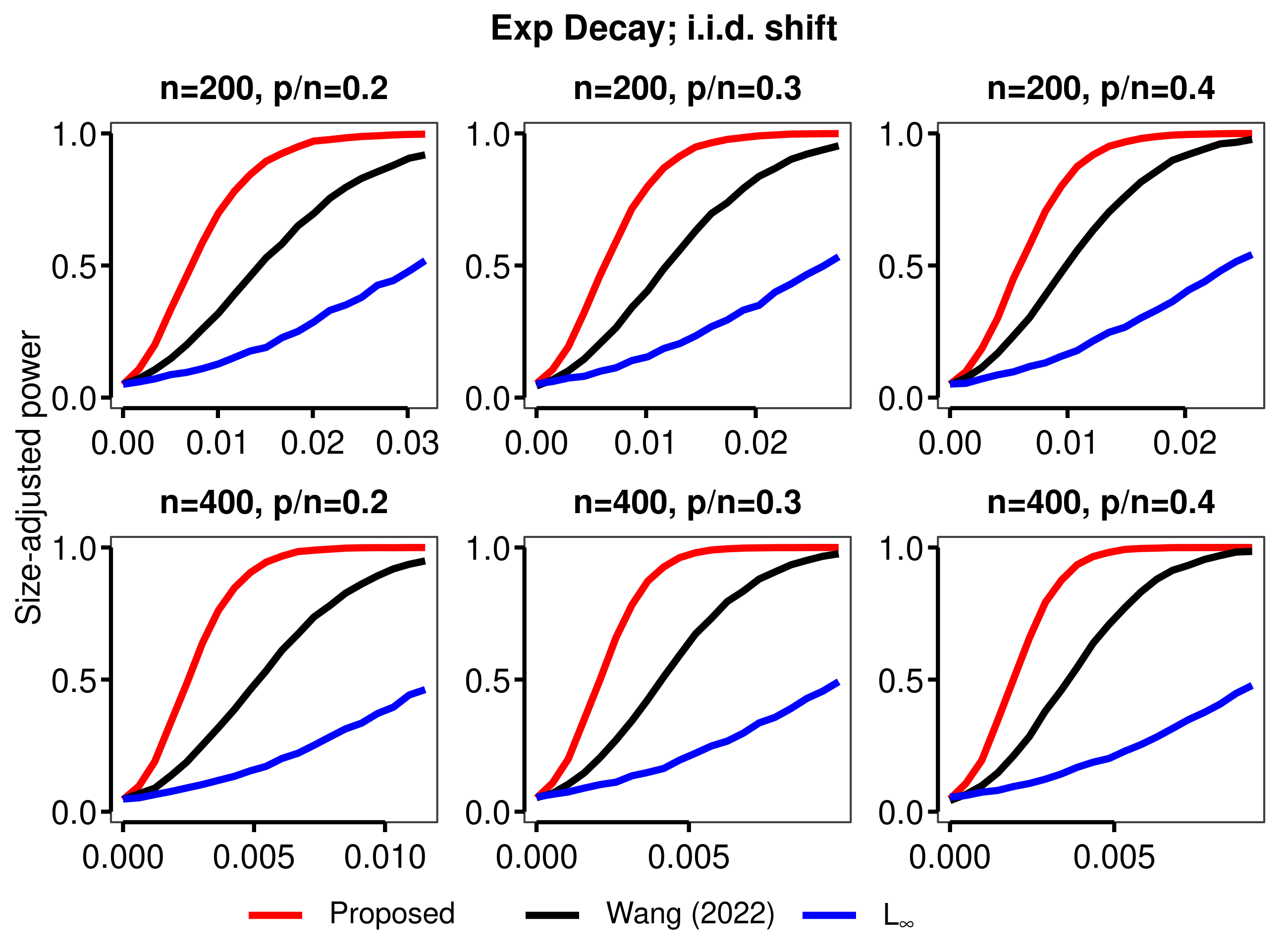}
    \includegraphics[width=0.49\linewidth,height=0.45\linewidth]{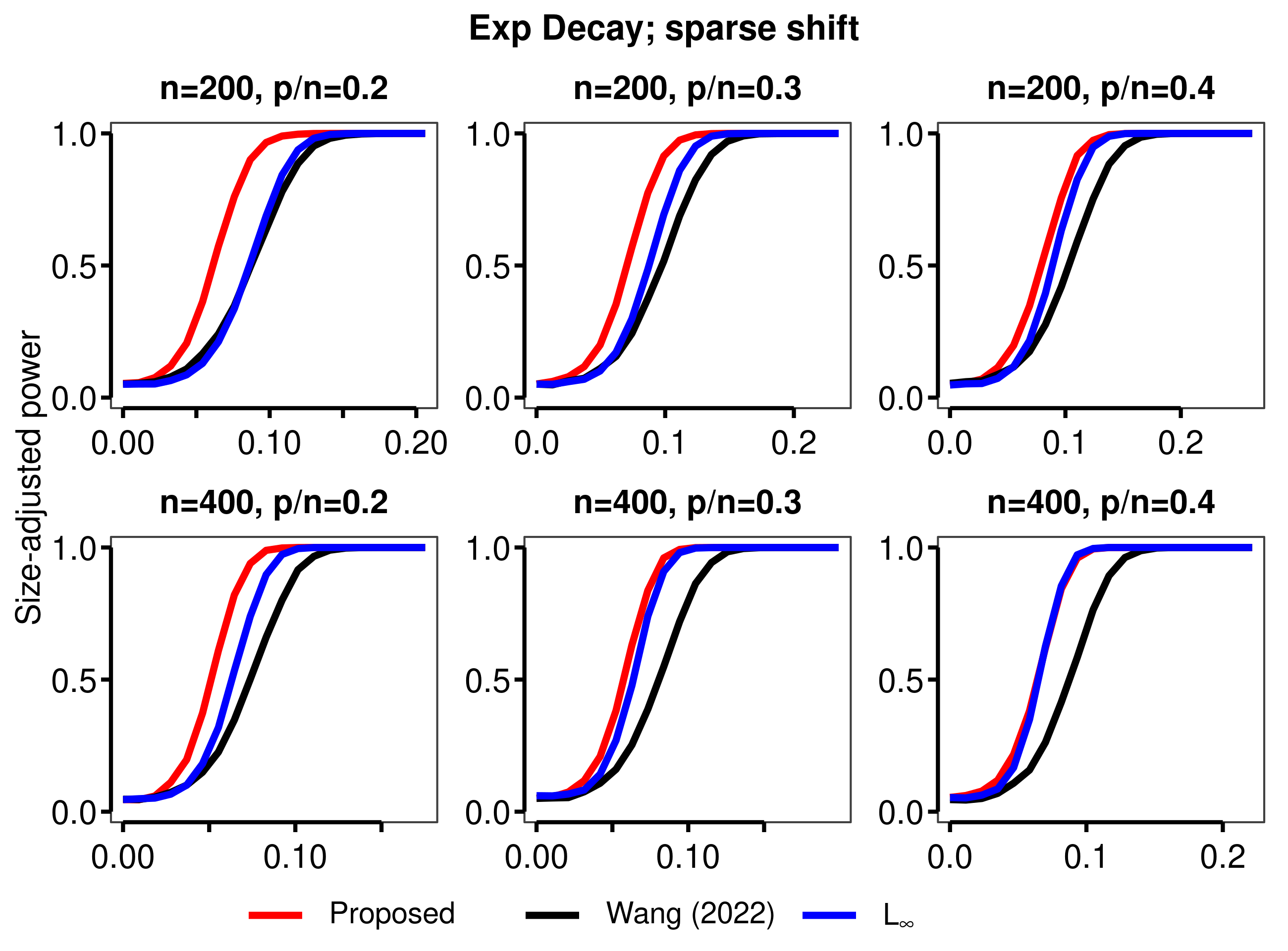}
    \caption{Size-adjusted empirical power under the change-in-mean setup when $\Sigma_p$ is Exp-decay. Left: iid shift. Right:Sparse shift.}
    \label{fig:mean_comp_power_exp}
\end{figure}
We consider two signal structures for $\delta_p$:
\textit{i.i.d. shift}, where $\delta_p\sim N(0,cI_p)$; and
\textit{sparse shift}, where three coordinates are selected uniformly at random and assigned values $\pm 5c$ with independent random signs, while all remaining coordinates are set to zero.

Under this setting, all three methods are applied over the same scanning set $\calT_{\rm sc}(\varepsilon)$ with $\varepsilon=0.1$, allowing a more direct comparison of their respective measures of departure from the null hypothesis.

Because the critical values of the $\ell_\infty$ procedure are difficult to approximate accurately, we simulate the critical values for all three methods and report size-adjusted empirical power to ensure a fair comparison.

Figures~\ref{fig:mean_comp_power_exp} and~\ref{fig:mean_comp_power_identity} report the size-adjusted empirical power under the Exp Decay and identity covariance models, respectively. Results for the Poly Decay model are provided in Appendix~\ref{sec:additional_simulation}. Across the dimensional regimes considered, the proposed procedure shows a clear advantage when the spectrum of $\Sigma_p$ is heterogeneous, highlighting the benefit of covariance normalization in such settings. When $\Sigma_p=I_p$, the covariance structure is isotropic and therefore particularly favorable to the unnormalized $\ell_2$-type procedure of \citet{wang2022inference}; even in this setting, the proposed procedure achieves comparable power. The results further indicate that the proposed method remains sensitive to sparse alternatives, with empirical power comparable to that of the $\ell_\infty$-based procedure under the settings considered.

\begin{figure}[htbp]
    \centering
    \includegraphics[width=0.49\linewidth,height=0.45\linewidth]{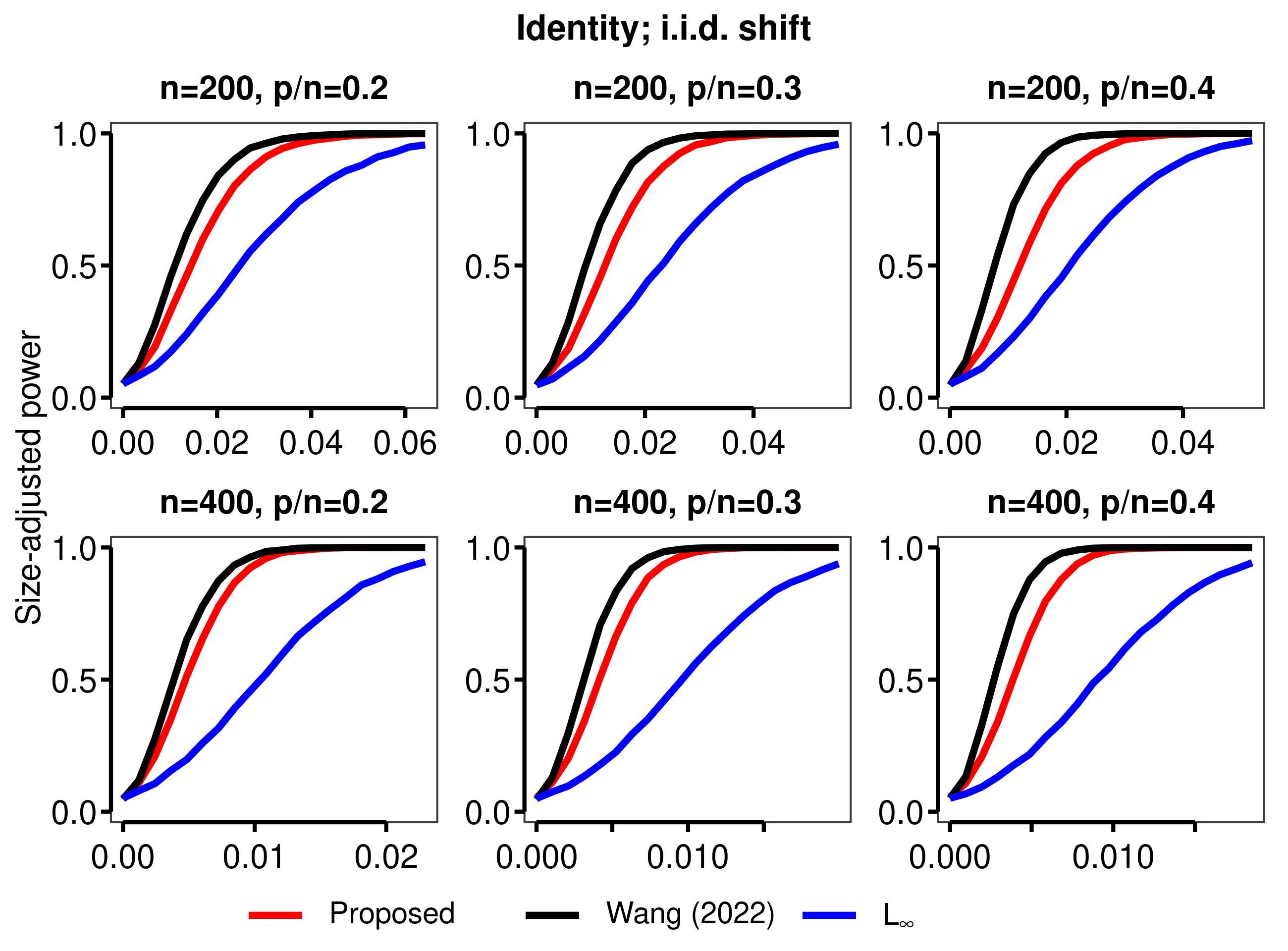}
    \includegraphics[width=0.49\linewidth,height=0.45\linewidth]{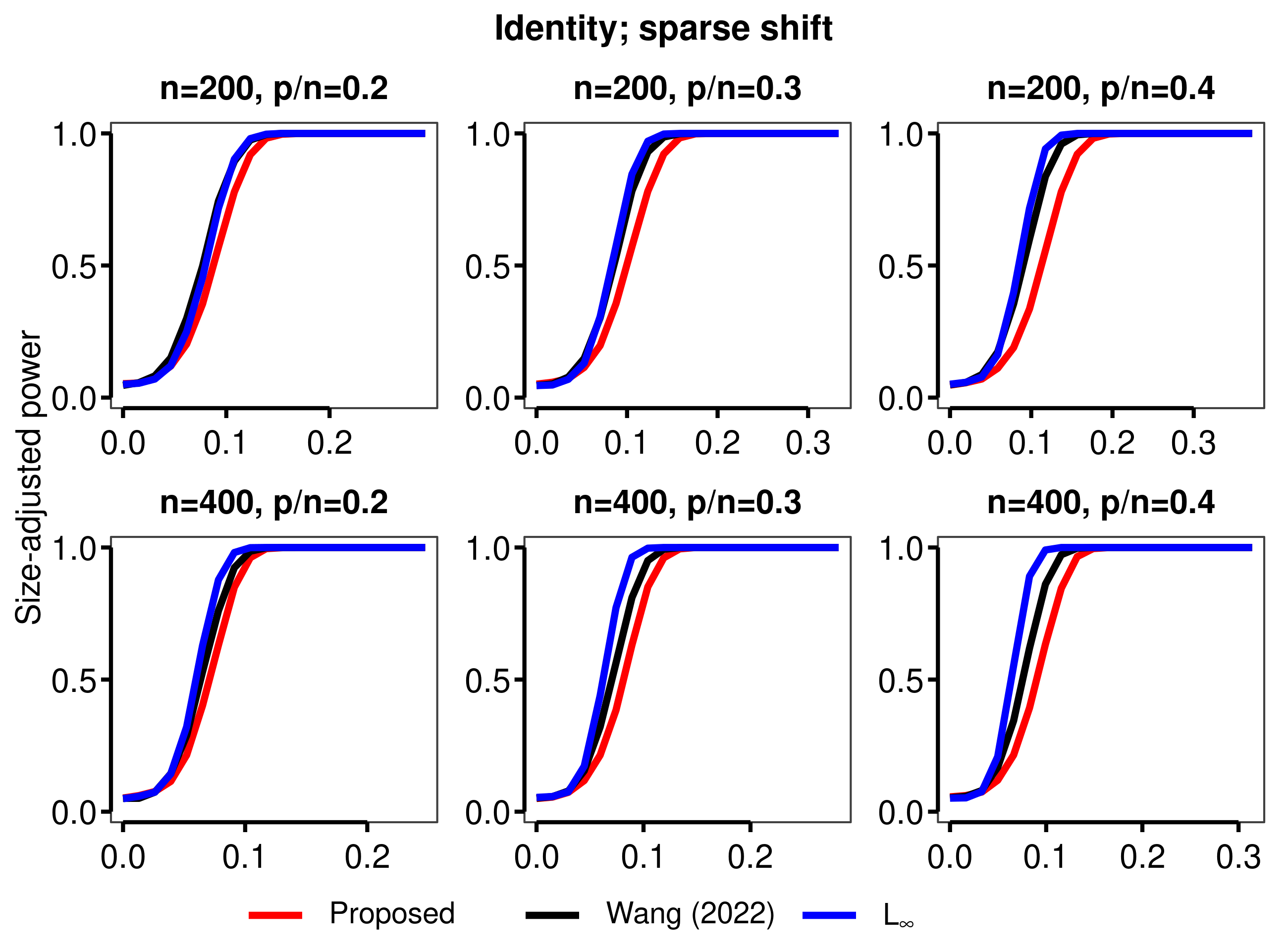}
    \caption{Size-adjusted empirical power under the change-in-mean setting with $\Sigma_p=I_p$. Left: i.i.d. shift. Right: sparse shift.}
    \label{fig:mean_comp_power_identity}
\end{figure}
\section{Real Data Application}
\label{sec:real_data}
We investigate structural changes in factor exposures of U.S. equity portfolios using data from the \emph{Kenneth R. French Data Library}. The library consists of 100 portfolios formed on size and book-to-market equity, which provide a widely used cross section of U.S. equity returns in the asset-pricing literature. The portfolios span different combinations of market capitalization and book-to-market characteristics, providing a natural high-dimensional setting for examining the stability of exposures to common equity-market factors.

We use monthly value-weighted portfolio returns together with the Fama--French market, size, and value factors \citep{fama1993common}. The available return panel covers July 1926--July 2026 and contains approximately $1.2\%$ missing entries. Our primary analysis focuses on July 1978--June 1999, the longest contiguous period for which all 100 portfolios are observed, yielding $n=252$ and $p=100$. To examine the results across different historical periods, we additionally consider the full sample, three consecutive 30-year windows beginning in July 1926, and the most recent 30-year window, August 1996--July 2026. Within each analysis window, we retain only portfolios with complete records, resulting in response dimensions ranging from $70$ to $100$.

Financial returns commonly exhibit volatility clustering and short-range serial dependence. To mitigate these effects, we apply marginal AR--GARCH filtering to each retained excess-return series and use the resulting standardized innovations as the components of the multivariate response vector $\bY_j$. The marginal Ljung-Box tests with 12 lags show no significant evidence of remaining serial corelation in either the innovations or their squares. Figure~\ref{fig:ff_innovations} in Appendix \ref{sec:additional_realdata} displays six representative innovation series from the primary analysis window.

For each analysis window, we consider the multivariate factor regression
\begin{equation}\label{eq:realdata_regression}
\bY_j=\beta_j^T\bX_j+\boldsymbol\epsilon_j,
\qquad
\bX_j=(1,\mathrm{MKT}_j,\mathrm{SMB}_j,\mathrm{HML}_j)^T,
\qquad
\beta_j\in\mathbb R^{4\times p}.
\end{equation}
Here, $\mathrm{MKT}$, $\mathrm{SMB}$, and $\mathrm{HML}$ are the three factors in the Fama--French three-factor model \citep{fama1993common}, which are widely used in empirical asset pricing to capture market, size, and value effects, respectively. Rather than the conventional Fama--French factor loadings, the coefficients $\beta_j$ represent factor exposures of the filtered and standardized returns. Our analysis concerns structural changes in these adjusted factor exposures over time. Figure~\ref{fig:rolling_estimate} displays rolling estimates of the coefficient components based on 10-year observation windows. 

Table~\ref{tab:ff_primary} summarizes the results for the primary analysis window, July 1978--June 1999. The analysis for the other analysis window is included in Appendix \ref{sec:additional_realdata}. We report the $p$-values from both the SC and MC procedures for testing coefficient stability using each of the four coefficient contrasts, together with their Holm--Bonferroni-adjusted values \citep{holm1979simple} to account for multiple testing across the four contrasts for each scan and each analysis window. 
\begin{figure}[htbp]
    \centering
    \includegraphics[width=0.9\linewidth]{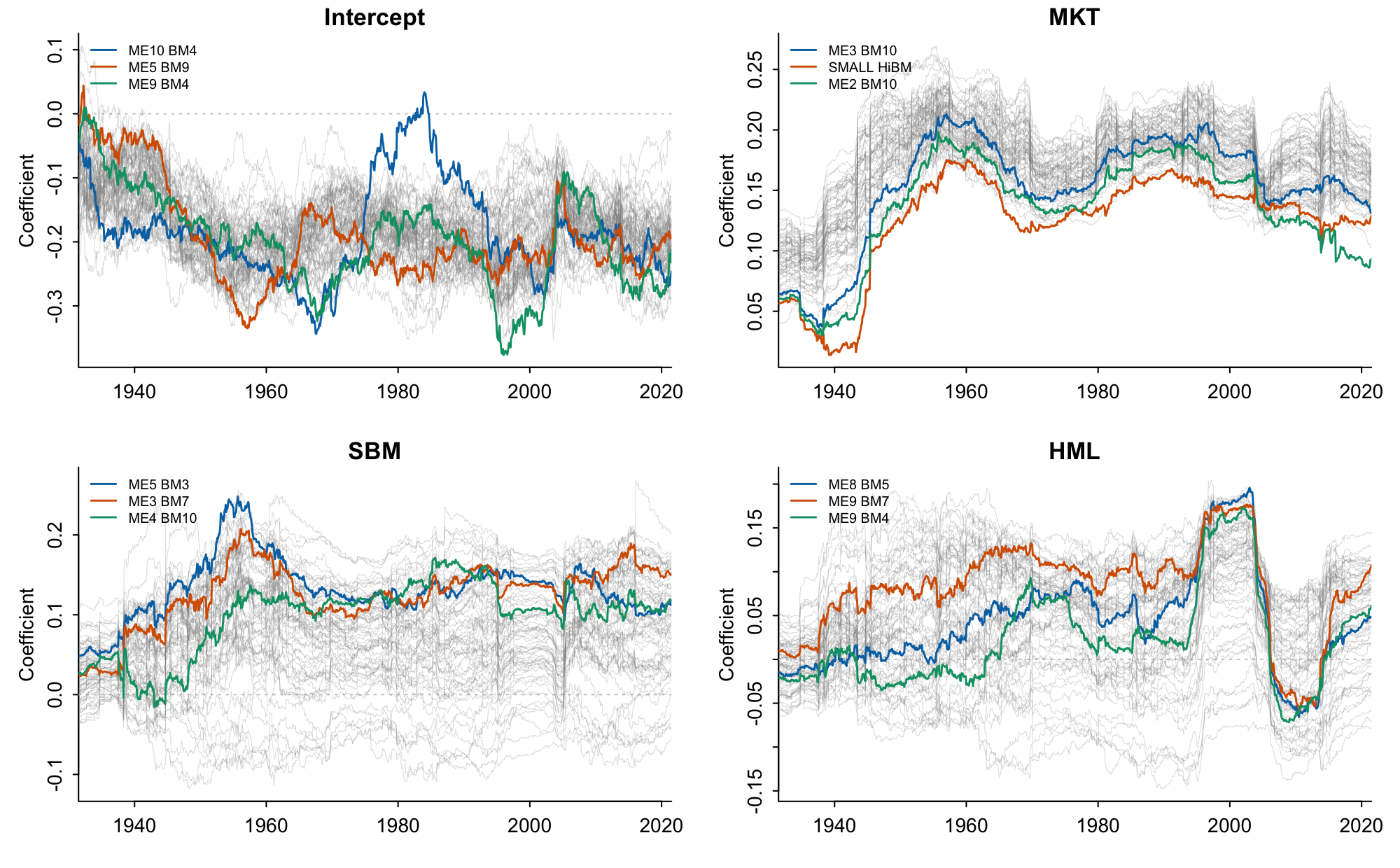}
    \caption{Rolling estimates of the factor coefficients based on 10-year observation windows centered at each month. For each coefficient, the colored curves correspond to the three portfolios with the largest differences between their $95\%$ and $5\%$ quantiles over time.}
    \label{fig:rolling_estimate}
\end{figure}
At the $5\%$ family-wise significance level, both procedures reject coefficient stability using the HML contrast. The SC test also yields a significant unadjusted $p$-value using the market contrast, while the MC test yields a significant unadjusted $p$-value using the SMB contrast, although neither remains significant after adjustment for multiple testing.

The results obtained using the HML contrast motivate further examination of the relationship between portfolio returns and the value factor over the period considered. \citet{fama1995size} document systematic relationships between book-to-market characteristics and firm profitability, providing relevant economic context for interpreting variation in HML exposures. The sample period also encompasses substantial changes in economic and market conditions, including the U.S. recession of July 1990--March 1991 \citep{nbercycles}.

As discussed in Section~\ref{subsec:further_discussion}, formal localization of the change points is beyond the scope of the current work. Nevertheless, the rolling coefficient estimates in Figure~\ref{fig:rolling_estimate} provide a descriptive view of how the factor exposures evolve over time. The intercept and the coefficients of MKT and SMB exhibit noticeably different patterns before and after approximately 1958, while the HML coefficient displays pronounced variation over the period from approximately 1990 to 2010. These patterns are broadly consistent with the structural instability detected by the proposed tests, although they should not be interpreted as formal estimates of the change-point locations.

\begin{table}[htbp]
\centering
\caption{Change-point detection results of the SC and MC tests for changes in the regression coefficients over July 1978--June 1999 ($n=252$, $p=100$), with $\varepsilon=0.1$. }
\label{tab:ff_primary}
\begin{tabular}{lrrrr}
\toprule
& \multicolumn{2}{c}{$\SC$} & \multicolumn{2}{c}{$\MC$} \\
\cmidrule(lr){2-3}\cmidrule(lr){4-5}
Coefficient & Original $p$ & Adjusted $p$ & Original $p$ & Adjusted $p$ \\
\midrule
Intercept & 0.7590 & 0.7590 & 0.9860 & 0.9860 \\
MKT       & 0.0310 & 0.0930 & 0.1145 & 0.2290 \\
SMB       & 0.1520 & 0.3040 & 0.0180 & 0.0540 \\
HML       & 0.0005 & 0.0020 & 0.0005 & 0.0020 \\
\bottomrule
\end{tabular}
\end{table}

\section{Discussion}\label{sec:discussion}

This paper develops structural-change tests for multivariate regression models in which the response dimension grows proportionally to the sample size while the predictor dimension remains fixed. The proposed procedures detect changes in a prescribed regression contrast using a least-squares-based Wald statistic standardized by the residual covariance matrix. To accommodate different change-point configurations, we introduce three scanning strategies. Under suitable regularity conditions, we establish weak convergence of the normalized statistic process to a centered Gaussian process whose covariance kernel can be approximated from the observed design. The local-power analysis further characterizes how the signal magnitude, regression design, and dimensional aspect ratio jointly determine asymptotic power.

The numerical results show that the proposed Gaussian approximation generally provides accurate finite-sample size control, although the multiple-change and discretized scans tend to be slightly conservative in some settings. The heavy-tailed experiments indicate a degree of robustness beyond the conditions imposed by the theory, while some size distortions arise in more challenging settings. The portfolio application illustrates the use of the proposed methods for assessing the stability of factor exposures of filtered returns in a high-dimensional setting, with the results providing strong evidence of changes in the exposure to the value factor.

An important extension is to accommodate temporal dependence across observations. This requires a better understanding of the spectral behavior of high-dimensional residual covariance estimators under temporal dependence. Relevant advances include results for simultaneously diagonalizable linear processes \citep{liu2015marcenko}, finite-lag linear processes \citep{deitmar2026spectral}, and dynamic volatility models \citep{ding2024high}. These developments provide useful foundations, but covariance normalization and calibration remain challenging under general temporal dependence, particularly for change-point problems where process-level limits are required. We leave this extension for future research.

\bibliographystyle{apalike} 
\bibliography{reference}       

\appendix

\makeatletter
\setlength{\@fptop}{50pt}
\setlength{\@fpsep}{8pt plus 1fil}
\setlength{\@fpbot}{50pt}
\makeatother

\section{Additional Simulation Results}
\label{sec:additional_simulation}

\begin{figure}[htb]
    \centering
    \includegraphics[width=0.49\linewidth,height=0.45\linewidth]{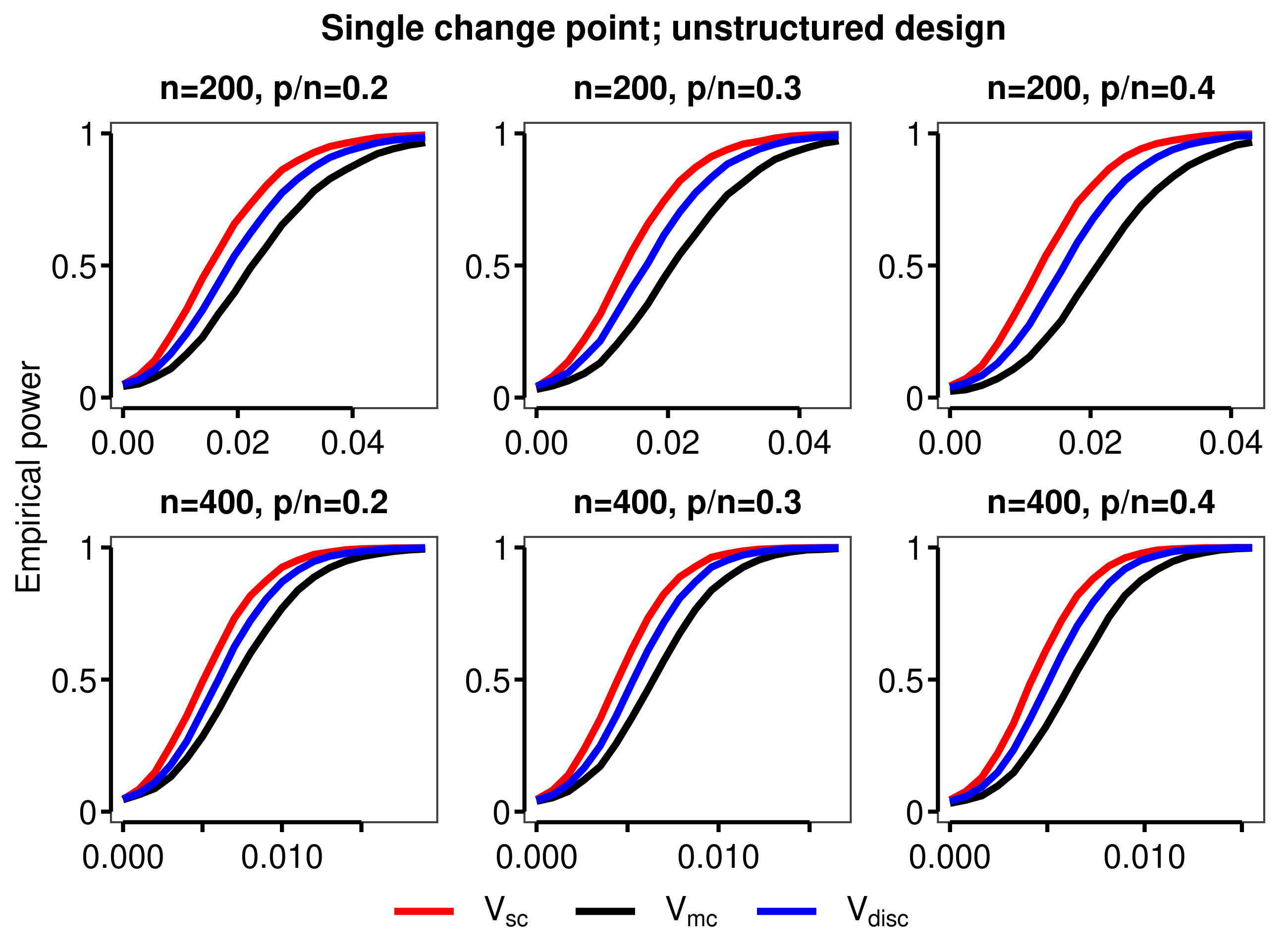}
    \includegraphics[width=0.49\linewidth,height=0.45\linewidth]{figures/eps010_power_single_normal_unstructured.png}
    \caption{Empirical power under single-change alternatives at $\varepsilon=0.1$ under the unstructured design setting and normal errors.}
    \label{fig:power_unstructured_normal_single_eps010}
\end{figure}

\begin{figure}[htb]
    \centering
    \includegraphics[width=0.49\linewidth,height=0.45\linewidth]{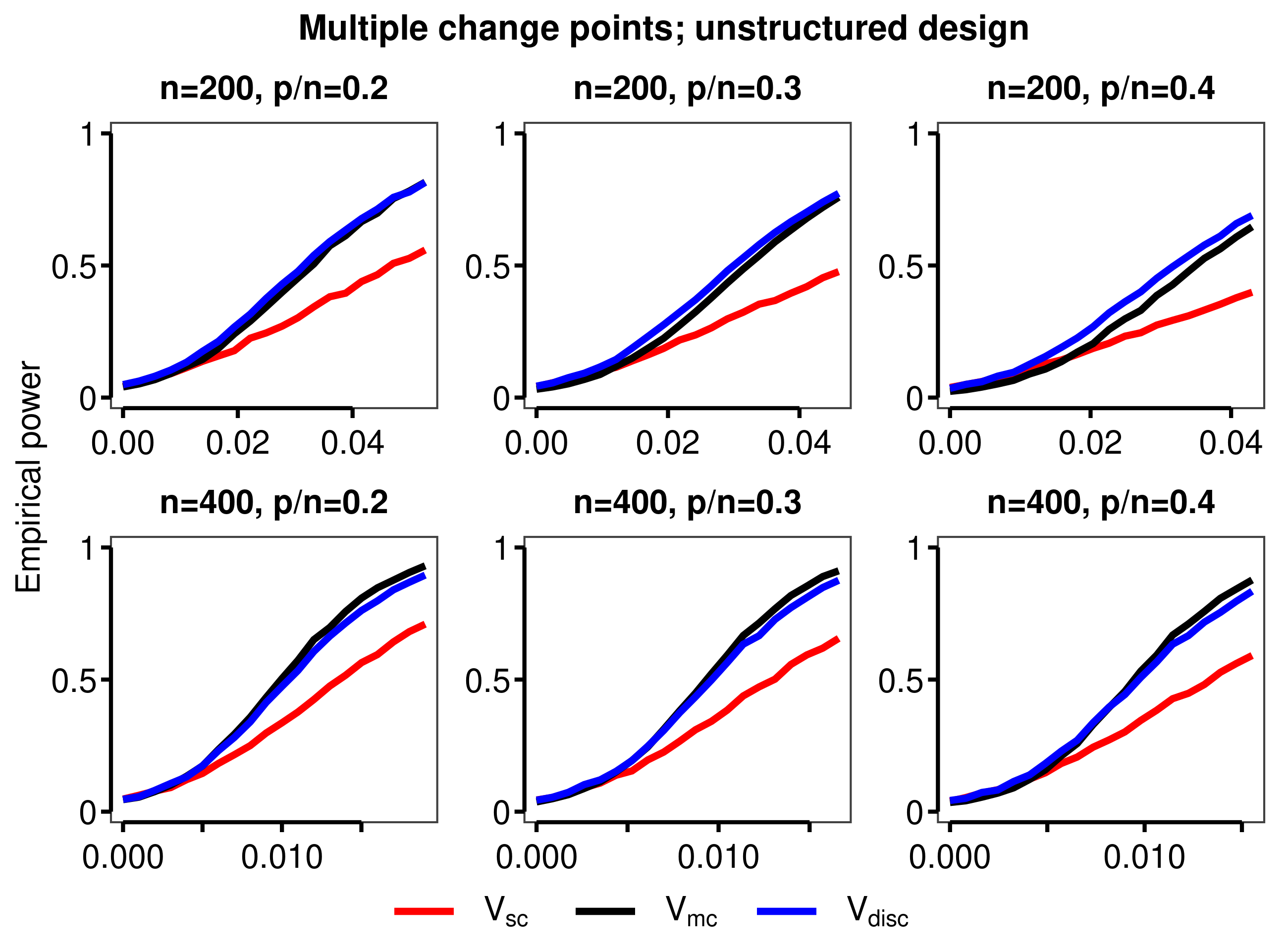}
    \includegraphics[width=0.49\linewidth,height=0.45\linewidth]{figures/eps010_power_epidemic_normal_unstructured.png}
    \caption{Empirical power under epidemic alternatives at $\varepsilon=0.1$ under the unstructured design setting and normal errors.}
    \label{fig:power_unstructured_normal_epidemic_eps010}
\end{figure}

\begin{figure}[htb]
    \centering
    \includegraphics[width=0.49\linewidth,height=0.45\linewidth]{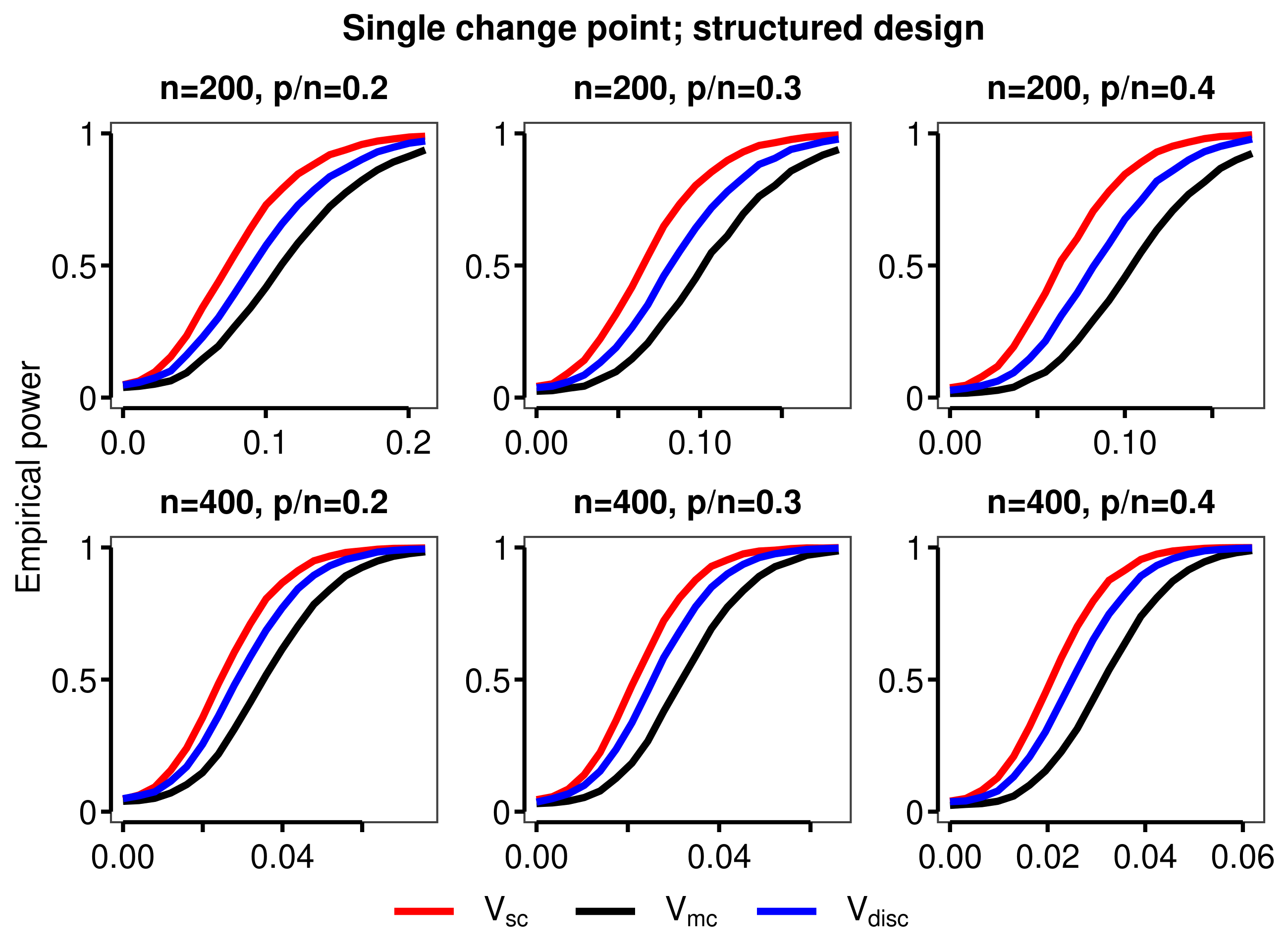}
    \includegraphics[width=0.49\linewidth,height=0.45\linewidth]{figures/eps005_power_single_normal_structured.png}
    \caption{Empirical power under single-change alternatives at $\varepsilon=0.05$ under the structured design setting and normal errors.}
    \label{fig:power_structured_normal_single_eps005}
\end{figure}

\begin{figure}[htb]
    \centering
    \includegraphics[width=0.49\linewidth,height=0.45\linewidth]{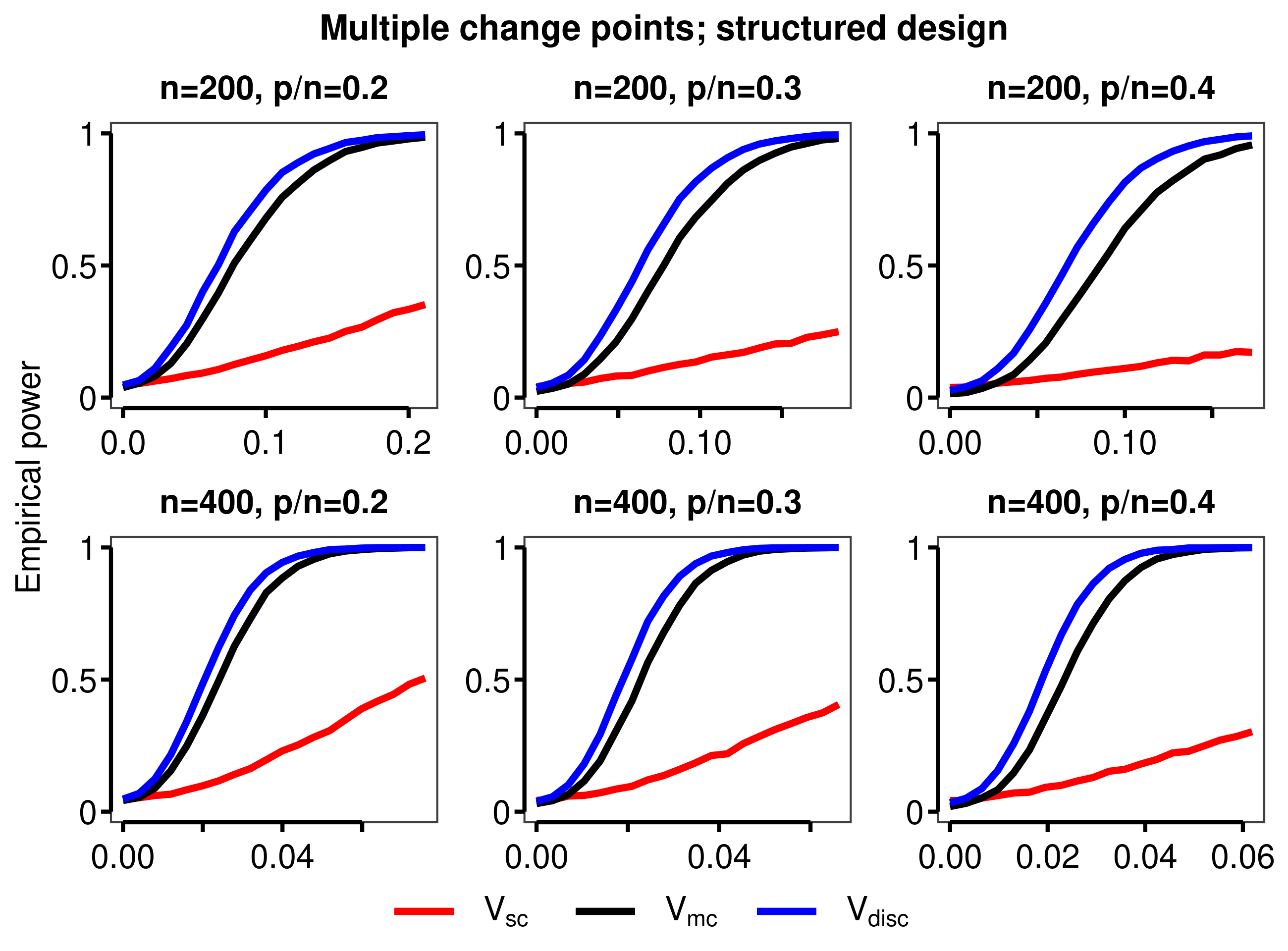}
    \includegraphics[width=0.49\linewidth,height=0.45\linewidth]{figures/eps005_power_epidemic_normal_structured.png}
    \caption{Empirical power under epidemic alternatives at $\varepsilon=0.05$ under the structured design setting and normal errors.}
    \label{fig:power_structured_normal_epidemic_eps005}
\end{figure}

\begin{figure}[htb]
    \centering
    \includegraphics[width=0.49\linewidth,height=0.45\linewidth]{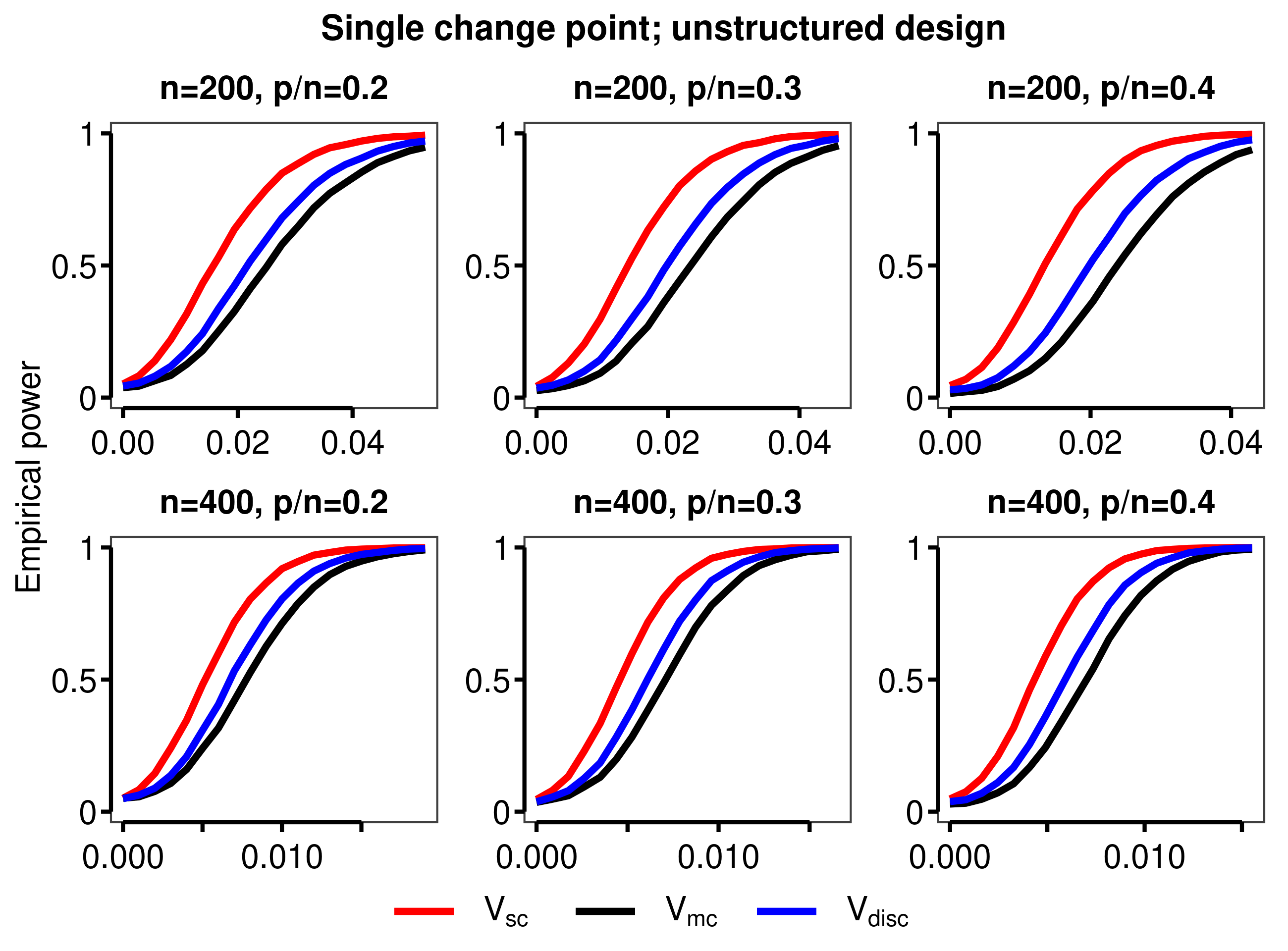}
    \includegraphics[width=0.49\linewidth,height=0.45\linewidth]{figures/eps005_power_single_normal_unstructured.png}
    \caption{Empirical power under single-change alternatives at $\varepsilon=0.1$ under the unstructured design setting and normal errors.}
    \label{fig:power_unstructured_normal_single_eps005}
\end{figure}

\begin{figure}[htb]
    \centering
    \includegraphics[width=0.49\linewidth,height=0.45\linewidth]{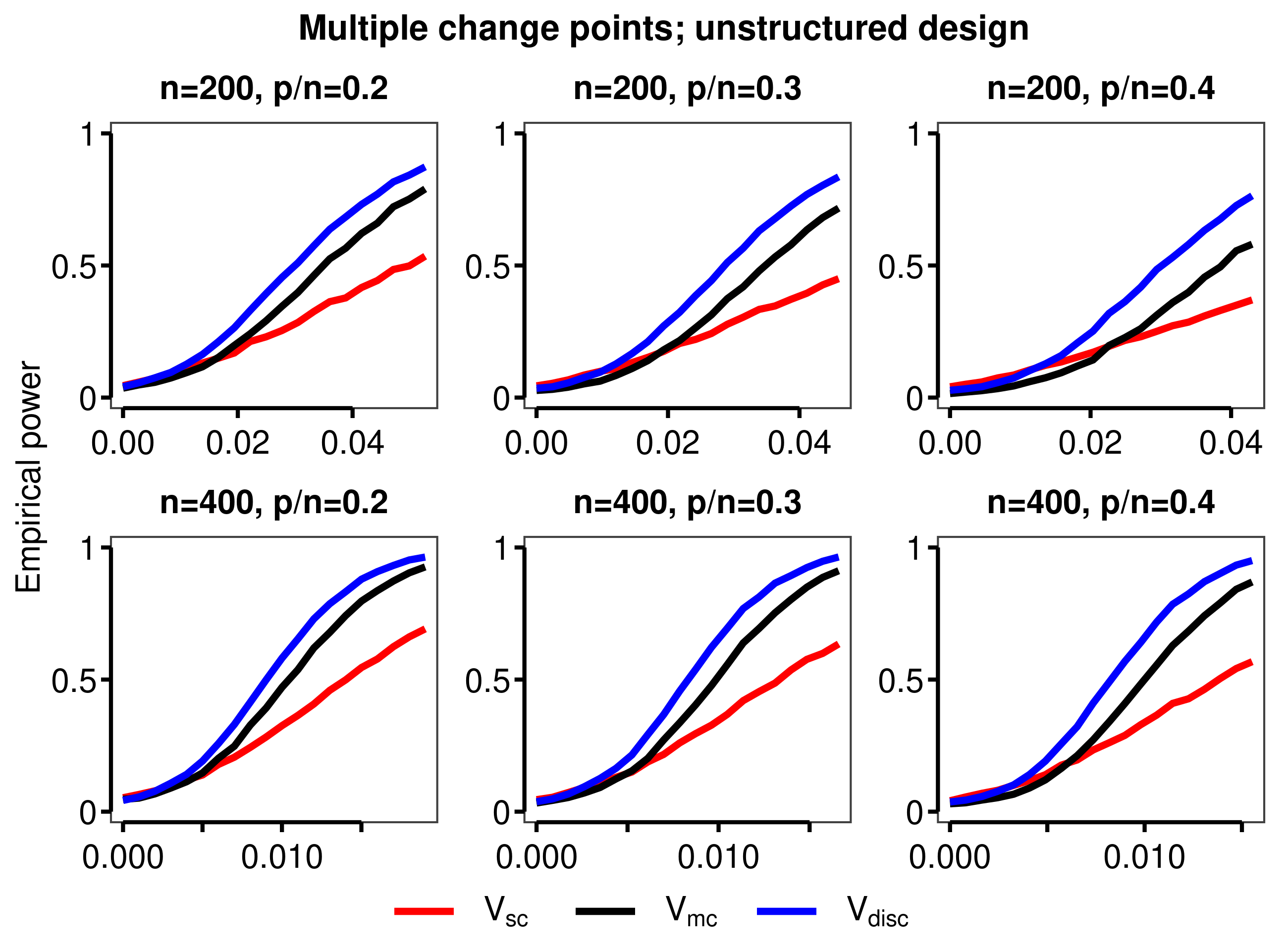}
    \includegraphics[width=0.49\linewidth,height=0.45\linewidth]{figures/eps005_power_epidemic_normal_unstructured.png}
    \caption{Empirical power under epidemic alternatives at $\varepsilon=0.05$ under the unstructured design setting and normal errors.}
    \label{fig:power_unstructured_normal_epidemic_eps005}
\end{figure}

\begin{figure}[htb]
    \centering
    \includegraphics[width=0.49\linewidth,height=0.45\linewidth]{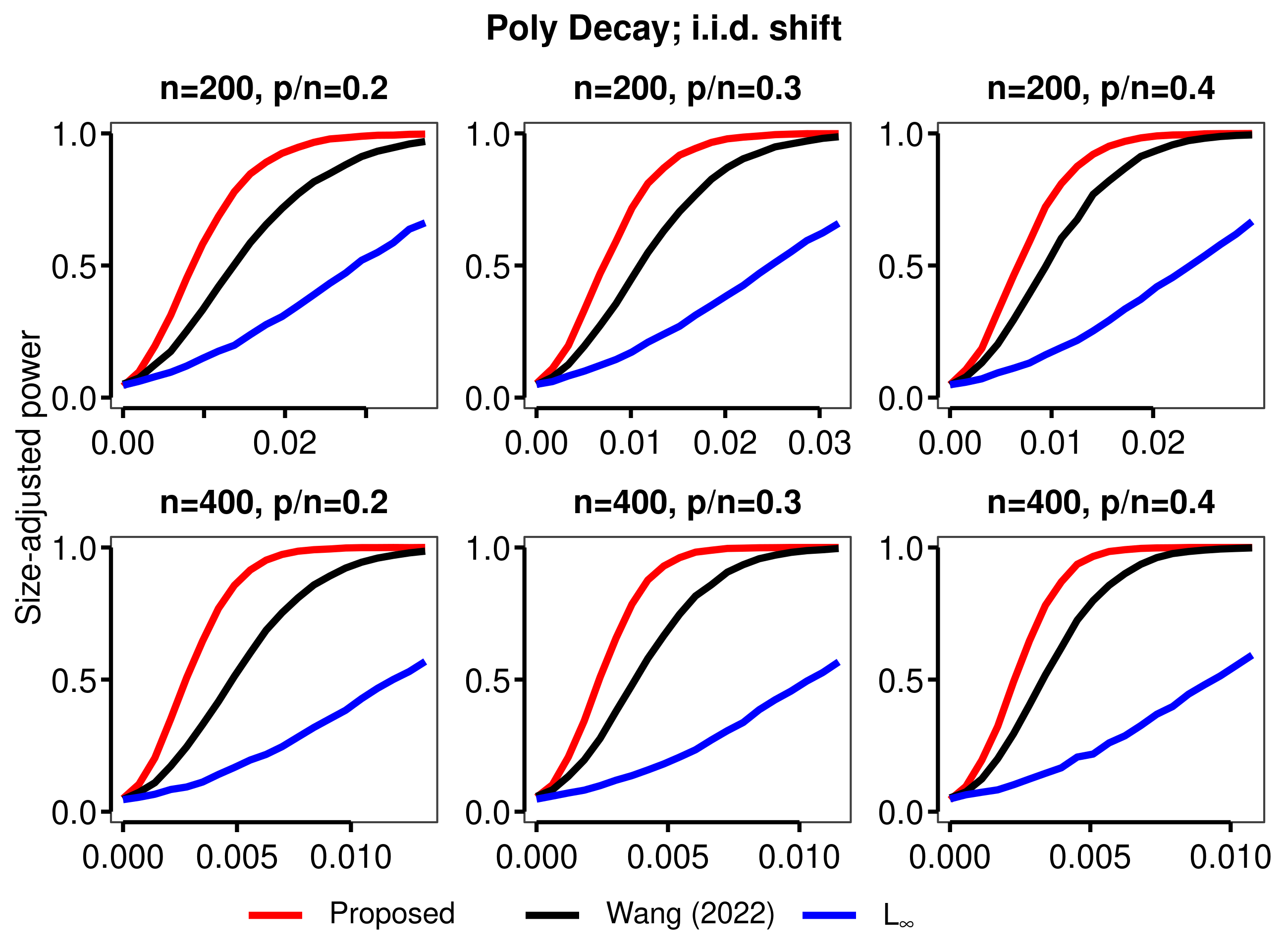}
    \includegraphics[width=0.49\linewidth,height=0.45\linewidth]{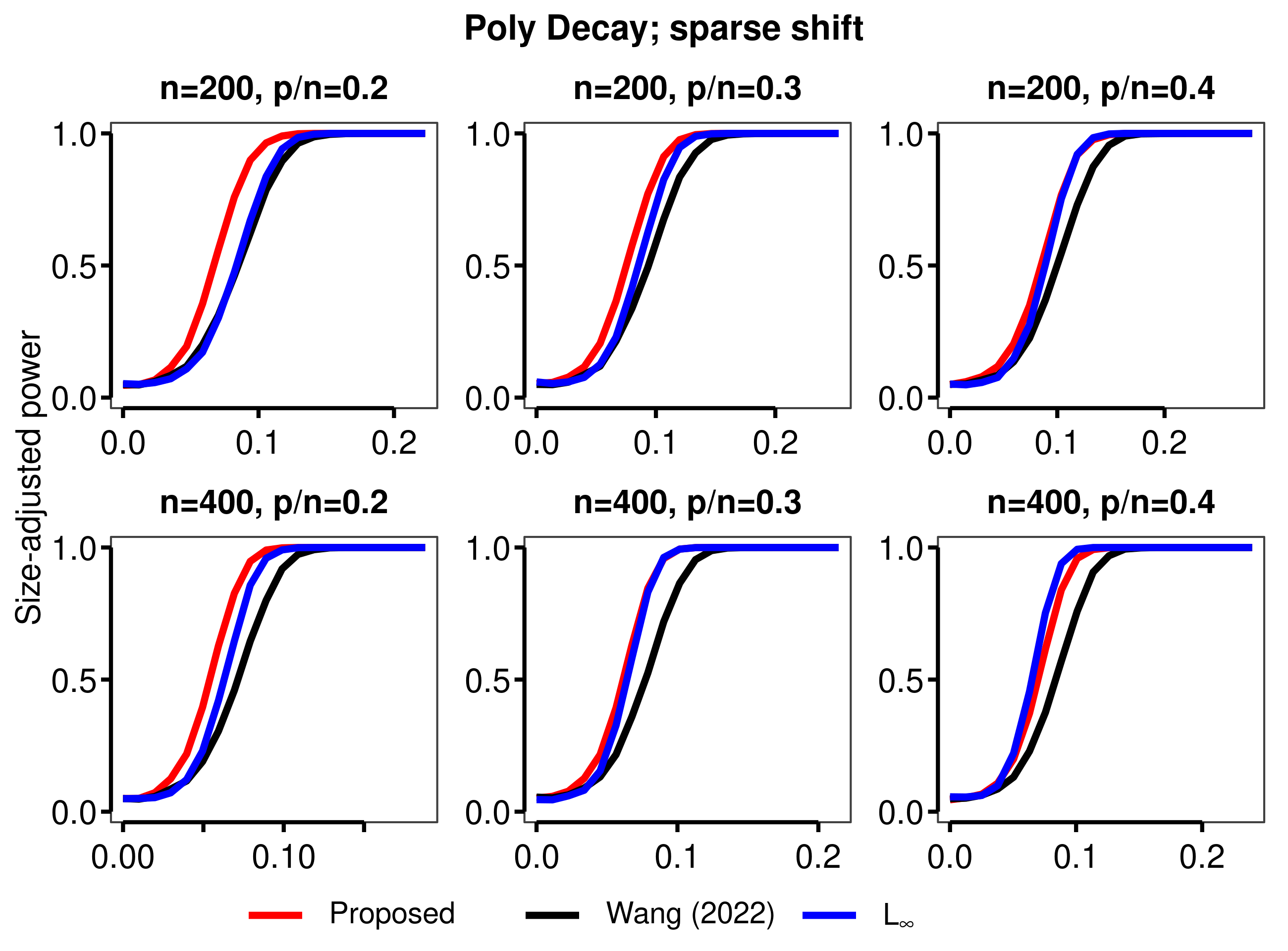}
    \caption{Size-adjusted empirical power under the change-in-mean setting with the Poly Decay covariance model. Left: i.i.d. shift. Right: sparse shift.}
    \label{fig:mean_comp_power_poly}
\end{figure}

\section{Additional Real Data Analysis Results} \label{sec:additional_realdata}
Figure \ref{fig:ff_innovations} displays six representative filtered and volatility-adjusted excess-return innovations used as the responses in the regression analysis. Marginal Ljung--Box tests show no significant evidence of remaining serial correlation in either the innovations or their squares.
\begin{figure}[htbp]
 \centering
 \includegraphics[width=0.8\linewidth]{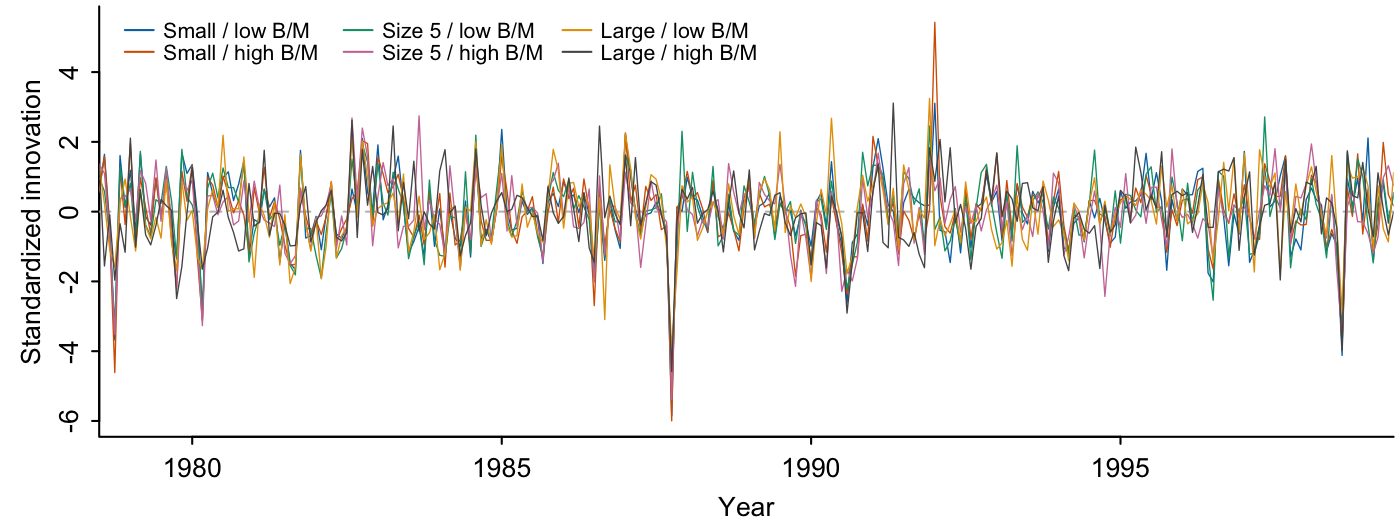}
 \caption{Six representative standardized marginal AR(1)--GARCH(1,1)--$t$ innovations, July 1978--June 1999.} 
 \label{fig:ff_innovations}
\end{figure}

Tables~\ref{tab:ff_1926_2026}--\ref{tab:ff_1996_2026} report the change-point detection results for the additional analysis windows. The market, SMB, and HML coefficients are significant after Holm--Bonferroni adjustment under both the SC and MC procedures in all five windows considered. The results for the intercept are less uniform, with significance depending on the analysis period. Time variation in factor exposures has also been considered in the asset-pricing literature. For example, \citet{ferson1999conditioning} study conditional asset-pricing models in which factor relationships depend on observable economic information, while \citet{fama1995size} relate size and book-to-market effects in stock returns to corresponding patterns in firm earnings. The results show that, for the portfolios considered here, changes in the fitted factor exposures arise over several different historical periods.

\begin{table}[htbp]
\centering
\footnotesize
\caption{Change-point detection results of the SC and MC tests for changes in the regression coefficients over July 1926--July 2026 ($n=1201$, $p=70$), with $\varepsilon=0.1$}
\label{tab:ff_1926_2026}
\begin{tabular}{lrrrr}
\toprule
 & \multicolumn{2}{c}{$\SC$} & \multicolumn{2}{c}{$\MC$} \\
\cmidrule(lr){2-3}\cmidrule(lr){4-5}
Contrast & Original $p$ & Adjusted $p$ & Original $p$ & Adjusted $p$ \\
\midrule
Intercept & 0.0005 & 0.0020 & 0.0040 & 0.0040 \\
Market    & 0.0005 & 0.0020 & 0.0005 & 0.0020 \\
SMB       & 0.0005 & 0.0020 & 0.0005 & 0.0020 \\
HML       & 0.0005 & 0.0020 & 0.0005 & 0.0020 \\
\bottomrule
\end{tabular}
\end{table}

\begin{table}[htbp]
\centering
\footnotesize
\caption{Same as Table \ref{tab:ff_1926_2026} but for July 1926--June 1956}
\label{tab:ff_1926_1956}
\begin{tabular}{lrrrr}
\toprule
 & \multicolumn{2}{c}{$\SC$} & \multicolumn{2}{c}{$\MC$} \\
\cmidrule(lr){2-3}\cmidrule(lr){4-5}
Contrast & Original $p$ & Adjusted $p$ & Original $p$ & Adjusted $p$ \\
\midrule
Intercept & 0.0005 & 0.0020 & 0.0010 & 0.0020 \\
Market    & 0.0005 & 0.0020 & 0.0005 & 0.0020 \\
SMB       & 0.0005 & 0.0020 & 0.0005 & 0.0020 \\
HML       & 0.0005 & 0.0020 & 0.0005 & 0.0020 \\
\bottomrule
\end{tabular}
\end{table}

\begin{table}[htbp]
\centering
\footnotesize
\caption{Same as Table \ref{tab:ff_1926_2026} but for July 1956--June 1986}
\label{tab:ff_1956_1986}
\begin{tabular}{lrrrr}
\toprule
 & \multicolumn{2}{c}{$\SC$} & \multicolumn{2}{c}{$\MC$} \\
\cmidrule(lr){2-3}\cmidrule(lr){4-5}
Contrast & Original $p$ & Adjusted $p$ & Original $p$ & Adjusted $p$ \\
\midrule
Intercept & 0.0355 & 0.0355 & 0.0095 & 0.0095 \\
Market    & 0.0030 & 0.0075 & 0.0005 & 0.0020 \\
SMB       & 0.0005 & 0.0020 & 0.0005 & 0.0020 \\
HML       & 0.0025 & 0.0075 & 0.0010 & 0.0020 \\
\bottomrule
\end{tabular}
\end{table}

\begin{table}[htbp]
\centering
\footnotesize
\caption{Same as Table \ref{tab:ff_1926_2026} but for July 1986--June 2016}
\label{tab:ff_1986_2016}
\begin{tabular}{lrrrr}
\toprule
 & \multicolumn{2}{c}{$\SC$} & \multicolumn{2}{c}{$\MC$} \\
\cmidrule(lr){2-3}\cmidrule(lr){4-5}
Contrast & Original $p$ & Adjusted $p$ & Original $p$ & Adjusted $p$ \\
\midrule
Intercept & 0.1960 & 0.1960 & 0.0545 & 0.0545 \\
Market    & 0.0005 & 0.0020 & 0.0005 & 0.0020 \\
SMB       & 0.0005 & 0.0020 & 0.0005 & 0.0020 \\
HML       & 0.0005 & 0.0020 & 0.0005 & 0.0020 \\
\bottomrule
\end{tabular}
\end{table}

\begin{table}[htbp]
\centering
\footnotesize
\caption{Same as Table \ref{tab:ff_1926_2026} but for August 1996--July 2026}
\label{tab:ff_1996_2026}
\begin{tabular}{lrrrr}
\toprule
 & \multicolumn{2}{c}{$\SC$} & \multicolumn{2}{c}{$\MC$} \\
\cmidrule(lr){2-3}\cmidrule(lr){4-5}
Contrast & Original $p$ & Adjusted $p$ & Original $p$ & Adjusted $p$ \\
\midrule
Intercept & 0.0030 & 0.0030 & 0.0075 & 0.0075 \\
Market    & 0.0010 & 0.0020 & 0.0005 & 0.0020 \\
SMB       & 0.0005 & 0.0020 & 0.0005 & 0.0020 \\
HML       & 0.0005 & 0.0020 & 0.0005 & 0.0020 \\
\bottomrule
\end{tabular}
\end{table}

\section{A Counterexample for Maxima Locations of the Scan Statistic}\label{sec:counterexample_maxima_location}
Consider the case $M=2$, $p=1$, and $n=6$. Let
\[
\theta=
\begin{pmatrix}
1\\
0
\end{pmatrix},
\qquad
\Sigma_p=1,
\]
and consider the deterministic design matrix
\[
\bX=
\begin{pmatrix}
1&0&0&-1&-2&-1\\
1&3&3&-2&-3&3
\end{pmatrix}.
\]
Suppose that the true break occurs after the third observation and that
\[
\beta_{(1)}
=
\begin{pmatrix}
1\\
2
\end{pmatrix},
\qquad
\beta_{(2)}
=
\begin{pmatrix}
2\\
-1
\end{pmatrix}.
\]
The corresponding population mean sequence is
\[
\left(
\beta_{(1)}^T\bX_1,\ldots,\beta_{(1)}^T\bX_3,
\beta_{(2)}^T\bX_4,\ldots,\beta_{(2)}^T\bX_6
\right)
=
(3,6,6,0,-1,-5).
\]

Let
\[
\mathcal Q(t)
=
\frac{
\bh_t^T\widetilde{\mathcal X}^T
(\beta_{(1)}^T,\beta_{(2)}^T)^T
\Sigma_p^{-1}
(\beta_{(1)}^T,\beta_{(2)}^T)
\widetilde{\mathcal X}\bh_t
}{
\|\bh_t\|_2^2
}.
\]
At the true split $t_0=(0,1/2,1)$,
\[
\bh_{t_0}
=
\begin{pmatrix}
1&
-\frac16&
-\frac16&
\frac{12}{107}&
\frac{29}{107}&
\frac{37}{107}
\end{pmatrix}^T,
\]
and hence
\[
\|\bh_{t_0}\|_2^2=\frac{2429}{1926},
\qquad
(\beta_{(1)}^T,\beta_{(2)}^T)
\widetilde{\mathcal X}\bh_{t_0}
=-1.
\]
It follows that
\[
\mathcal Q(t_0)
=
\frac{1926}{2429}
\approx 0.793.
\]

Consider instead the incorrect split $t_1=(0,2/3,1)$, corresponding to
a break after the fourth observation. Then
\[
\bh_{t_1}
=
\begin{pmatrix}
\frac{20}{37}&
-\frac{9}{37}&
-\frac{9}{37}&
-\frac{17}{37}&
\frac13&
\frac13
\end{pmatrix}^T,
\]
and
\[
\|\bh_{t_1}\|_2^2=\frac{281}{333},
\qquad
(\beta_{(1)}^T,\beta_{(2)}^T)
\widetilde{\mathcal X}\bh_{t_1}
=
-\frac{122}{37}.
\]
Therefore,
\[
\mathcal Q(t_1)
=
\frac{(122/37)^2}{281/333}
=
\frac{133956}{10397}
\approx 12.884.
\]
Thus,
\[
\mathcal Q(t_1)>\mathcal Q(t_0),
\]
showing that, for a general deterministic design, the population signal
criterion need not be maximized at the true change point.

\clearpage

\section{Mathematical Notation}\label{sec:notation}
The rest of the appendix is devoted to the proof of technical results. For reference, we collect below the notation used throughout the proof.
\begin{enumerate}
\itemsep2ex
\item Without further specification, we use $\calK$ to denote a generic positive constant, whose value may differ from line to line. We reserve $\calK_0$ for a universal positive constant whose value remains fixed throughout the paper. 

\item For a deterministic sequence $a_n$, we write $X_n=O_{L_1}(a_n)$ if  $\mE |X_n|  =  O(|a_n|)$. Namely, there exists a constant $\calK>0$, independent of $n$, such that $\mE|X_n|\leq \calK |a_n|$ for all sufficiently large $n$. Similarly, we write $X_n=o_{L_1}(a_n)$ if $\mE |X_n| = o(|a_n|)$. 
 Similarly, we write $X_n=O_{L_2}(a_n)$ and $X_n=o_{L_2}(a_n)$ if $\bigl(\mE|X_n|^2\bigr)^{1/2}=O(|a_n|)$ and $\bigl(\mE|X_n|^2\bigr)^{1/2}=o(|a_n|)$,  respectively.


\item For a deterministic vector $a$, let $a_j$ denote its $j$th entry. Similarly, for a collection of vectors $\{u_q\}$, let $u_{qj}$ denote the $j$th entry of the $q$th vector.

\item For a vector, $\|\cdot\|_2$ and $\|\cdot\|_\infty$ denote the Euclidean norm and the maximum norm, respectively. For a matrix, $\|\cdot\|_2$ denotes the spectral norm.

\item For a matrix $A=(a_{ij})$, $\|A\|_{\max}=\max_{i,j}|a_{ij}|$ denotes the entrywise maximum norm.

\item $ Q = \bX^T (\bX\bX^T)^{-1/2}\in \mathbb{R}^{n\times M}$.

\item We write $\bZ=(\bz_1,\bz_2,\ldots,\bz_n)$. For any $j$, let $\bZ_j$ denote the matrix obtained from $\bZ$ by replacing $\bz_j$ with the zero vector. Likewise, $\bZ_{ij}$ denotes the matrix obtained from $\bZ$ by replacing both $\bz_i$ and $\bz_j$ with the zero vector.

\item $\bOmega = n^{-1} \bZ \bZ^T$.

\item $\bA = (\bOmega + \rho_n I_p)^{-1}$, where $\rho_n = n^{-3/2}$. 

\item $\bA_j  = ( n^{-1} \bZ_j\bZ_j^T + \rho_n I_p)^{-1}$, that is $\bA$ but with $\bZ$ replaced by $\bZ_j$. 

\item $\bA_{ij}  = ( n^{-1} \bZ_{ij}\bZ_{ij}^T + \rho_n I_p)^{-1}$, that is $\bA$ but with $\bZ$ replaced by $\bZ_{ij}$. 

\item $\beta_j = \Big(1+n^{-1}\bz_j^T \bA_j\bz_j\Big)^{-1}$,  $\beta_j^{\tr} = \Big(1+n^{-1} \tr\bA_j\Big)^{-1}$, $\beta^{\mE}  = \Big(1+\mE n^{-1} \tr\bA_1\Big)^{-1}$. 

\item $\theta_j = n^{-1} \bz_j^T \bA_j \bz_j - n^{-1} \tr\bA_j$. 

\item $\varrho_{qj} = n^{-1} \bz_j^T \bA_j \bZ_j u_q v_q^T \bZ_j^T \bA_j \bz_j - n^{-1} u_q^T\bZ_j^T \bA_j^{2}\bZ_jv_q$.

\item $ \sigma(\bz_1,\dots, \bz_j) \mbox{ is the }\sigma\mbox{-algebra generated by }\bz_1,\dots, \bz_j$. 

\item $\mE_j[\cdot]  = \mE[\cdot\mid \sigma(\bz_1,\dots,\bz_j)]$, $j = 0,\dots,n $, with the convention $ \mE_0[\cdot]= \mE[\cdot]$. 


\item For vectors $u =(u_1,\dots, u_n)$ and $v = (v_1,\dots, v_n)$ in $\mathbb{R}^n$, we call
\[ H(u,v) =   p^{-1} u^T \bZ^T \bA \bZ v = p^{-1} \sum_{i,j=1}^n u_i v_j \bz_i^T \bA \bz_j . \] 
If $u = v$, we shall simply write $H(u,u) = H(u)$. 
\end{enumerate}

\par

\section{Technical Lemmas}\label{s_sec:technical_lemma}

We collect several technical lemmas that are used throughout the proofs. The proofs of some lemmas are omitted, as closely related results have appeared in the literature; see, for example, Section~S.10 of \citet{li2020high} and equations (3.10)--(3.14) of \citet{Pan2011central}. The matrices in the following lemmas may be either real- or complex-valued.

The bounds below apply both to the original noise matrix $\bZ$ under Condition~\ref{enum:moment_conditions} and to its truncated counterpart satisfying \eqref{eq:truncated_variable_condition}; see Section~\ref{sec:proof_thm_null}. Analogous bounds under related truncation conditions appear in \citet{li2020high} and \citet{Pan2011central}. The rates below are stated for the present setting.

\begin{lemma}[Theorem A.43 of \cite{bai2010spectral}]
\label{lemma:rank_one_perturbation}
    Let $A$ and $B$ be two $p\times p$ Hermitian matrices. Then,
    \[ \|F^A - F^B \|_\infty \leq \frac{1}{p} \operatorname{rank}(A-B).\]
\end{lemma}

\begin{lemma}[Woodbury Matrix Identity]\label{lemma:woodbury}
The following identity holds:
\[
(A + UCV)^{-1} = A^{-1} - A^{-1} U \left( C^{-1} + V A^{-1} U \right)^{-1} V A^{-1},
\]
for conformable matrices \(A,U,C,V\), provided that the displayed inverses exist.
\end{lemma}

\begin{lemma}[\citet{fan1951maximum}]
\label{lemma:singular_value_ineq}
Let $A$ and $C$ be two $p\times n$ complex matrices. Then, for any nonnegative integers $i$ and $j$, we have 
\[s_{i+j+1}(A+C) \leq s_{i+1}(A) + s_{j+1} (C),\]
where $s_i(\cdot)$ is the $i$th largest singular value.
\end{lemma}

\begin{lemma}[Burkholder's inequality]
\label{lemma:Burkholder_martingale_moments}
Let $\{Y_i\}$ be a complex-valued martingale difference sequence with respect to the filtration $\{\sigma_i\}$. Then, for every integer $m\geq2$,
\[\mE \Big|\sum\limits_{i}Y_i \Big|^m \leq \calK_m \mE(\sum\limits_{i} \mE(|Y_i|^2\mid \sigma_{i-1}) )^{m/2} + \calK_m \mE (\sum\limits_{i} |Y_i|^m).\]
\end{lemma}

\begin{lemma}[\citet{bai1998no}]
\label{lemma:Burkholder}
Let $\bY=(Y_1,\ldots,Y_p)^T$, where the $Y_i$ are i.i.d. real-valued random variables with mean zero and variance one, and let $\bB=(b_{ij})_{p\times p}$ be a deterministic complex matrix. Then, for every integer $m\geq2$,
\[\mE|\bY^T \bB \bY -\tr\bB|^m \leq \calK_m (\mE Y_1^4\tr\bB\bB^*)^{m/2} + \calK_m \mE Y_1^{2m} \tr[(\bB\bB^*)^{m/2}],\]
where $\bB^*$ denotes the conjugate transpose of $\bB$, and $\calK_m$ depends only on $m$.
\end{lemma}

\begin{lemma}
\label{lemma:concentration_quadrat}
For a sequence of deterministic matrices $\bD$ such that $\sup_n\|\bD\|_2 < \infty$ and any fixed $m\geq2$,
\begin{align*}
n^{-m} \mE \Big| \bz_1^T \bD \bz_1 - \tr(\bD) \Big|^m &\leq \calK_m n^{-m} \Big\{ [\tr (\bD\bD^*)]^{m/2} + \tr[(\bD\bD^*)^{m/2}] \Big\} = O (n^{-m/2}),
\end{align*}
for some constant $\calK_m$. 
The result follows directly from Lemma~\ref{lemma:Burkholder}.
\end{lemma}

\begin{lemma}
\label{lemma:bound_quad_rankone}
For sequences of deterministic matrices $\bD$ and $\bG$ such that $\sup_n\|\bD\|_2 <\infty$ and $\sup_n\|\bG\|_2<\infty$, and for any fixed $m\geq2$,
\[ n^{-m}\mE \Big|\bz_1^T \bD e_i e_j^T \bG \bz_1 \Big|^m\leq \calK_m n^{-m} \|\bD\|_2^m \|\bG\|_2^m  = O(n^{-m}).\]
The result follows again from Lemma~\ref{lemma:Burkholder}.
\end{lemma}

\begin{lemma}
\label{lemma:zalpha_a_zalpha_a}
For a sequence of deterministic matrices $\bD$ such that $\sup_n\|\bD\|_2 < \infty$ and a sequence of vectors ${u}$ such that $\limsup_{n\to\infty} \sqrt{n}\|{u}\|_{\infty}\leq \calK_{\max}<\infty$,
\[\mE | n^{-1} {u}^T \bZ^T \bD \bZ{u} |^m \leq \calK_m  \|{u}\|_2^{2m} \|\bD\|^m_2, \]
for $ m\geq 2$ and some constant $\calK_m>0$.
\end{lemma}
The proof follows from repeated applications of Lemma~\ref{lemma:Burkholder_martingale_moments}. The argument is adapted from the proof of (3.17) in \citet{Pan2011central}. The only modification is to replace the vectors $\mathbf{s}_i$ in their proof by ${u}_i\bz_i$. We therefore omit the routine details.

\begin{lemma}
\label{lemma:bound_z1_zalpha_b}
For a sequence of deterministic matrices $\bD$ such that $\sup_n\|\bD\|_2 < \infty$ and a sequence of vectors ${u}$ such that $\limsup_{n\to\infty} \sqrt{n}\|{u}\|_{\infty}\leq \calK_{\max}<\infty$,
\[\mE \Big|n^{-1/2} \bz_1^T \bD \bZ_1{u}\Big|^m  \leq \calK_m  \|{u}\|_2^m \|\bD\|^m_2, \]
for  $ m\geq 4$ and some constant $\calK_m>0$.
\end{lemma}

The proof is based on Lemmas~\ref{lemma:Burkholder} and
\ref{lemma:zalpha_a_zalpha_a}. Indeed, by Lemma~\ref{lemma:Burkholder},
\begin{align*}
\mE\Big|n^{-1/2}\bz_1^T\bD\bZ_1{u}\Big|^m  &\le
\calK
\mE\Big|
n^{-1}\bz_1^T\bD\bZ_1{u}{u}^T\bZ_1^T\bD^*\bz_1
-
n^{-1}{u}^T\bZ_1^T\bD^*\bD\bZ_1{u}
\Big|^{m/2} \\
&+
\calK
\mE\Big|
n^{-1}{u}^T\bZ_1^T\bD^*\bD\bZ_1{u}
\Big|^{m/2}.
\end{align*}
The first term is bounded by Lemma~\ref{lemma:Burkholder}, while the second follows directly from Lemma~\ref{lemma:zalpha_a_zalpha_a} with $\bD^*\bD$ in place of $\bD$. Combining the two bounds yields the desired result.


\begin{lemma}
\label{lemma:z1_z2}
For a sequence of deterministic matrices $\bD$ such that $\sup_n\|\bD\|_2 < \infty$,
\[\mE\big| n^{-1/2} \bz_1^T \bD \bz_2 \big|^m \leq \calK_m \|\bD\|_2^m,\]
for  $m\geq4$ and some constant $\calK_m>0$.
\end{lemma}
The proof follows from the same type of moment argument as that used in Lemma~\ref{lemma:bound_z1_zalpha_b}, and is therefore omitted.


\begin{lemma} 
\label{lemma:comprehensive}
For sequences of deterministic matrices $\bD_1,\dots,\bD_m$, $\bG_1,\dots, \bG_s$, and $\bJ$ with uniformly bounded spectral norms, and a sequence of vectors ${u}$ such that $\limsup_{n\to\infty}\sqrt{n}\|{u}\|_{\infty}\leq \calK_{\max}<\infty$,
\begin{align*}
&\mE \Big|\prod\limits_{i=1}^m \frac{1}{n} \bz_1^T \bD_i \bz_1 \prod\limits_{j=1}^s \frac{1}{n} 
(\bz_1^T \bG_j \bz_1 -\tr \bG_j) \big(n^{-1/2}\bz_1^T \bJ\bZ_1{u}\big)^{t} \Big|\\
&\leq \calK_{m,s,t} \prod_{i=1}^m\|\bD_i\|_2 \prod_{j=1}^s\|\bG_j\|_2 \|\bJ\|_2^t \|{u}\|_2^t n^{-s/2}, 
\end{align*}
where $m\geq 0$, $s\geq1$, and $t\geq 0$ are fixed integers, and $\calK_{m,s,t}>0$ is a constant.
\end{lemma}
The result follows from Lemmas~\ref{lemma:concentration_quadrat} and \ref{lemma:bound_z1_zalpha_b} and the Cauchy--Schwarz inequality.


\begin{lemma}
\label{lemma:z1_z1_z1_a}
For a vector $\mathbf{r}$ and deterministic matrices $\bD=(d_{ij})$ and $\bG$,   
\[\mE[(\bz_1^T \bD \bz_1 -\tr\bD) \bz_1^T \bG \mathbf{r} ] = \mE z_{11}^3 \sum\limits_{i=1}^p d_{ii} e_i^T \bG\mathbf{r}, \]
where $e_i$ is the canonical vector with the $i$th entry $1$. 
\end{lemma}

\begin{lemma}
\label{lemma:z_1_z1_z1_z1}
For deterministic matrices $\bD=(d_{ij})$ and $\bG=(g_{ij})$,
\begin{align*}
&\mE[(\bz_1^T \bD \bz_1 -\tr\bD)(\bz_1^T \bG \bz_1 -\tr \bG)] \\
&=  |\mE z_{11}^2|^2 \tr \bD\bG^T + \tr\bD\bG + (\mE z_{11}^4 - |\mE z_{11}^2|^2 -2) \sum\limits_{i=1}^p d_{ii}g_{ii}.
\end{align*}
\end{lemma}

\section{Proof of Theorem \ref{thm:main_null}}\label{sec:proof_thm_null}

Under the null hypothesis $H_0$, the coefficient is constant, so that $\beta_j\equiv\beta_{(0)}$. The statistic is invariant to both $\beta_{(0)}$ and $\Sigma_p$; hence, without loss of generality, we may take $\beta_{(0)}=0$ and $\Sigma_p=I_p$. Indeed, for any $t=(t_1,t_2,t_3)\in\calT_{\rm mc}(\varepsilon)$,
\begin{align*}
&\bY \bh_t = \Big[\beta_{(0)}^T, \beta_{(0)}^T \Big] \calX_t \bu_t + \Sigma_p^{1/2}\bZ \bh_t\\ 
&\phantom{\bh_t\bh_t}= \Big[\beta_{(0)}^T, \beta_{(0)}^T\Big]\tilde{\theta} + \Sigma_p^{1/2} \bZ \bh_t
=\Sigma_p^{1/2}\bZ \bh_t,\\
&\bS = \frac{1}{n} \bY (I_n - \bX^T (\bX \bX^T)^{-1} \bX) \bY^T =  \frac{1}{n} \Sigma_p^{1/2} \bZ (I_n - \bX^T (\bX\bX^T)^{-1}\bX) \bZ^T \Sigma_p^{1/2}.
\end{align*}
The first term in the second line vanishes because $[\beta_{(0)}^T,\beta_{(0)}^T]\tilde{\theta}=0$. 

We shall call $u_t = \bh_t/\|\bh_t\|_2$. Consequently,
\begin{equation}\label{eq:V_t_reform1}
\begin{split}
p^{-1}V_t &= \frac{p^{-1}\bh_t^T \bY^T \bS^{-1} \bY \bh_t}{\|\bh_t\|_2^2 } \\
          &=p^{-1} u_t^T \bZ^T \Big(\frac{1}{n} \bZ (I_n - \bX^T (\bX\bX^T)^{-1}\bX )\bZ^T \Big)^{-1} \bZ u_t\\
          &= p^{-1} \sum_{i,j=1}^n u_{ti} u_{tj} \bz_i^T \Big(\frac{1}{n} \bZ (I_n - \bX^T (\bX\bX^T)^{-1}\bX )\bZ^T \Big)^{-1} \bz_j.
\end{split}
\end{equation}

To analyze the asymptotic behavior of $V_t$, we now work under the reduction $\Sigma_p=I_p$ and write
\[
\bS = \frac{1}{n} \bZ (I_n - \bX^T (\bX\bX^T)^{-1}\bX) \bZ^T
= \bOmega - \frac{1}{n} \bZ Q Q^T \bZ^T,
\]
where $Q=\bX^T(\bX\bX^T)^{-1/2}\in\mathbb{R}^{n\times M}$. Direct leave-one-out analysis of $\bS^{-1}$ is complicated because the rank-$M$ centering term couples all observations.

We therefore express $\bS^{-1}$ in terms of the uncentered sample covariance matrix. Since
\[
\bOmega=n^{-1}\bZ\bZ^T
=\bS+\frac{1}{n}\bZ Q Q^T\bZ^T,
\]
the Woodbury matrix identity in Lemma~\ref{lemma:woodbury} gives
\[
\bS^{-1}
=\bOmega^{-1}
+\frac{1}{n}\bOmega^{-1}\bZ Q
\Big(I_M-n^{-1}Q^T\bZ^T\bOmega^{-1}\bZ Q\Big)^{-1}
Q^T\bZ^T\bOmega^{-1},
\]
provided that the displayed inverses exist. Their invertibility is addressed below.

It follows that
\begin{align*}
    \frac{1}{p} V_t =& \frac{1}{p} u_t^T \bZ^T \bOmega^{-1} \bZ u_t \\
    &+ \frac{1}{pn} u_t^T \bZ^T \bOmega^{-1} \bZ Q
    \Big(I_M - n^{-1} Q^T\bZ^T \bOmega^{-1} \bZ Q\Big)^{-1}
    Q^T \bZ^T \bOmega^{-1} \bZ u_t.
\end{align*}
Thus, it suffices to establish the joint asymptotic behavior of the two processes
\[  \Big\{\frac{1}{p} u_t^T \bZ^T \bOmega^{-1} \bZ u_t,\quad t\in \calT_{\rm mc}(\varepsilon)\Big\} \quad \text{and} \quad \Big\{\frac{1}{p} u_t^T \bZ^T \bOmega^{-1} \bZ Q,\quad t\in \calT_{\rm mc}(\varepsilon) \Big\},\]
together with the $M\times M$ random matrix $n^{-1} Q^T \bZ^T \bOmega^{-1}\bZ Q$, and to show that the smallest singular value of $I_M-n^{-1}Q^T\bZ^T\bOmega^{-1}\bZ Q$ is bounded away from zero with probability tending to one. Note that $p^{-1} u_t^T \bZ^T \bOmega^{-1} \bZ Q \in \mathbb{R}^M$. We shall deal with all components of it individually. Indeed, we shall establish the asymptotic normality of $\frac{1}{p}V_t$ after appropriate scaling. The asymptotic distribution of $\{p^{-1} V_t\}^{1/2}$ will directly follow from the delta-method.

We next verify the regularity of the deterministic vectors entering these forms. By the definition of $\bh_t$,
\[
\|\bh_t\|_2^2
=\tilde{\theta}^T(\calX_t\calX_t^T)^{-1}\tilde{\theta}
=\frac{1}{n}\tilde{\theta}^T(n^{-1}\calX_t\calX_t^T)^{-1}\tilde{\theta}.
\]
Condition~\ref{enum:full_rank} therefore implies, uniformly over $t\in\calT_{\rm mc}(\varepsilon)$,
\[
c_\varepsilon
\leq n\|\bh_t\|_2^2
\leq c_\varepsilon^{-1}\|\tilde{\theta}\|_2^2.
\]
Moreover, because $M$ is fixed and $\limsup_n\|\bX\|_{\max}<\infty$,
\[
\|\bh_t\|_\infty
\leq \calK n^{-1},
\qquad t\in\calT_{\rm mc}(\varepsilon).
\]
Indeed, every column of $\calX_t$ has uniformly bounded Euclidean norm, whereas Condition~\ref{enum:full_rank} gives $\|(\calX_t\calX_t^T)^{-1}\|_2\leq (nc_\varepsilon)^{-1}$.
These bounds yield
\[
\sup_{t\in\calT_{\rm mc}(\varepsilon)}\sqrt{n}\|u_t\|_\infty<\infty.
\]
The same design condition implies that the smallest eigenvalue of $n^{-1}\bX\bX^T$ is bounded away from zero: this follows, for example, by applying the condition to $t=(0,1/2,1)$ and summing the two diagonal blocks of $n^{-1}\calX_t\calX_t^T$. Together with the uniform entry-wise bound on $\bX$, this gives
\[
\sqrt{n}\|Q\|_{\max}<\infty.
\]

We introduce the following class of vectors.
\begin{definition}
    \label{def:regular_vectors}
    A deterministic vector $u \in \mathbb{R}^n$ is called regular if $\sqrt{n} \|u\|_\infty \leq \calK_0$. Here, $\calK_0$ is a fixed constant chosen to be sufficiently large and does not depend on $n$. 
\end{definition}
It follows that $u_t$ is regular for every $t$, uniformly over the scanning set, and that each column of $Q$ is regular. In addition, $\|u_t\|_2=1$ by construction and $Q^TQ=I_M$.
Therefore, for each $t$, $p^{-1} u_t^T \bZ^T \bOmega^{-1}\bZ u_t$ and each component of $p^{-1} u_t^T \bZ^T \bOmega^{-1} \bZ Q$ take the form
\[ p^{-1} u^T \bZ^T \bOmega^{-1} \bZ v,\]
where $u$ and $v$ are two regular vectors. 

We consider the invertibility of $\bOmega$. As shown in \citet{bai1993limit}, under Conditions~\ref{enum:moment_conditions} and \ref{enum:high_dimensional_regime},
\[ \lambda_{\min}(\bOmega) \stackrel{a.s.}{\longrightarrow} (1-\sqrt{\gamma})^2.\]
Thus, $\bOmega$ is invertible with probability tending to one. Almost-sure convergence alone, however, does not provide a sufficiently sharp lower-tail bound for $\lambda_{\min}(\bOmega)$ to control high moments of $\|\bOmega^{-1}\|_2$. We therefore use the standard truncation and regularization argument from random matrix theory described below.

Select a positive sequence $\kappa_n$ such that 
\begin{equation*}
\kappa_n \to 0 \qquad \mbox{and} \qquad \kappa_n^{-4} \mE[ z_{11}^4\mathbbm{1}(|z_{11}|\geq \kappa_nn^{1/2})]\to0.
\end{equation*}
The existence of $\kappa_n$ is shown in \citet{yin1988limit}. We then truncate $z_{ij}$ to be $z_{ij}\mathbbm{1}(|z_{ij}|\leq\kappa_nn^{1/2})$. The truncated variable is then standardized to maintain zero mean and unit variance. Namely, we construct 
\[
\frac{ {z}_{ij}\mathbbm{1}(|{z}_{ij}| \leq \kappa_n n^{1/2}) - \mE {z}_{ij}\mathbbm{1}(|{z}_{ij}| \leq \kappa_n n^{1/2})}{ \{\mE [{z}_{ij}\mathbbm{1}(|{z}_{ij}| \leq \kappa_n n^{1/2}) - \mE {z}_{ij}\mathbbm{1}(|{z}_{ij}| \leq \kappa_n n^{1/2})]^2\}^{1/2}}.
\]

To simplify notation, we continue to denote the modified variables by $z_{ij}$. After this modification, for some constant $\calK$ and all sufficiently large $n$,
\begin{equation} 
|z_{ij}| \leq \calK \kappa_n n^{1/2}, \qquad \mE [z_{ij}]=0, \qquad \mE [z_{ij}^2] = 1.
\label{eq:truncated_variable_condition}
\end{equation}
Moreover, for every fixed $m\geq2$, the truncated and standardized variables satisfy
\[\sup_{i,j}\mE|z_{ij}|^m\leq\calK_m.\] 
This follows from Condition~\ref{enum:moment_conditions} and the fact that the standardizing variance converges to one.
The truncation yields a polynomial lower-tail bound for the smallest eigenvalue of $\bOmega$. In particular, the following result is used in \citet{bai2004clt}.
\begin{lemma}
\label{lemma:poly_bound_tSigma}
Suppose the entries of $\bZ$ satisfy \eqref{eq:truncated_variable_condition}. For any positive $\ell$ and any $\mathfrak{D} \in  (0, (1-\sqrt{\gamma})^2)$,
\[
\mP ( \lambda_{\min}(\bOmega) < \mathfrak{D} ) = o(n^{-\ell}).
\]
\end{lemma}

Set $\mathfrak{D}=(1-\sqrt{\gamma})^2/2$. On the event $\{\lambda_{\min}(\bOmega)\geq\mathfrak{D}\}$, we have $\|\bOmega^{-1}\|_2\leq\mathfrak{D}^{-1}$. To obtain a deterministic bound on the complementary event, we add a small regularization term to $\bOmega$ and define
\[ \bA = \Big( \bOmega+ \rho_n I_p\Big)^{-1},\]
where $\rho_n=n^{-3/2}$. Then
\[ \|\bA\|_2 \leq  \mathfrak{D}^{-1} \mathbbm{1}(\lambda_{\min}(\bOmega) \geq \mathfrak{D}) + \rho_n^{-1}  \mathbbm{1}(\lambda_{\min}(\bOmega) < \mathfrak{D}). \]

The preceding bound and Lemma~\ref{lemma:poly_bound_tSigma} imply the following result.
\begin{lemma}\label{lemma:bound_moment_A}
Suppose that the entries of $\bZ$ satisfy \eqref{eq:truncated_variable_condition}. For every fixed $m\geq2$ and all sufficiently large $n$,
\[
\mE\|\bA\|_2^m
+\sup_{1\leq j\leq n}\mE\|\bA_j\|_2^m
+\sup_{1\leq i\ne j\leq n}\mE\|\bA_{ij}\|_2^m
\leq \calK_m.
\]
Moreover,
\[
|\beta_j|\leq1,\qquad |\beta_j^{\tr}|\leq1,
\qquad |\beta^{\mE}|\leq1,
\qquad 1\leq j\leq n.
\]
\end{lemma}

\begin{proof}
For any fixed $m\geq2$, choose $\ell>3m/2$ in Lemma~\ref{lemma:poly_bound_tSigma}. Since $\rho_n^{-m}=n^{3m/2}$,
\[
\mE\|\bA\|_2^m
\leq \mathfrak{D}^{-m}
+\rho_n^{-m}\mP\{\lambda_{\min}(\bOmega)<\mathfrak{D}\}
=O(1).
\]
The same argument applies after deleting one or two columns. Indeed, the corresponding nonzero data matrices have $n-1$ or $n-2$ columns, their aspect ratios still converge to $\gamma$, and the change from normalization by $n-r$ to normalization by $n$ is asymptotically negligible. Exchangeability makes the resulting bounds uniform in the deleted indices. Finally, all three resolvents are positive definite, so the quadratic forms and traces in the definitions of $\beta_j$ and $\beta_j^{\tr}$ are nonnegative. This proves the remaining assertions.
\end{proof}

We will repeatedly use the following conditional consequences of Lemma~\ref{lemma:bound_moment_A}. They make explicit how the deterministic-matrix bounds in the preceding lemmas are applied to leave-one-out resolvents.
\begin{lemma}\label{lemma:leave_out_moment_bounds}
Let $u$ and $v$ be regular vectors, and let $k\geq1$ and $m\geq2$ be fixed integers. Uniformly over $1\leq j\leq n$,
\begin{align*}
\mE\left|n^{-1}\{\bz_j^T\bA_j^k\bz_j-\tr(\bA_j^k)\}\right|^m
&\leq \calK_m n^{-m/2},\\
\mE\left|n^{-1/2}\bz_j^T\bA_j^k\bZ_j u\right|^m
&\leq \calK_m\|u\|_2^m,\\
\mE\left|n^{-1}u^T\bZ_j^T\bA_j^k\bZ_jv\right|^m
&\leq \calK_m\|u\|_2^m\|v\|_2^m.
\end{align*}
In addition,
\[\mE|\varrho_{qj}|^4\leq\calK\|u_q\|_2^4\|v_q\|_2^4.\]
For distinct $i$ and $j$, analogous bounds hold with $\bA_{ij}$ and $\bZ_{ij}$ in place of $\bA_j$ and $\bZ_j$, using either deleted column $\bz_i$ or $\bz_j$ as the independent vector.
\end{lemma}

\begin{proof}
The matrices $\bA_j$ and $\bZ_j$ depend only on columns other than $\bz_j$ and are therefore independent of $\bz_j$. Conditioning on $\bZ_j$ and applying Lemmas~\ref{lemma:Burkholder} and \ref{lemma:bound_moment_A} give the first bound. For the second bound, a conditional moment inequality for linear forms and the inequality
\[
n^{-1}\|\bA_j^k\bZ_ju\|_2^2
\leq \|\bA_j\|_2^{2k-1}\big\| n^{-1/2} \bA_j^{1/2} \bZ_j \big\|_2^2 \|u\|_2^2 \leq \|\bA_j\|_2^{2k-1}\|u\|_2^2 
\]
give the asserted rate. Here, we are using the fact that  the eigenvalues of $n^{-1} \bZ_j^T \bA_j \bZ_j$ are of the form $x(x+\rho_n)^{-1}$, $x>0$. The third bound follows from similar arguments.
Finally, conditional on $\bZ_j$, $\varrho_{qj}$ is a centered quadratic form whose coefficient matrix has rank one and spectral norm bounded by
\[
n^{-1}\|\bA_j\bZ_j u_q\|_2\|\bA_j\bZ_jv_q\|_2
\leq \|\bA_j\|_2\|u_q\|_2\|v_q\|_2.
\]
Another application of Lemma~\ref{lemma:Burkholder} proves the fourth-moment bound. The leave-two-out statements follow identically.
\end{proof}

The same conditioning argument, followed by H\"older's inequality, also yields the product bounds in Lemma~\ref{lemma:comprehensive} when its deterministic matrices are replaced by fixed powers of the appropriate leave-one-out resolvent. The corresponding bounds for conditional expectations of leave-one-out resolvents follow from Jensen's inequality.  For example,
\[ \mE \Big| n^{-1} \{\bz_j^T (\mE_j \bA_j^k) \bz_j - \tr(\mE_j \bA_j^k) \}  \Big|^m \leq \calK_m n^{-m/2},\]
since $\mE \|\mE_j \bA_j\|_2^k \leq \mE \mE_j \|\bA_j\|^k_2  = \mE \|\bA_j\|_2^k = O(1)$. Here, we use the fact that $\|\cdot\|_2^k$ is convex for $k\geq 1$.
We use these conditional versions below without further comment.

We next record two resolvent bounds used repeatedly below.
\begin{lemma}
\label{lemma:nonran_bound_resol}
(\citet{bai1998no}) 
For any $p\times p$ matrix $\bD$,
\[\Big| \tr\{\bA\bD - \bA_{j}\bD\} \Big| \leq \|\bA_j\|_2 \|\bD\|_2. \] 
\end{lemma}
Indeed, the Woodbury matrix identity gives
\begin{equation}\label{eq:woodbury_leave_one_out}
\bA  = \bA_j - \frac{n^{-1} \bA_j \bz_j \bz_j^T \bA_j }{1+ n^{-1} \bz_j^T \bA_j  \bz_j} = \bA_j  -\frac{1}{n} \beta_j  \bA_j  \bz_j \bz_j^T \bA_j .
\end{equation}
Since $\bA_j-\bA$ is positive semidefinite, trace-norm duality gives
\[
\left|\tr\{(\bA-\bA_j)\bD\}\right|
\leq \|\bD\|_2\tr(\bA_j-\bA).
\]
Moreover,
\[
n^{-1}\beta_j\bz_j^T\bA_j^2\bz_j
\leq \|\bA_j\|_2 n^{-1}\beta_j\bz_j^T\bA_j\bz_j
\leq\|\bA_j\|_2.
\]

\begin{lemma}
\label{lemma:converg_A}
For every fixed $m\geq 2$,
\[
n^{-m}\mE\big|\tr\bA-\mE\tr\bA\big|^m=O(n^{-m/2}).
\]
The same bound holds uniformly with $\bA$ replaced by $\bA_j$ or $\bA_{ij}$.
\end{lemma}

\begin{proof}
Write $\mE_\ell[\cdot]=\mE[\cdot\mid\sigma(\bz_1,\ldots,\bz_\ell)]$. Since $\bA_\ell$ does not depend on $\bz_\ell$,
\[
\tr\bA-\mE\tr\bA
=\sum_{\ell=1}^n(\mE_\ell-\mE_{\ell-1})
\tr(\bA-\bA_\ell).
\]
Lemma~\ref{lemma:nonran_bound_resol} and Lemma~\ref{lemma:bound_moment_A} give uniform bounds for every fixed moment of the summands. Lemma~\ref{lemma:Burkholder_martingale_moments} and Lemma \ref{lemma:concentration_quadrat} therefore yields
\[
\mE|\tr\bA-\mE\tr\bA|^m=O(n^{m/2}).
\]
For $\bA_j$, subtract $\bA_{\ell j}$ in the $\ell$th martingale increment; the increment with $\ell=j$ vanishes. The proof for $\bA_{ij}$ is identical. This argument is also a direct adaptation of the trace-concentration arguments in \citet{bai1998no} and \citet{Pan2011central}.
\end{proof}

In the following sections, unless stated otherwise, we work under the truncation condition \eqref{eq:truncated_variable_condition}. Section~\ref{subsec:truncation} will show formally that truncating the variables and replacing $\bOmega^{-1}$ by $\bA$ do not alter the weak limits of the processes and random matrices above.

Recall that
\[
H(u,v)
=p^{-1}u^T\bZ^T\bA\bZ v
=p^{-1}\sum_{i,j=1}^n u_i v_j\bz_i^T\bA\bz_j,
\]
where $u$ and $v$ are regular vectors. Consequently, once the $M\times M$ inverse in the Woodbury correction is well-defined, the regularized Woodbury expression for $p^{-1}V_t$ is a smooth function of finitely many quantities of the form $H(u,v)$, with $u$ and $v$ chosen from $u_t$ and the columns of $Q$.

The key step in the proof is therefore to analyze the following superposition of such bilinear forms:
\[
\calH=\calH(\{u_q\},\{v_q\})=\sum_{q=1}^m H(u_q,v_q)
=p^{-1}\sum_{q=1}^m u_q^T\bZ^T\bA\bZ v_q,
\]
where $m$ is fixed and $\{u_q\}_{q=1}^m$ and $\{v_q\}_{q=1}^m$ are two collections of regular vectors. The superposition arises naturally in the analysis of the finite-dimensional distribution of the $\{p^{-1} u_t^T \bZ^T \bA \bZ u_t \}$ and $\{p^{-1} u_t^T \bZ^T \bA \bZ Q\}$ and an increment when $u_t$ is shifted to $u_{t'}$.

Subsection~\ref{subsec:martingale_construct} constructs and simplifies a martingale decomposition of $\calH-\mE\calH$. Subsection~\ref{subsec:convergence_mean} studies its mean.  Subsection~\ref{subsec:variance_martingale} analyzes the conditional variance.  Subsection~\ref{subsec:increment_bound} analyzes stochastic increments. Subsection~\ref{subsec:asymptotic_normality_calH} establishes the asymptotic normality. Subsection~\ref{subsec:weak_convergence_process} discusses the weak convergence of the processes. Subsection~\ref{subsec:truncation} demonstrates that the variable truncation and regularization will not alter the weak limit.

\subsection{Construction and simplification of the martingale decomposition}
\label{subsec:martingale_construct}

The goal of this subsection is to derive a simplified martingale decomposition of
$\sqrt{p}(\calH-\mE\calH)$.
Specifically, we show that
\[
\sqrt{p}\bigl(\calH-\mE\calH\bigr)=
\sum_{j=1}^n
\sum_{q=1}^m
\beta^{\mE}
\Bigl[
\calH_{qj}^{(1)}
+
\calH_{qj}^{(2)}
+
\calH_{qj}^{(3)}
\Bigr]
+ O_{L_2}\Big( \sum_{q=1}^m  r_{q} \Big)
\]
where
\begin{align*}
   \calH_{qj}^{(1)}  &= \mE_j [p^{-1/2} u_{qj} \bz_j^T \bA_j  \bZ_j v_q], \\
   \calH_{qj}^{(2)}  &= \mE_j [ p^{-1/2} v_{qj} \bz_j^T \bA_j  \bZ_j u_q],\\
   \calH_{qj}^{(3)}  &= - \mE_j [ p^{-1/2} \varrho_{qj}] = - p^{-1/2}\mE_j\Big[\frac{1}{n} \bz_j^T \bA_j  \bZ_ju_q v_q^T \bZ_j^T \bA_j  \bz_j - \frac{1}{n} u_q^T\bZ_j^T \bA_j^{2} \bZ_jv_q\Big],
\end{align*} 
and
\[ r_{q} = \max\left\{ n^{-1/4} \|u_q\|_2\|v_q\|_2,~~~  n^{1/4} \sqrt{\sum_{j=1}^n |u_{qj}|^2|v_{qj}|^2}\right\}.\]
For sequences of regular vectors $u_q$ and $v_q$, we have $\|u_q\|_2+\|v_q\|_2=O(1)$ and
\[
\sum_{j=1}^n|u_{qj}|^2|v_{qj}|^2
\leq \|u_q\|_\infty^2\|v_q\|_2^2=O(n^{-1}).
\]
Consequently, $r_q=O(n^{-1/4})$.

Each $\calH_{qj}^{(k)}$ is a martingale difference with respect to $\sigma(\bz_1,\ldots,\bz_j)$. For $k=1,2$, this follows by conditioning first on all columns except $\bz_j$ and using $\mE\bz_j=0$. For $k=3$, the same conditioning gives
\[
\mE(\varrho_{qj}\mid\bZ_j)=0,
\]
because the second term in $\varrho_{qj}$ is the conditional trace of the first. Hence $\mE_{j-1}\calH_{qj}^{(k)}=0$ for $k=1,2,3$.


For each $q$, the martingale telescoping identity gives
\begin{align*}
&\sqrt{p} \Big( H(u_q, v_q) - \mE  H (u_q, v_q)\Big) \\
& = \sqrt{p} \sum\limits_{j=1}^n \Big\{\mE_j[H (u_q, v_q)] -\mE_{j-1}[H (u_q,v_q)]\Big\}\\
& = \sqrt{p} \sum\limits_{j=1}^n  p^{-1}\Big\{\mE_j [  u_q^T \bZ^T \bA \bZ v_q   -  u_q^T \bZ_j^T \bA_j  \bZ_jv_q ]
  -\mE_{j-1}[u_q^T \bZ^T \bA  \bZ v_q - u_q^T \bZ_j^T \bA_j  \bZ_j v_q ]\Big\}\\
& = \sqrt{p} \sum\limits_{j=1}^n (\mE_j-\mE_{j-1}) [d_{qj1} + d_{qj2} + d_{qj3}].
\end{align*}
The second equality holds because $u_q^T\bZ_j^T\bA_j\bZ_jv_q$ does not depend on $\bz_j$, and hence its $\mE_j$ and $\mE_{j-1}$ conditional expectations coincide.
Here,
\begin{align*}
d_{qj1}&= p^{-1} u_q^T(\bZ - \bZ_j)^T\bA \bZ v_q, \\
d_{qj2} &=p^{-1} u_q^T\bZ_j^T(\bA -\bA_j )\bZ v_q,\\
d_{qj3}&=p^{-1} u_q^T\bZ_j^T\bA_j (\bZ-\bZ_j) v_q = p^{-1} v_{qj}\bz_j^T \bA_j  \bZ_j u_q. 
\end{align*}

Using the Woodbury identity in \eqref{eq:woodbury_leave_one_out}, we further decompose $d_{qj1}$ and $d_{qj2}$ as follows.
\begin{align*}
    d_{qj1} &= \frac{1}{p} u_{qj} \bz_j^T \bA_j  \bZ_j v_q + \frac{1}{p} u_{qj} v_{qj} \bz_j^T \bA_j  \bz_j \\
    &\quad - \frac{1}{pn} u_{qj} \Big(\bz_j^T \bA_j  \bz_j\Big) \Big(\bz_j^T \bA_j  \bZ_j v_{q}\Big) \beta_j  - \frac{1}{pn} u_{qj} v_{qj} \Big(\bz_j^T \bA_j \bz_j\Big)^2 \beta_j .\\
    & = d_{qj1}^{(1)} + d_{qj1}^{(2)} + d_{qj1}^{(3)} + d_{qj1}^{(4)}, \quad \text{say}.\\
    d_{qj2} & = - \frac{1}{pn} v_{qj} \Big(\bz_j^T \bA_j  \bz_j\Big) \Big(\bz_j^T \bA_j  \bZ_j u_{q}\Big) \beta_j - \frac{1}{pn} \Big(\bz_j^T \bA_j  \bZ_j v_q\Big) \Big(\bz_j^T \bA_j \bZ_j u_q\Big) \beta_j \\
            & = \frac{1}{p}v_{qj} [\beta_j -1] \bz_j^T \bA_j  \bZ_j u_q - \frac{1}{p} \varrho_{qj}\beta_j  - \frac{1}{pn} u_q^T \bZ_j^T \bA_j^2  \bZ_jv_q \beta_j \\
            & = d_{qj2}^{(1)} + d_{qj2}^{(2)} + d_{qj2}^{(3)}, \quad \text{say}. 
\end{align*}

We claim that 
\begin{align}
   p^{1/2} \sum\limits_{j=1}^n &(\mE_j - \mE_{j-1}) d_{qj1} = \beta^{\mE}  \sum_{j=1}^n \calH_{qj}^{(1)} + O_{L_2}(r_q). \label{eq:simplification_diqj1} \\
   p^{1/2} \sum\limits_{j=1}^n &(\mE_j - \mE_{j-1}) d_{qj2} = (\beta^{\mE}  -1) \sum_{j=1}^n \calH_{qj}^{(2)} + \beta^{\mE}  \sum_{j=1}^n \calH_{qj}^{(3)} + O_{L_2}(r_q).\label{eq:simplification_diqj2}
\end{align}
Combining with the form of $d_{qj3}$, the claimed martingale decomposition of $\sqrt{p}(\calH- \mE\calH)$ will follow. 


We prove \eqref{eq:simplification_diqj1} by treating the four terms in the decomposition of $d_{qj1}$ separately. For $d_{qj1}^{(1)}$,
\begin{align*}
     &p^{1/2} \sum_{j=1}^n (\mE_j - \mE_{j-1}) d_{qj1}^{(1)} =  p^{1/2} \sum_{j=1}^n (\mE_j - \mE_{j-1}) \frac{1}{p} u_{qj} \bz_j^T \bA_j  \bZ_j v_q \\
     &= p^{1/2} \sum_{j=1}^n \mE_j \Big[\frac{1}{p} u_{qj} \bz_j^T \bA_j  \bZ_j v_q\Big]= \sum_{j=1}^n \calH_{qj}^{(1)}.
\end{align*}

For $d_{qj1}^{(2)}$, Lemmas~\ref{lemma:Burkholder_martingale_moments}, \ref{lemma:Burkholder}, and \ref{lemma:bound_moment_A} give
\begin{align*}
   & \mE \Big|p^{1/2} \sum_{j=1}^n (\mE_j -\mE_{j-1}) d_{qj1}^{(2)}\Big|^2 \\
  & = \mE \Big|p^{1/2} \sum_{j=1}^n (\mE_j -\mE_{j-1})  \frac{1}{p} u_{qj} v_{qj} \bz_j^T \bA_j  \bz_j \Big|^2\\
  &\leq \calK p^{-1} \sum_{j=1}^n |u_{qj} v_{qj}|^2  \mE \Big| (\mE_j -\mE_{j-1}) \bz_j^T \bA_j  \bz_j  \Big|^2\\
  &  \leq\calK p^{-1} \sum_{j=1}^n |u_{qj}|^2|v_{qj}|^2 \mE\Big|\bz_j^T (\mE_j\bA_{j} )\bz_j - \tr(\mE_j\bA_j )  \Big|^2 \\
  & \leq \calK p^{-1} \sum_{j=1}^n |u_{qj}|^2|v_{qj}|^2 p \mE \|\bA_j \|^2_2\leq \calK \sum_{j=1}^n |u_{qj}|^2 |v_{qj}|^2.
\end{align*}

For $d_{qj1}^{(3)}$, we replace $\beta_j$ and $n^{-1}\bz_j^T\bA_j\bz_j$ successively by $\beta_j^{\tr}$ and $n^{-1}\tr(\bA_j)$, and control the resulting errors. By definition,
\[ \beta_j  - \beta_j^{\tr}  = - \beta_j  \beta_j^{\tr} \theta_j . \]
\begin{align*}
    &\frac{1}{pn} u_{qj} \Big(\bz_j^T \bA_j  \bz_j\Big) \Big(\bz_j^T \bA_j  \bZ_j v_{q}\Big) \beta_j  \\
    &=  \frac{1}{pn} u_{qj} \Big(\bz_j^T \bA_j  \bz_j\Big) \Big(\bz_j^T \bA_j  \bZ_j v_{q}\Big) \beta^{\tr}_j  -\frac{1}{pn} u_{qj} \Big(\bz_j^T \bA_j  \bz_j\Big) \Big(\bz_j^T \bA_j  \bZ_j v_q \Big) \beta_j  \beta_j^{\tr}  \theta_j  \\ 
    & = \frac{1}{p} u_{qj} \theta_j  \Big(\bz_j^T \bA_j  \bZ_j v_{q}\Big) \beta^{\tr}_j  + \frac{1}{pn} u_{qj} \tr(\bA_j ) \Big(\bz_j^T \bA_j  \bZ_j v_{q}\Big) \beta^{\tr}_j \\
    & \quad-\frac{1}{pn} u_{qj} \Big(\bz_j^T \bA_j  \bz_j\Big) \Big(\bz_j^T \bA_j  \bZ_j v_q \Big) \beta_j  \beta_j^{\tr}  \theta_j \\
    & = \frac{1}{p} u_{qj} \theta_j  \Big(\bz_j^T \bA_j  \bZ_j v_{q}\Big) \beta^{\tr}_j  + \frac{1}{p} u_{qj} \Big(\bz_j^T \bA_j \bZ_j v_{q} \Big)(1-\beta_j^{\tr} )\\
    &\quad-\frac{1}{pn} u_{qj} \Big(\bz_j^T \bA_j  \bz_j\Big) \Big(\bz_j^T \bA_j  \bZ_j v_q \Big) \beta_j  \beta_j^{\tr}  \theta_j .
\end{align*}

By Lemmas~\ref{lemma:Burkholder_martingale_moments}, \ref{lemma:comprehensive}, and \ref{lemma:bound_moment_A}, the first and third terms satisfy

\begin{align*}
 &\mE \Big| p^{1/2} \sum_{j=1}^n (\mE_j -\mE_{j-1}) \frac{1}{p} u_{qj} \theta_j  \Big(\bz_j^T \bA_j  \bZ_j v_{q}\Big) \beta^{\tr}_j  \Big|^2\\
 &\leq \calK  \sum_{j=1}^n |u_{qj}|^2 \mE \Big| \theta_j  \Big(\frac{1}{\sqrt{p}} \bz_j^T \bA_j  \bZ_j v_q \Big) \Big|^2 \leq \calK n^{-1/2}\|u_q\|_2^2 \|v_q\|_2^2. \\ 
 &\mE \Big|p^{1/2}\sum_{j=1}^n (\mE_j -\mE_{j-1}) \frac{1}{\sqrt{p}} u_{qj} \Big(\frac{1}{n} \bz_j^T \bA_j  \bz_j\Big) \Big(\frac{1}{\sqrt{p}} \bz_j^T \bA_j  \bZ_j v_q \Big) \beta_j  \beta_j^{\tr}  \theta_j \Big|^2\\
 &\leq \calK n^{-1/2}\|u_q\|_2^2 \|v_q\|_2^2. 
\end{align*}
Here, both products
\[\theta_j \Big(\frac{1}{n} \bz_j^T \bA_j  \bz_j\Big) \Big(\frac{1}{\sqrt{p}} \bz_j^T \bA_j  \bZ_j v_q \Big)\quad \text{and}\quad \theta_j  \Big(\frac{1}{\sqrt{p}} \bz_j^T \bA_j  \bZ_j v_q \Big)\]
have the form covered by the conditional version of Lemma~\ref{lemma:comprehensive}.

It follows that
\begin{align*}
&p^{1/2} \sum_{j=1}^n (\mE_j -\mE_{j-1}) d_{qj1}^{(3)} = p^{1/2} \sum_{j=1}^n \mE_j\left[\frac{1}{p}u_{qj} \Big(\bz_j^T \bA_j  \bZ_j v_q\Big)(\beta^{\tr}_j -1)\right]\\
& \qquad + O_{L_2}(n^{-1/4} \|u_q\|_2 \|v_q\|_2) \\
& =  - \sum_{j=1}^{n} \calH_{qj}^{(1)} + p^{-1/2} \sum_{j=1}^n \mE_j\left[u_{qj} \Big(\bz_j^T \bA_j  \bZ_j v_q\Big)\beta^{\tr}_j\right]  + O_{L_2}(r_q) 
\end{align*}

For $d_{qj1}^{(4)}$,
\begin{align*}
   &-  p^{1/2} \sum_{j=1}^n (\mE_j -\mE_{j-1}) d_{qj1}^{(4)} \\
   & = p^{1/2} \sum_{j=1}^n (\mE_j -\mE_{j-1}) \frac{1}{pn} u_{qj}v_{qj} \Big(\bz_j^T \bA_j  \bz_j\Big)^2 \beta_j \\
   & = p^{1/2} \sum_{j=1}^n (\mE_j -\mE_{j-1}) \frac{1}{pn} u_{qj}v_{qj} \Big(\bz_j^T \bA_j  \bz_j\Big)^2 \beta^{\tr}_j \\
   &\quad -  p^{1/2} \sum_{j=1}^n (\mE_j -\mE_{j-1}) \frac{1}{pn} u_{qj}v_{qj} \Big(\bz_j^T \bA_j  \bz_j\Big)^2 \beta^{\tr}_j \beta_j \theta_j \\
   & = p^{1/2} \sum_{j=1}^n (\mE_j - \mE_{j-1}) \frac{n}{p} u_{qj}v_{qj} \theta^2_j  \beta_j^{\tr} \\
   &\qquad + 2p^{1/2} \sum_{j=1}^n (\mE_j -\mE_{j-1}) \frac{n}{p} u_{qj}v_{qj} \theta_j  \frac{1}{n}\tr(\bA_j )  \beta_j^{\tr} \\
   &\qquad -  p^{1/2} \sum_{j=1}^n (\mE_j -\mE_{j-1}) \frac{n}{p} u_{qj}v_{qj} \theta_j   \Big(n^{-1}\bz_j^T \bA_j  \bz_j\Big)^2 \beta^{\tr}_j \beta_j 
\end{align*}

By Lemma \ref{lemma:concentration_quadrat} and Lemma \ref{lemma:comprehensive}, the terms on the right-hand side are such that 
\begin{align*}
    &\mE \Big|p^{1/2} \sum_{j=1}^n (\mE_j - \mE_{j-1}) u_{qj} v_{qj} \theta_j^2  \beta_j^{\tr}   \Big|^2 \leq \calK p \sum_{j=1}^n |u_{qj}|^2 |v_{qj}|^2 \mE |\theta_j |^4\\
    & \qquad \leq \calK \sum_{j=1}^n |u_{qj}|^2 |v_{qj}|^2 n^{-1}.  \\ 
    & \mE \Big| 2p^{1/2} \sum_{j=1}^n (\mE_j -\mE_{j-1}) \frac{n}{p} u_{qj}v_{qj} \theta_j  \frac{1}{n}\tr(\bA_j )  \beta_j^{\tr} \Big|^2 \\
    &\qquad \leq\calK p \sum_{j=1}^n |u_{qj}|^2 |v_{qj}|^2 \mE |\theta_j |^2  \leq \calK \sum_{j=1}^n |u_{qj}|^2|v_{qj}|^2.\\ 
    & \mE \Big|  p^{1/2} \sum_{j=1}^n (\mE_j -\mE_{j-1}) \frac{n}{p} u_{qj}v_{qj} \theta_j   \Big(n^{-1}\bz_j^T \bA_j  \bz_j\Big)^2 \beta^{\tr}_j \beta_j  \Big|^2\\
    &\qquad \leq \calK p \sum_{j=1}^n |u_{qj}|^2 |v_{qj}|^2 \mE \Big[ |\theta_j |^2 \Big|n^{-1} \bz_j^T \bA_j \bz_j\Big|^4  \Big]\\
    &\qquad \leq \calK \sqrt{n} \sum_{j=1}^n |u_{qj}|^2|v_{qj}|^2. 
\end{align*}
We conclude that
\[\sqrt{p} \sum_{j=1}^n (\mE_j -\mE_{j-1}) d_{qj1}^{(4)} = O_{L_2}(r_q). \]

Combining the four bounds gives
\[p^{1/2} \sum\limits_{j=1}^n (\mE_j - \mE_{j-1}) d_{qj1} = p^{-1/2} \sum_{j=1}^n \mE_j\left[u_{qj}\bz_j^T \bA_j  \bZ_j v_{q}  \beta^{\tr}_j\right]  + O_{L_2}(r_q).\]
To show \eqref{eq:simplification_diqj1}, we only need to replace $\beta^{\tr}_j $ with $\beta^{\mE} $ and show that the induced difference is $O_{L_2}(r_q)$. To this end, observe that 
\[ \beta^{\tr}_j  - \beta^{\mE}  = \beta_j^{\tr}  \beta^{\mE}  (\frac{1}{n} \mE \tr(\bA_j ) - \frac{1}{n}\tr(\bA_j )).  \]
Here, exchangeability gives $\mE\tr(\bA_j)=\mE\tr(\bA_1)$. Lemmas~\ref{lemma:converg_A} and \ref{lemma:leave_out_moment_bounds} imply
\begin{align*} 
&\mE \Big| p^{-1/2} \sum_{j=1}^n u_{qj} \mE_j\left[\bz_j^T \bA_j  \bZ_j v_q \beta^{\tr}_j \beta^{\mE} (\frac{1}{n} \mE \tr(\bA_j ) - \frac{1}{n}\tr(\bA_j ))\right]  \Big|^2 \\
&\leq \calK \sum_{j=1}^n |u_{qj}|^2 \mE\left[\Big|n^{-1/2} \bz_j^T \bA_j \bZ_j v_q \Big|^2 \Big|\frac{1}{n}\tr(\bA_j ) - \frac{1}{n} \mE \tr(\bA_j ) \Big|^2\right]\\
&\leq \calK n^{-1} \|u_q\|_2^2 \|v_q\|_2^2.
\end{align*}
The claim in \eqref{eq:simplification_diqj1} follows. 

We next prove \eqref{eq:simplification_diqj2} by treating $d_{qj2}^{(1)}$, $d_{qj2}^{(2)}$, and $d_{qj2}^{(3)}$ separately. First,
\[
\beta_j-\beta^{\mE}
=-\beta_j\beta^{\mE}\left\{\theta_j+\frac{1}{n}\tr(\bA_j)-\frac{1}{n}\mE\tr(\bA_j)\right\}.
\]
It follows that
\begin{align*}
&p^{1/2}\sum_{j=1}^n(\mE_j-\mE_{j-1})d_{qj2}^{(1)}\\
&=p^{1/2}\sum_{j=1}^n(\mE_j-\mE_{j-1})
\frac{1}{p}v_{qj}(\beta^{\mE}-1)\bz_j^T\bA_j\bZ_ju_q\\
&\quad-p^{1/2}\sum_{j=1}^n(\mE_j-\mE_{j-1})
\frac{1}{p}v_{qj}\beta_j\beta^{\mE}\theta_j\bz_j^T\bA_j\bZ_ju_q\\
&\quad-p^{1/2}\sum_{j=1}^n(\mE_j-\mE_{j-1})
\frac{1}{p}v_{qj}\beta_j\beta^{\mE}
\left\{\frac{1}{n}\tr(\bA_j)-\frac{1}{n}\mE\tr(\bA_j)\right\}
\bz_j^T\bA_j\bZ_ju_q.
\end{align*}
Lemmas~\ref{lemma:converg_A}, \ref{lemma:comprehensive}, and \ref{lemma:leave_out_moment_bounds}, together with martingale orthogonality, show that the squared $L_2$ norms of the second and third terms in the preceding display are both bounded by
\[
\calK n^{-1/2}\|u_q\|_2^2\|v_q\|_2^2.
\]
Therefore,
\[
\sqrt{p}\sum_{j=1}^n(\mE_j-\mE_{j-1})d_{qj2}^{(1)}
=(\beta^{\mE}-1)\sum_{j=1}^n\calH_{qj}^{(2)}+O_{L_2}(r_q).
\]

For $d_{qj2}^{(2)}$,
\begin{align*}
\sqrt{p}\sum_{j=1}^n(\mE_j-\mE_{j-1})d_{qj2}^{(2)}
&=-\beta^{\mE}\sum_{j=1}^n(\mE_j-\mE_{j-1})p^{-1/2}\varrho_{qj}\\
&\quad-p^{-1/2}\sum_{j=1}^n(\mE_j-\mE_{j-1})
\{\varrho_{qj}(\beta_j-\beta^{\mE})\}.
\end{align*}
The squared $L_2$ norm of the second term is bounded by
\begin{align*}
\calK p^{-1}\sum_{j=1}^n
\mE\left[|\varrho_{qj}|^2
\left\{|\theta_j|^2+\left|\frac{1}{n}\tr(\bA_j)-\frac{1}{n}\mE\tr(\bA_j)\right|^2\right\}\right].
\end{align*}
By Lemmas~\ref{lemma:leave_out_moment_bounds} and \ref{lemma:converg_A}, Cauchy--Schwarz yields, uniformly in $j$,
\[
\mE|\varrho_{qj}\theta_j|^2
+\mE\left[|\varrho_{qj}|^2
\left|\frac{1}{n}\tr(\bA_j)-\frac{1}{n}\mE\tr(\bA_j)\right|^2\right]
\leq \calK n^{-1}\|u_q\|_2^2\|v_q\|_2^2.
\]
Consequently,
\[
\sqrt{p}\sum_{j=1}^n(\mE_j-\mE_{j-1})d_{qj2}^{(2)}
=\beta^{\mE}\sum_{j=1}^n\calH_{qj}^{(3)}+O_{L_2}(r_q).
\]

Finally,
\begin{align*}
&p^{1/2}\sum_{j=1}^n(\mE_j-\mE_{j-1})d_{qj2}^{(3)}=-p^{1/2}\sum_{j=1}^n(\mE_j-\mE_{j-1})
\frac{1}{pn}u_q^T\bZ_j^T\bA_j^2\bZ_jv_q\beta_j^{\tr}\\
&\quad+p^{1/2}\sum_{j=1}^n(\mE_j-\mE_{j-1})
\frac{1}{pn}u_q^T\bZ_j^T\bA_j^2\bZ_jv_q\beta_j\beta_j^{\tr}\theta_j\\
&=p^{1/2}\sum_{j=1}^n(\mE_j-\mE_{j-1})
\frac{1}{pn}u_q^T\bZ_j^T\bA_j^2\bZ_jv_q\beta_j\beta_j^{\tr}\theta_j.
\end{align*}
The last equality holds because the first summand does not depend on $\bz_j$. Lemma~\ref{lemma:leave_out_moment_bounds} and Cauchy--Schwarz give
\begin{align*}
&\mE\left|p^{1/2}\sum_{j=1}^n(\mE_j-\mE_{j-1})
\frac{1}{pn}u_q^T\bZ_j^T\bA_j^2\bZ_jv_q\beta_j\beta_j^{\tr}\theta_j\right|^2\leq \calK n^{-1/2}\|u_q\|_2^2\|v_q\|_2^2.
\end{align*}
Thus,
\[
\sqrt{p}\sum_{j=1}^n(\mE_j-\mE_{j-1})d_{qj2}^{(3)}=O_{L_2}(r_q),
\]
which completes the proof of \eqref{eq:simplification_diqj2}.

Combining the bounds for $d_{qj1}$, $d_{qj2}$, and $d_{qj3}$ yields
\[
\sqrt{p}\bigl(\calH-\mE\calH\bigr)
=
\sum_{j=1}^n
\sum_{q=1}^m
\beta^{\mE} 
\Bigl[
\calH_{qj}^{(1)}
+
\calH_{qj}^{(2)}
+
\calH_{qj}^{(3)}
\Bigr]
+ O_{L_2}\Big( \sum_{q=1}^m  r_{q} \Big).
\]

\subsection{Convergence of the Mean}\label{subsec:convergence_mean}

We next establish the deterministic centering of $\calH$. For every pair of regular vectors $u_q$ and $v_q$, we show that
\begin{equation}\label{eq:mean_convergence_bound}
\sqrt{p}\left|\mE H(u_q,v_q)-u_q^Tv_q\right|
=o(\|u_q\|_2\|v_q\|_2).
\end{equation}

The identity $\bA\bz_j=\beta_j\bA_j\bz_j$ gives
\begin{align*}
\sqrt{p}\mE H(u_q,v_q)
&=p^{-1/2}\sum_{j=1}^n u_{qj}v_{qj}\mE(\bz_j^T\bA_j\bz_j\beta_j)+p^{-1/2}\sum_{j=1}^n u_{qj}\mE(\bz_j^T\bA_j\bZ_jv_q\beta_j).
\end{align*}
Using
\[
\beta_j=\beta^{\mE}-\beta_j\beta^{\mE}
\left\{\frac{1}{n}\bz_j^T\bA_j\bz_j-\frac{1}{n}\mE\tr(\bA_j)\right\},
\]
and conditioning on $\bZ_j$ where appropriate, we obtain
\begin{align*}
\sqrt{p}\mE H(u_q,v_q)
&=\sqrt{p}\,u_q^Tv_q\frac{p^{-1}\mE\tr(\bA_1)}{1+n^{-1}\mE\tr(\bA_1)}\\
&\quad+np^{-1/2}u_q^Tv_q\beta^{\mE}
\mE\left[\beta_1\left\{\frac{1}{n}\bz_1^T\bA_1\bz_1-\frac{1}{n}\mE\tr(\bA_1)\right\}\right]\\
&\quad-p^{-1/2}(\beta^{\mE})^2\sum_{j=1}^n u_{qj}
\mE\left[\bz_j^T\bA_j\bZ_jv_q
\left\{\frac{1}{n}\bz_j^T\bA_j\bz_j-\frac{1}{n}\mE\tr(\bA_j)\right\}\right]\\
&\quad+p^{-1/2}(\beta^{\mE})^2\sum_{j=1}^n u_{qj}
\mE\left[\bz_j^T\bA_j\bZ_jv_q\beta_j
\left\{\frac{1}{n}\bz_j^T\bA_j\bz_j-\frac{1}{n}\mE\tr(\bA_j)\right\}^2\right]\\
&=:D_1+D_2+D_3+D_4.
\end{align*}

We first control $D_2$ and $D_4$. Since the centered quantity in braces has mean zero,
\[
D_2=-np^{-1/2}u_q^Tv_q(\beta^{\mE})^2
\mE\left[\beta_1\left\{\frac{1}{n}\bz_1^T\bA_1\bz_1-\frac{1}{n}\mE\tr(\bA_1)\right\}^2\right].
\]
Lemmas~\ref{lemma:converg_A} and \ref{lemma:leave_out_moment_bounds} imply, uniformly in $j$,
\begin{align}
\mE\left|\frac{1}{n}\bz_j^T\bA_j\bz_j-\frac{1}{n}\mE\tr(\bA_j)\right|^4
&=O(n^{-2}),\nonumber\\
\mE\left|p^{-1/2}\bz_j^T\bA_j\bZ_jv_q\beta_j\right|^2
&=O(\|v_q\|_2^2).\label{eq:z1_A_Z_b_beta}
\end{align}
Consequently, Cauchy--Schwarz and $\sum_j|u_{qj}|\leq\sqrt{n}\|u_q\|_2$ give
\[
|D_2|+|D_4|=O(n^{-1/2}\|u_q\|_2\|v_q\|_2).
\]

It remains to control $D_3$. Conditional on $\bZ_j$, Lemma~\ref{lemma:z1_z1_z1_a} gives
\begin{align*}
&\mE\left[\bz_j^T\bA_j\bZ_jv_q
\left\{\frac{1}{n}\bz_j^T\bA_j\bz_j-\frac{1}{n}\mE\tr(\bA_j)\right\}\right]=\frac{\mE z_{11}^3}{n}\sum_{i=1}^p
\mE\left[(e_i^T\bA_je_i)(e_i^T\bA_j\bZ_jv_q)\right].
\end{align*}
We use the following concentration bounds:
\begin{align}
\sup_{1\leq i,i'\leq p}\sup_{1\leq j\leq n}
\mE\left|e_i^T\bA_je_{i'}-\mE(e_i^T\bA_je_{i'})\right|^2
&=O(n^{-1}),\label{eq:concentrate_diag}\\
\sup_{1\leq i\leq p}\sup_{1\leq j\leq n}
\mE\left|p^{-1/2}e_i^T\bA_j\bZ_jv_q-
\mE(p^{-1/2}e_i^T\bA_j\bZ_jv_q)\right|^2
&=O(n^{-1}\|v_q\|_2^2).\label{eq:concentrate_A_zalpha_a}
\end{align}

Define
\[{\beta}_{\ell j}  = \frac{1}{1+n^{-1}\bz^T_\ell {\bA}_{\ell j} \bz_\ell }.\]

For \eqref{eq:concentrate_diag}, martingale orthogonality and the Woodbury identity give, for any $i,i'$,
\begin{align*}
&\mE \Big|e_i^T \bA_{j}  e_{i'} - \mE e_i^T \bA_j  e_{i'}\Big|^2 = \mE \Big| \sum\limits_{\ell\neq j}^{n} (\mE_{\ell} - \mE_{\ell-1}) [e_i^T (\bA_{j}-\bA_{\ell j})e_{i'}]\Big |^2\\
&\leq \frac{\calK}{n^2}\sum\limits_{\ell\neq j}^n \mE \Big| \bz_\ell^T \bA_{\ell j} e_{i'} e_i^T \bA_{\ell j}  \bz_\ell  \beta_{\ell j} \Big|^2.
\end{align*}
By Lemmas~\ref{lemma:bound_quad_rankone} and \ref{lemma:bound_moment_A}, conditioning on $\bZ_{\ell j}$ gives
\[\mE |\bz_\ell^T \bA_{\ell j} e_{i'} e_i^T \bA_{\ell j}  \bz_\ell|^2 \leq \calK \mE \|\bA_{\ell j} \|^4_2 = O(1). \]
This proves \eqref{eq:concentrate_diag}.

For \eqref{eq:concentrate_A_zalpha_a}, the same martingale argument yields
\begin{align*}
& p^{-1}\mE |e_i^T \bA_j \bZ_jv_{q} - \mE e_i^T \bA_{j} \bZ_j v_{q} |^2\\
&= p^{-1}\mE\Big|\sum\limits_{\ell\neq j}^n (\mE_{\ell}-\mE_{\ell-1}) [e_i^T \bA_j \bZ_jv_{q} - e_i^T \bA_{\ell j} \bZ_{j\ell} v_{q}] \Big|^2\\
& \leq \calK p^{-1}\sum\limits_{\ell\neq j}^n \mE \Big| v_{q\ell} \beta_{\ell j}  e_i^T \bA_{\ell j}  \bz_\ell - \frac{1}{n} e_i^T \bA_{\ell j}  \bz_\ell \bz_{\ell}^T \bA_{\ell j} \bZ_{j\ell}v_{q} \beta_{\ell j} \Big|^2\\ 
&\leq \calK p^{-1} \sum_{\ell\neq j}^n |v_{q\ell}|^2 \mE \Big| e_i^T \bA_{\ell j}   \bz_{\ell} \Big|^2  + \calK p^{-1} \sum_{\ell \neq j}^n \mE \Big|\frac{1}{n} e_i^T \bA_{\ell j}  \bz_\ell \bz_{\ell}^T \bA_{\ell j} \bZ_{j\ell}v_{q} \beta_{\ell j}  \Big|^2
\end{align*}
The leave-two-out moment bounds give
\begin{align*}
\sup_i \mE|e_i^T \bA_{\ell j}  \bz_\ell |^4 \leq \calK \mE \|\bA_{\ell j} \|^4 = O(1).  
\end{align*}
Also, 
\[\mE |p^{-1/2} e_i^T \bA_{\ell j}  \bz_\ell \bz_{\ell}^T \bA_{\ell j} \bZ_{j\ell}v_{q} \beta_{\ell j} |^2 =O(\|v_q\|_2^2).   \]
Since $p/n$ is bounded away from zero and infinity, substituting these two estimates into the preceding display proves \eqref{eq:concentrate_A_zalpha_a} with rate $O(n^{-1}\|v_q\|_2^2)$.

Substituting the preceding conditional identity into $D_3$ and applying \eqref{eq:concentrate_diag}, \eqref{eq:concentrate_A_zalpha_a}, and Cauchy--Schwarz yield
\begin{align*}
D_3
&=-n^{-1}p^{-1/2}(\beta^{\mE})^2\mE z_{11}^3
\sum_{j=1}^n u_{qj}\sum_{i=1}^p
\mE(e_i^T\bA_je_i)\mE(e_i^T\bA_j\bZ_jv_q)\\
&\quad+O(n^{-1/2}\|u_q\|_2\|v_q\|_2)\\
&=-n^{-1}p^{-1/2}(\beta^{\mE})^2\mE z_{11}^3
\sum_{j=1}^n u_{qj}\sum_{\ell\neq j}v_{q\ell}
\sum_{i=1}^p\mE(e_i^T\bA_1e_i)\mE(e_i^T\bA_1\bz_2)\\
&\quad+O(n^{-1/2}\|u_q\|_2\|v_q\|_2),
\end{align*}
where the last equality follows from exchangeability. We next evaluate the remaining expectation involving $\bz_2$.
\begin{align*}
&\mE e_i^T\bA_1  \bz_2=  \mE e_i^T\bA_{12} \bz_2 \beta_{21} \\
&= -\mE e_i^T\bA_{12} \bz_2 \beta_{21}  \beta^{\mE}  (\frac{1}{n}\bz_2^T\bA_{12} \bz_2 - \frac{1}{n} \mE\tr\bA_1  )\\
&=  -\mE e_i^T\bA_{12} \bz_2 (\beta^{\mE} )^2 (\frac{1}{n}\bz_2^T\bA_{12} \bz_2 - \frac{1}{n} \mE\tr\bA_1  ) \\
&~~ +\mE e_i^T\bA_{12} \bz_2 (\beta^{\mE} )^2\beta_{21}  \{\frac{1}{n}\bz_2^T\bA_{12} \bz_2 - \frac{1}{n} \mE\tr\bA_1  \}^2 \\ 
& = -\frac{1}{n} (\beta^{\mE} )^2 \mE z_{11}^3 \sum\limits_{\ell=1}^p \mE e_\ell^T \bA_{12}  e_\ell e_\ell^T \bA_{12}  e_i\\ 
&~~ +\mE e_i^T\bA_{12} \bz_2 (\beta^{\mE} )^2\beta_{21}  \{\frac{1}{n}\bz_2^T\bA_{12} \bz_2 - \frac{1}{n} \mE\tr\bA_1  \}^2.
\end{align*}
By Lemmas~\ref{lemma:concentration_quadrat}, \ref{lemma:bound_quad_rankone}, and \ref{lemma:bound_moment_A}, Cauchy--Schwarz gives
\[ \sup_i\mE |e_i^T\bA_{12} \bz_2 (\beta^{\mE} )^2\beta_{21}  \{\frac{1}{n}\bz_2^T\bA_{12} \bz_2 - \frac{1}{n} \mE\tr\bA_1  \}^2| = O(n^{-1}).\]
Let $\alpha_1$ and $\alpha_2$ denote the diagonals of $\mE\bA_1$ and $\mE\bA_{12}$, respectively. Substituting the preceding expansion and using \eqref{eq:concentrate_diag} once more give
\begin{align*}
D_3
&=n^{-2}p^{-1/2}(\beta^{\mE})^4(\mE z_{11}^3)^2
\sum_{j=1}^n u_{qj}\sum_{\ell\neq j}v_{q\ell}
\alpha_1^T\mE(\bA_{12})\alpha_2+O(n^{-1/2}\|u_q\|_2\|v_q\|_2).
\end{align*}
Indeed, the covariance terms introduced when factoring the expectations contribute only to the displayed remainder.
\begin{align*}
 |\alpha_1^T  \mE \bA_{12}   \alpha_2 | & \leq  \|\mE\bA_{12} \|_2 \|\alpha_1 \|_2 \|\alpha_2 \|_2 \leq p\|\mE\bA_{1}\|_2\|\mE\bA_{12}\|_2^2=O(n).
\end{align*}
Since $\|u_q\|_1\|v_q\|_1\leq n\|u_q\|_2\|v_q\|_2$, it follows that
\[|D_{3}| = O(n^{-1/2}\|u_q\|_2 \|v_q\|_2).\]

It remains to evaluate the deterministic coefficient in $D_1$. Let $\gamma_n=p/n$. Since $-\rho_n$ remains a fixed positive distance from the limiting spectral support 
\[ [(1-\sqrt{\gamma_n})^2,~ +\infty ]\]
the standard resolvent expansion for sample covariance matrices applies uniformly at $-\rho_n$; see, for example, \citet{bai2004clt}. After absorbing the $O(n^{-1})$ difference caused by deleting one column, this expansion gives
\[p^{-1}\mE\tr(\bA_1)=m_n(-\rho_n)+O(p^{-1}),\]
where $m_n(z)$ is the solution of the Mar\v{c}enko--Pastur equation
\[m_n(z) = \frac{1}{1-\gamma_n-\gamma_n z m_n(z)-z}.\]
The relevant solution is smooth in a neighborhood of the origin. The fixed-point equation gives
\[ m_n(-\rho_n)  = \frac{1}{1-\gamma_n  + \gamma_n \rho_n m_n(-\rho_n) + \rho_n}. \]
Because the limiting spectral support is bounded away from zero and $\rho_n=n^{-3/2}$,
\[ m_n(-\rho_n) = \frac{1}{1-\gamma_n} + o(p^{-1/2}).\]
Consequently,
\[\frac{p^{-1}\mE\tr(\bA_1)}{1+n^{-1}\mE\tr(\bA_1)}=1+o(p^{-1/2}).\]

Combining this expansion with the bounds for $D_2$, $D_3$, and $D_4$ proves \eqref{eq:mean_convergence_bound}.

\subsection{Conditional Variance of the Martingale Decomposition}
\label{subsec:variance_martingale}

In this subsection, we reduce the conditional variance of the martingale decomposition to quantities whose limits will be evaluated below:
\[
\sum_{j=1}^n
\mE_{j-1}
\left[
\left(
\sum_{q=1}^m
\beta^{\mE} 
\Big[
\calH_{qj}^{(1)}
+
\calH_{qj}^{(2)}
+
\calH_{qj}^{(3)}
\Big]
\right)^2
\right].
\]
To this end, it suffices to study all cross-product terms of the form $\calH_{qj}^{(g)} \calH_{q'j}^{(g')}$. Since $\calH_{qj}^{(2)}$ is the counterpart of
$\calH_{qj}^{(1)}$ obtained by interchanging $u_q$ and $v_q$,
it suffices to determine the limit of $\beta^{\mE}$ and the asymptotic behavior of terms of the following four types:
\begin{align*}
&\sum_{j=1}^n
\mE_{j-1}
\calH_{qj}^{(1)}
\calH_{ q'j}^{(1)},
\quad \sum_{j=1}^n
\mE_{j-1}
\calH_{qj}^{(1)}
\calH_{ q'j}^{(2)},
\quad\sum_{j=1}^n
\mE_{j-1}
\calH_{qj}^{(1)}
\calH_{ q'j}^{(3)},
\quad\sum_{j=1}^n
\mE_{j-1}
\calH_{qj}^{(3)}
\calH_{ q'j}^{(3)},
\end{align*}
for all $q,q'\in\{1,\ldots,m\}$.
The remaining cross terms follow from these four cases by interchanging the corresponding $u$- and $v$-vectors.

The calculation in Subsection~\ref{subsec:convergence_mean} also gives
\[
\beta^{\mE}=1-\gamma_n+o(p^{-1/2})\longrightarrow 1-\gamma.
\]

Set
\[ \eta_{q,q'} = n^{-3/16} \|u_q\|_2 \|u_{q'}\|_2 \|v_q \|_2 \|v_{q'}\|_2. \]
For regular vectors, $\eta_{q,q'}=O(n^{-3/16})$.
We claim the following expansions:
\begin{align}
&\sum_{j=1}^n
\mE_{j-1}
\calH^{(1)}_{qj}
\calH^{(1)}_{q'j}=
p^{-1}
\sum_{j=1}^n
u_{qj}u_{q'j}
\mE_{j-1}
\Big\{
\mE_j[\bz_j^T\bA_j \bZ_jv_q]
\,
\mE_j[\bz_j^T\bA_j\bZ_jv_{q'}]
\Big\} \nonumber \\
&=
p^{-1}
(\beta^{\mE})^2
\sum_{j=1}^n
u_{qj}u_{q'j}
\tr\!\Big[
(\mE_j\bA_j)^2
\Big]
\sum_{r=1}^{j-1}
v_{qr}v_{q'r}+ O_{L_1}\Big( \eta_{q,q'}\Big).\label{eq:limit_H1H1}\\
&\sum_{j=1}^n
\mE_{j-1}
\calH^{(1)}_{qj}
\calH^{(2)}_{q'j}
=
p^{-1}
\sum_{j=1}^n
u_{qj}v_{q'j}
\mE_{j-1}
\Big\{
\mE_j[\bz_j^T\bA_j \bZ_jv_q]
\,
\mE_j[\bz_j^T\bA_j  \bZ_ju_{q'}]
\Big\}\nonumber \\
&\qquad=
p^{-1}
(\beta^{\mE})^2
\sum_{j=1}^n
u_{qj}v_{q'j}
\tr\!\Big[
(\mE_j\bA_j)^2
\Big]
\sum_{r=1}^{j-1}
v_{qr}u_{q'r} + O_{L_1}\Big(\eta_{q,q'}\Big).\label{eq:limit_H1H2} \\
&\sum_{j=1}^n
\mE_{j-1}
\calH^{(1)}_{qj}
\calH^{(3)}_{q'j}
= \frac{-1}{p}
\sum_{j=1}^n
u_{qj}
\mE_{j-1}
\Big\{
\mE_j[\bz_j^T\bA_j \bZ_jv_q]
\,
\mE_j[\varrho_{q'j}]
\Big\} =  O_{L_1}\Big(\eta_{q,q'}\Big).  \label{eq:limit_H1H3}\\
&\sum_{j=1}^n
\mE_{j-1}
\calH^{(3)}_{qj}
\calH^{(3)}_{q'j}
=
p^{-1}
\sum_{j=1}^n
\mE_{j-1}
\Big\{
\mE_j\varrho_{qj}\,
\mE_j\varrho_{q'j}
\Big\} =
(\beta^{\mE})^4\frac{1}{pn^2}
\sum_{j=1}^n
\Big[
\tr\!\big(
(\mE_j\bA_j)^2 
\big)
\Big]^2\nonumber\\
&\times
\Bigg[
\left(
\sum_{r=1}^{j-1}u_{qr}v_{q'r}
\right)
\left(
\sum_{r=1}^{j-1}v_{qr}u_{q'r}
\right)
+
\left(
\sum_{r=1}^{j-1}u_{qr}u_{q'r}
\right)
\left(
\sum_{r=1}^{j-1}v_{qr}v_{q'r}
\right)
\Bigg] + O_{L_1}\Big(\eta_{q,q'}\Big). \label{eq:limit_H3H3}
\end{align}

We provide the proof only for \eqref{eq:limit_H1H1}. The remaining results follow from analogous arguments. Closely related calculations can be found in Sections~S.3.1.3--S.3.1.6 of \citet{li2020high}, and we therefore omit the corresponding details.

For $j\geq2$, introduce the conditionally independent copy
\[
\underline{\bZ}_j
= [\bz_1,\ldots,\bz_{j-1},0,
\underline{\bz}_{j+1},\ldots,\underline{\bz}_n],
\]
where $\underline{\bz}_{j+1},\ldots,\underline{\bz}_n$ are i.i.d. copies of $\bz_1$ that are independent of $\bz_1,\ldots,\bz_{j-1}$ and of the original future columns. Define $\underline{\bA}_j$ from $\underline{\bZ}_j$ in the same way that $\bA_j$ is defined from $\bZ_j$. Conditional on $\sigma(\bz_1,\ldots,\bz_{j-1})$, the pairs $(\bZ_j,\bA_j)$ and $(\underline{\bZ}_j,\underline{\bA}_j)$ are independent and identically distributed. Consequently,
\begin{align*}
&\mE_{j-1}\Big\{\mE_j[\bz_j^T\bA_j\bZ_jv_q]
\mE_j[\bz_j^T\bA_j\bZ_jv_{q'}]\Big\}=\mE_{j-1}\bz_j^T\mE_j\left[
\bA_j\bZ_jv_qv_{q'}^T\underline{\bZ}_j^T\underline{\bA}_j
\right]\bz_j\\
&\quad=\mE_{j-1}\tr\left(
\bA_j\bZ_jv_qv_{q'}^T\underline{\bZ}_j^T\underline{\bA}_j
\right).
\end{align*}

For $r<j$, let $\underline{\bZ}_{rj}$ be obtained from $\underline{\bZ}_j$ by replacing its $r$th column with zero. With this notation, for $j\geq2$,
\begin{align*}
&\mE_{j-1} \calH_{qj}^{(1)} \calH_{q'j}^{(1)} = \mE_{j-1} \Big[ \mE_j[p^{-1/2} u_{qj} \bz_j^T \bA_j  \bZ_{j} v_q ] \mE_j [p^{-1/2} u_{q'j} \bz_j^T \bA_j   \bZ_j v_{q'}]   \Big]\\
&=p^{-1} u_{qj} u_{q'j} \mE_{j-1} v_q^T \bZ_j^T \bA_j  \underline{\bA}_j   \underline{\bZ}_j v_{q'}\\
&=p^{-1} u_{qj} u_{q'j} \sum_{r=1}^{j-1} v_{qr}v_{q'r} \mE_{j-1} \bz_r^T \bA_j  \underline{\bA}_j   \bz_r + p^{-1} u_{qj}u_{q'j} \sum_{r=1}^{j-1} v_{qr} \mE_{j-1} \bz_r^T \bA_j  \underline{\bA}_j   \underline{\bZ}_{rj} v_{q'}\\
&\quad + p^{-1} u_{qj} u_{q'j} \sum_{r=j+1}^n v_{qr} \mE_{j-1}\bz_r^T \bA_j  \underline{\bA}_{j}  \underline{\bZ}_j v_{q'} \\
& = u_{qj} u_{q'j} \Big( \calJ_j^{(1)} + \calJ_j^{(2)} + \calJ_j^{(3)}\Big), \quad\text{say}. 
\end{align*}
We will show that
\begin{align*}
    &\sup_{j\geq2} \mE \Big| \calJ_j^{(1)} - (\beta^{\mE})^2 \Big(\sum_{r=1}^{j-1} v_{qr}v_{q'r} \Big) p^{-1}\tr\mE_{j-1} \Big[\bA_j \underline{\bA}_{j}   \Big]  \Big| = O(n^{-1/2} \|v_q\|_2 \|v_{q'}\|_2 ), \\ 
    &\sup_{j\geq2} \mE |\calJ_j^{(2)}| = O(n^{-3/16} \|v_q\|_2 \|v_{q'} \|_2 ),\quad \sup_{j\geq2} \mE |\calJ_j^{(3)}| = O(n^{-1/2} \|v_q\|_2 \|v_{q'} \|_2 ).\\
\end{align*}
Recall the definition of $\beta_{rj}$. Define its counterpart
\[
\underline{\beta}_{rj}
=\frac{1}{1+n^{-1}\bz_r^T\underline{\bA}_{rj}\bz_r},
\]
where $\underline{\bA}_{rj}$ is defined in the same way as $\bA_{rj}$, with $\underline{\bZ}_{rj}$ in place of $\bZ_{rj}$.

For $\calJ_j^{(1)}$,
\begin{align*}
&\mE \Big| \calJ_j^{(1)} -  (\beta^{\mE})^2  \Big(\sum\limits_{r=1}^{j-1} v_{qr} v_{q'r}\Big) p^{-1} \tr \left[\mE_{j-1}[\bA_j \underline{\bA}_j  ]\right]\Big| \\  
&=\mE \Big| \sum\limits_{r=1}^{j-1} v_{qr} v_{q'r} \mE_{j-1} \beta_{rj} \underline{\beta}_{rj}   
   p^{-1} \bz_r^T \bA_{rj} \underline{\bA}_{rj}  \bz_r - (\beta^{\mE})^2 \Big(\sum\limits_{r=1}^{j-1} v_{qr} v_{q'r}\Big) p^{-1} \tr \left[\mE_{j-1}[\bA_j \underline{\bA}_j  ]\right]\Big|  \\ 
&\leq \sum_{r=1}^{j-1} |v_{qr}||v_{q'r}| \mE \Big| \beta_{rj} \underline{\beta}_{rj}   
  \frac{1}{p}\bz_r^T \bA_{rj} \underline{\bA}_{rj}  \bz_r - (\beta^{\mE})^2 \frac{1}{p}\tr\{\bA_j \underline{\bA}_j  \}\Big|\\ 
&\leq  \sum_{r=1}^{n} |v_{qr}||v_{q'r}| \mE \Big| (\beta_{rj} -\beta^{\mE} )\underline{\beta}_{rj}   \frac{1}{p}\bz_r^T \bA_{rj} \underline{\bA}_{rj}  \bz_r \Big| \\ 
&+ \sum_{r=1}^{n} |v_{qr}||v_{q'r}| \mE \Big| {\beta}^{\mE} (\underline{\beta}_{rj}  -{\beta}^{\mE}  ) \frac{1}{p}\bz_r^T \bA_{rj} \underline{\bA}_{rj}  \bz_r \Big| \\ 
&  + \sum_{r=1}^{n} |v_{qr}||v_{q'r}|  |{\beta}^{\mE} |^2 \mE \Big|\frac{1}{p}\bz_r^T \bA_{rj} \underline{\bA}_{rj}  \bz_r -\frac{1}{p}\tr\{\bA_j \underline{\bA}_j  \} \Big| \\ 
&= O( n^{-1/2} \|v_q\|_2\|v_{q'}\|_2).
\end{align*}
The last line follows from Lemmas~\ref{lemma:bound_moment_A}, \ref{lemma:nonran_bound_resol}, \ref{lemma:converg_A}, and \ref{lemma:concentration_quadrat}, which imply
\begin{align*}
&|\beta^{\mE} |\leq 1, \\  
& \mE \Big| \frac{1}{p} \bz_r^T \bA_{rj} \underline{\bA}_{rj}  \bz_r\Big|^2 = O(1),\\  
& \mE \Big|\frac{1}{p} \bz_r^T \bA_{rj} \underline{\bA}_{rj}  \bz_r - \frac{1}{p} \tr\big[\bA_j \underline{\bA}_j  \big] \Big|^2  = O(n^{-1}),\\  
& \mE|\beta_{rj} -\beta^{\mE} |^2\leq \calK\mE\Big|\beta_{rj} \beta^{\mE} (\frac{1}{n}\bz_r^T \bA_{rj}\bz_r - \frac{1}{n}\mE \tr\bA_{rj})\Big|^2\\
 &\qquad + \calK\mE\Big|\beta_{rj} \beta^{\mE} (\frac{1}{n}\mE \tr\bA_{rj}- \frac{1}{n}\mE \tr\bA_{j})\Big|^2 = O(n^{-1}). 
\end{align*}
The same argument gives $\mE|\underline{\beta}_{rj}-\beta^{\mE}|^2=O(n^{-1})$.

For $\calJ_j^{(2)}$, the moment bounds
\begin{align*}
&\sup_{r<j}\mE \Big|(\beta_{rj} -\beta^{\mE} )p^{-1/2} \bz_r^T\bA_{rj} \underline{\bA}_{rj}  \underline{\bZ}_{rj} v_{q'} \Big| = O\Big(\frac{1}{\sqrt{n}}\|v_{q'}\|_2\Big), \\
&\sup_{r<j}\mE \Big|\frac{1}{p}(\beta_{rj} - \beta^{\mE} )\tr\big[\bA_{rj} \underline{\bA}_{rj}   \big] p^{-1/2} \bz_r^T\underline{\bA}_{rj}  \underline{\bZ}_{rj} v_{q'} \Big| = O\Big(\frac{1}{\sqrt{n}}\|v_{q'}\|_2\Big),\\
&\sup_{r<j}\mE \Big| \beta_{rj} \underline{\beta}_{rj}  \frac{1}{n}[\bz_r^T\bA_{rj} \underline{\bA}_{rj}  \bz_r- \tr\{\bA_{rj} \underline{\bA}_{rj}   \}]
 p^{-1/2} \bz_r^T\underline{\bA}_{rj}  \underline{\bZ}_{rj}v_{q'} \Big|\\
 &\quad= O\Big(\frac{1}{\sqrt{n}}\|v_{q'}\|_2\Big),
\end{align*}
follow from Lemmas~\ref{lemma:bound_moment_A}, \ref{lemma:converg_A}, \ref{lemma:concentration_quadrat}, and \ref{lemma:bound_z1_zalpha_b}. Hence,
\begin{align*}
&\calJ_j^{(2)} = p^{-1} \mE_{j-1}\sum\limits_{r=1}^{j-1} v_{qr}\beta_{rj} \bz_r^T \bA_{rj} \underline{\bA}_{j}  \underline{\bZ}_{rj} v_{q'} \\
& = p^{-1}\mE_{j-1} \sum\limits_{r=1}^{j-1} v_{qr}\beta_{rj} \bz_r^T\bA_{rj} \underline{\bA}_{rj}  \underline{\bZ}_{rj}v_{q'} - \mE_{j-1}\sum\limits_{r=1}^{j-1} \frac{1}{pn} v_{qr}\beta_{rj} \underline{\beta}_{rj}  \bz_r^T\bA_{rj} \underline{\bA}_{rj}  \bz_r\bz_r^T\underline{\bA}_{rj}  \underline{\bZ}_{rj}v_{q'} \\
& =  \sum\limits_{r=1}^{j-1} v_{qr} \beta^{\mE}  p^{-1} \mE_{j-1}\bz_r^T\bA_{rj} \underline{\bA}_{rj}  \underline{\bZ}_{rj}v_{q'}- \sum\limits_{r=1}^{j-1} \frac{ v_{qr}}{n}(\beta^{\mE})^2  \mE_{j-1} \frac{1}{p}\tr\{\bA_{rj} \underline{\bA}_{rj}   \}\bz_r^T\underline{\bA}_{rj}  \underline{\bZ}_{rj}v_{q'} \\
&\qquad + O_{L_1}( n^{-1/2} \|v_q\|_2\|v_{q'}\|_2 ). 
\end{align*}
The residual term above is uniform over $j$. 

We claim that the first two terms are also negligible. To verify the claim, we first show  
\begin{align}
&\sup_j\mE\Big| p^{-1}\sum\limits_{r=1}^{j-1} v_{qr} \bz_r^T \bA_{rj}  \underline{\bA}_{rj}   \underline{\bZ}_{rj}v_{q'} \Big|^2 = O(n^{-3/8}\|v_q\|_2^2 \|v_{q'}\|_2^2 ),\label{eq:exp_Jj2_1} \\ 
&\sup_j\mE\Big| p^{-1}\sum\limits_{r=1}^{j-1} v_{qr} \bz_r^T \underline{\bA}_{rj}   \underline{\bZ}_{rj}v_{q'} \Big|^2 = O(n^{-1}\|v_q\|_2^2 \|v_{q'}\|_2^2).\label{eq:exp_Jj2_2}
\end{align}
The proofs are analogous, so we establish only \eqref{eq:exp_Jj2_1}.

As with $\bZ_{rj}$, let $\bZ_{rr'j}$ be obtained from $\bZ$ by replacing $\bz_r$, $\bz_{r'}$, and $\bz_j$ with zero. For $r,r'<j$, define $\underline{\bZ}_{rr'j}$ by replacing the $r$th and $r'$th columns of $\underline{\bZ}_j$ with zero. Define $\bA_{rr'j}$ and $\underline{\bA}_{rr'j}$ analogously. We also set
\begin{align*}
&\beta_{rj}^{\tr}  = \frac{1}{1+ n^{-1}\tr\bA_{rj} },\\
&\theta_{rj}   =\frac{1}{n} \bz_r^T \bA_{rj}  \bz_r - \frac{1}{n}\tr \bA_{rj} ,\\
&\beta_{rr'j}  = \frac{1}{1+ n^{-1}\bz_{r'}^T\bA_{rr'j} \bz_{r'} },\\
&\underline{\beta}_{rr'j}  = \frac{1}{1+ n^{-1}\bz_{r'}^T\underline{\bA}_{rr'j} \bz_{r'} }.
\end{align*}  

For the diagonal terms in the expansion of \eqref{eq:exp_Jj2_1}, Lemma~\ref{lemma:bound_z1_zalpha_b} gives
\begin{align*}
&\sum\limits_{r=1}^{j-1} v_{qr}^2 \mE\Big|p^{-1}\bz_r^T \bA_{rj} \underline{\bA}_{rj}  \underline{\bZ}_{rj}v_{q'} \Big|^2
\leq p^{-1}\sum\limits_{r=1}^{n} v_{qr}^2 \mE\Big|p^{-1/2}\bz_r^T \bA_{rj} \underline{\bA}_{rj}  \underline{\bZ}_{rj}v_{q'} \Big|^2\\
&=O\Big(n^{-1}\|v_q\|_2^2 \|v_{q'}\|_2^2 \Big).
\end{align*}

For the cross terms with $r\neq r'$, arguments analogous to those used in the proof of (3.35) in \citet{Pan2011central} yield
\[
p^{-1}\left|\mE\left[
\bz_r^T\bA_{rj}\underline{\bA}_{rj}\underline{\bZ}_{rj}v_{q'}
\bz_{r'}^T\bA_{r'j}\underline{\bA}_{r'j}\underline{\bZ}_{r'j}v_{q'}
\right]\right|
\leq \calK n^{-3/8}\|v_{q'}\|_2^2.
\]
Therefore,
\[
\sum_{r\neq r'}|v_{qr}||v_{qr'}|p^{-2}
\left|\mE\left[
\bz_r^T\bA_{rj}\underline{\bA}_{rj}\underline{\bZ}_{rj}v_{q'}
\bz_{r'}^T\bA_{r'j}\underline{\bA}_{r'j}\underline{\bZ}_{r'j}v_{q'}
\right]\right|
\leq \calK n^{-3/8}\|v_q\|_2^2\|v_{q'}\|_2^2,
\]
where we used $\sum_{r,r'}|v_{qr}||v_{qr'}|\leq n\|v_q\|_2^2$.

For the second term in the expansion of $\calJ_j^{(2)}$,
 \begin{align*}
 \mE\Big| p^{-1/2}\Big[ &\tr\{\bA_{rj} \underline{\bA}_{rj}   \} - 
\tr\{\bA_{j} \underline{\bA}_{j}   \} 
 \Big]\bz_r^T\underline{\bA}_{rj}  \underline{\bZ}_{rj}v_{q'} \Big|^2\\
 &\leq \calK \mE \Big| p^{-1/2}\bz_r^T\underline{\bA}_{rj}  \underline{\bZ}_{rj}v_{q'}  \Big|^2 = O(\|v_{q'}\|_2^2),
 \end{align*}
by an inequality analogous to that in Lemma~\ref{lemma:nonran_bound_resol}.

Therefore,
\begin{align*}
&\sum\limits_{r=1}^{j-1} \frac{1}{n}v_{qr}(\beta^{\mE})^2  \frac{1}{p} \mE_{j-1}\tr\{\bA_{rj} \underline{\bA}_{rj}   \}\bz_r^T\underline{\bA}_{rj}  \underline{\bZ}_{rj}v_{q'}  \\
&= (\beta^{\mE})^2 \mE_{j-1}\Big[\frac{1}{np}\tr\{\bA_{j} \underline{\bA}_{j}   \}\sum\limits_{r=1}^{j-1}  v_{qr}\bz_r^T\underline{\bA}_{rj}  \underline{\bZ}_{rj}v_{q'}\Big] + o_{L_1}\Big(n^{-1/2}\|v_q\|_2 \|v_{q'}\|_2 \Big).
\end{align*}
By \eqref{eq:exp_Jj2_1}, \eqref{eq:exp_Jj2_2}, and Lemma~\ref{lemma:nonran_bound_resol},
\begin{align*}
&\mE\Big[\frac{1}{np}\tr\{\bA_{j} \underline{\bA}_{j}   \}\sum\limits_{r=1}^{j-1}  v_{qr}\bz_r^T\underline{\bA}_{rj}  \underline{\bZ}_{rj}v_{q'}\Big]^2 \leq \calK \mE\Big[p^{-1}\sum\limits_{r=1}^{n}  v_{qr}\bz_r^T\underline{\bA}_{rj}  \underline{\bZ}_{rj}v_{q'}\Big]^2 \\
&=O(n^{-1} \|v_q\|_2^2 \|v_{q'}\|_2^2). 
\end{align*}
Thus, the second term in the expansion of $\calJ_j^{(2)}$ is also $O_{L_1}(n^{-1/2} \|v_q\|_2 \|v_{q'}\|_2)$ uniformly in $j$. 

For $\calJ_j^{(3)}$, the bounds
\begin{align*}
&\mE|\theta_{rj} |^2 = O(n^{-1}),\\
&\mE|p^{-1/2}\bz_r^T \bA_{rj} \underline{\bA}_{j}  \underline{\bZ}_j v_{q'}|^2 = O(\|v_{q'}\|_2^2),
\end{align*}
give
\begin{align*}
\calJ_j^{(3)}& = p^{-1}\mE_{j-1} \sum\limits_{r=j+1}^n v_{qr} \bz_r^T \bA_j  \underline{\bA}_j   \underline{\bZ}_j v_{q'} \\
& = p^{-1}\mE_{j-1} \sum\limits_{r=j+1}^n v_{qr} \beta_{rj}  \bz_r^T \bA_{rj}  \underline{\bA}_j  \underline{\bZ}_jv_{q'} \\ 
& = p^{-1}\mE_{j-1} \sum\limits_{r=j+1}^n v_{qr} \beta_{rj}^{\tr}  \bz_r^T \bA_{rj} \underline{\bA}_{j}  \underline{\bZ}_j v_{q'} \\
& - p^{-1}\mE_{j-1} \sum\limits_{r=j+1}^n v_{qr} \beta_{rj}  \beta_{rj}^{\tr}  \theta_{rj}  \bz_r^T \bA_{rj} \underline{\bA}_{j}  \underline{\bZ}_j v_{q'}\\
&= - p^{-1}\mE_{j-1} \sum\limits_{r=j+1}^n v_{qr} \beta_{rj}  \beta_{rj}^{\tr}  \theta_{rj}  \bz_r^T \bA_{rj} \underline{\bA}_{j}  \underline{\bZ}_j v_{q'} \\
&= O_{L_1}\Big(n^{-1/2} \|v_q\|_2 \|v_{q'}\|_2\Big). 
\end{align*}
The fourth equality follows by conditioning on all columns except $\bz_r$: the coefficient of $\bz_r$ in the preceding term is independent of $\bz_r$, and $\mE\bz_r=0$.

Moreover, when $j=1$, it is straightforward to verify that
\[ \mE \calH_{q1}^{(1)}\calH_{q'1}^{(1)} = O\Big(n^{-3/16}\|u_q\|_2 \|u_{q'}\|_2 \|v_q\|_2 \|v_{q'}\|_2 \Big).\]

It follows that
\begin{align*}
   \sum_{j=1}^n \mE_{j-1} \calH_{qj}^{(1)} \calH_{q'j}^{(1)}  &=\sum_{j=1}^n u_{qj} u_{q'j} (\beta^{\mE})^2  \Big( \sum_{r=1}^{j-1} v_{qr}v_{q'r} \Big) p^{-1} \tr\mE_{j-1} \Big[ \bA_j \underline{\bA}_{j}   \Big] \\
    &+ O_{L_1} \Big( n^{-3/16} \|u_q\|_2\|u_{q'}\|_2 \|v_{q}\|_2 \|v_{q'}\|_2 \Big)\\
    &=\sum_{j=1}^n u_{qj} u_{q'j} (\beta^{\mE})^2 \Big( \sum_{r=1}^{j-1} v_{qr}v_{q'r} \Big) p^{-1} \tr \Big[ (\mE_j{\bA}_{j})^2   \Big] \\
    &+ O_{L_1} \Big( n^{-3/16} \|u_q\|_2\|u_{q'}\|_2 \|v_{q}\|_2 \|v_{q'}\|_2 \Big).
\end{align*}
Here, the second equality uses conditional independence to obtain
$\mE_{j-1}(\bA_j\underline{\bA}_j)=(\mE_j\bA_j)^2$.
This proves \eqref{eq:limit_H1H1}.

\subsection{A High-Order Moment Bound for Increments}
\label{subsec:increment_bound}

In this subsection, we consider a single bilinear form, corresponding to $m=1$, and hold $v$ fixed while varying the first argument from $u_1$ to $u_2$. By bilinearity, the resulting increment is
\begin{align*}
&\{H(u_1,v)-\mE H(u_1,v)\}
-\{H(u_2,v)-\mE H(u_2,v)\}\\
&\qquad=H(u_1-u_2,v)-\mE H(u_1-u_2,v).
\end{align*}
Suppose that $u_1$, $u_2$, and $v$ are regular, and set
\[
\tilde{u}=u_1-u_2,
\qquad
\calK_{\max}=\sqrt{n}\|v\|_\infty.
\]
Then $\sqrt{n}\|\tilde{u}\|_\infty\leq2\calK_0$, so the moment bounds for regular vectors continue to apply after enlarging the fixed regularity constant if necessary. Moreover, $\calK_{\max}$ is uniformly bounded and $\|v\|_2\leq\calK_{\max}$.

For every fixed even integer $r\geq8$, uniformly over such regular vectors, we will prove that
\begin{equation}
\label{eq:eighth_moment_bound}
\mE \left| \sqrt{p}\{H(u_1-u_2,v)-\mE H(u_1-u_2,v)\}\right|^r
=O(\|u_1-u_2\|_2^r).
\end{equation}

Recall the definitions of $d_{qj1}$, $d_{qj2}$, and $d_{qj3}$ in Subsection~\ref{subsec:martingale_construct}. Replacing $(u_q,v_q)$ by $(\tilde{u},v)$, denote the corresponding quantities by $d_{j1}$, $d_{j2}$, and $d_{j3}$. For $\ell=1,2,3$, define
\[
Y_{j\ell}=\sqrt{p}(\mE_j-\mE_{j-1})d_{j\ell}.
\]
Then $\{Y_{j\ell}:1\leq j\leq n\}$ is a martingale difference sequence for each $\ell$, and
\[
\sqrt{p}\{H(\tilde{u},v)-\mE H(\tilde{u},v)\}
=\sum_{j=1}^n\sum_{\ell=1}^3Y_{j\ell}.
\]
Conditional Jensen's inequality gives, for every fixed $r\geq2$,
\begin{equation}
\label{eq:increment_conditional_jensen}
\mE|Y_{j\ell}|^r
\leq \calK_r\mE|\sqrt{p}\,d_{j\ell}|^r.
\end{equation}

We first bound the three terms on the right-hand side of \eqref{eq:increment_conditional_jensen}. The decomposition of $d_{j1}$ gives
\begin{align*}
\mE|\sqrt{p}\,d_{j1}|^r
&\leq \calK_r|\tilde{u}_j|^r
\mE|p^{-1/2}\bz_j^T\bA_j\bZ_jv|^r\\
&\quad+\calK_r p^{r/2}|\tilde{u}_j|^r|v_j|^r
\mE|p^{-1}\bz_j^T\bA_j\bz_j|^r\\
&\quad+\calK_r|\tilde{u}_j|^r
\mE\left|\left(n^{-1}\bz_j^T\bA_j\bz_j\right)
\left(p^{-1/2}\bz_j^T\bA_j\bZ_jv\right)\beta_j\right|^r\\
&\quad+\calK_r p^{r/2}|\tilde{u}_j|^r|v_j|^r
\mE\left|p^{-1}n^{-1}(\bz_j^T\bA_j\bz_j)^2\beta_j\right|^r.
\end{align*}
Lemmas~\ref{lemma:bound_moment_A}, \ref{lemma:concentration_quadrat}, and \ref{lemma:leave_out_moment_bounds}, together with H\"older's inequality, imply that, for every fixed $r\geq4$,
\begin{align*}
&\sup_j\mE|p^{-1/2}\bz_j^T\bA_j\bZ_jv|^r
\leq\calK_r\|v\|_2^r,\\
&\sup_j\mE|p^{-1}\bz_j^T\bA_j\bz_j|^r
\leq\calK_r,\\
&\sup_j\mE\left|\left(n^{-1}\bz_j^T\bA_j\bz_j\right)
\left(p^{-1/2}\bz_j^T\bA_j\bZ_jv\right)\beta_j\right|^r
\leq\calK_r\|v\|_2^r,\\
&\sup_j\mE\left|p^{-1}n^{-1}(\bz_j^T\bA_j\bz_j)^2\beta_j\right|^r
\leq\calK_r,
\end{align*}
where the product bounds follow from H\"older's inequality applied with moment order $2r$. Since $p/n$ is bounded away from zero and infinity,
\begin{equation}
\label{eq:increment_d1_bound}
\mE|\sqrt{p}\,d_{j1}|^r
\leq\calK_r\calK_{\max}^r|\tilde{u}_j|^r.
\end{equation}

The decomposition of $d_{j2}$ can be written as
\begin{align*}
\sqrt{p}\,d_{j2}
&=-v_j\left(n^{-1}\bz_j^T\bA_j\bz_j\right)
\left(p^{-1/2}\bz_j^T\bA_j\bZ_j\tilde{u}\right)\beta_j\\
&\quad-\frac{\sqrt{p}}{n}
\left(p^{-1/2}\bz_j^T\bA_j\bZ_jv\right)
\left(p^{-1/2}\bz_j^T\bA_j\bZ_j\tilde{u}\right)\beta_j.
\end{align*}
The same moment bounds and H\"older's inequality therefore yield
\begin{align}
\mE|\sqrt{p}\,d_{j2}|^r
&\leq\calK_r\left(|v_j|^r+n^{-r/2}\|v\|_2^r\right)
\|\tilde{u}\|_2^r\leq\calK_r\calK_{\max}^r n^{-r/2}\|\tilde{u}\|_2^r.
\label{eq:increment_d2_bound}
\end{align}
Finally,
\begin{align}
\mE|\sqrt{p}\,d_{j3}|^r
&=|v_j|^r\mE|p^{-1/2}\bz_j^T\bA_j\bZ_j\tilde{u}|^r\leq\calK_r\calK_{\max}^r n^{-r/2}\|\tilde{u}\|_2^r.
\label{eq:increment_d3_bound}
\end{align}

Set $k=r/2$. Applying Lemma~\ref{lemma:Burkholder_martingale_moments}, first to the sum over $j$ and then using the triangle inequality for the three values of $\ell$, gives
\begin{align}
&\mE\left|\sum_{j=1}^n\sum_{\ell=1}^3Y_{j\ell}\right|^{2k}\leq\calK_k\sum_{\ell=1}^3
\left\{
\mE\left(\sum_{j=1}^n\mE_{j-1}|Y_{j\ell}|^2\right)^k
+\sum_{j=1}^n\mE|Y_{j\ell}|^{2k}
\right\}.
\label{eq:increment_burkholder}
\end{align}
By Minkowski's and conditional Jensen's inequalities,
\begin{equation}
\mE\left(\sum_{j=1}^n\mE_{j-1}|Y_{j\ell}|^2\right)^k
\leq
\left\{\sum_{j=1}^n(\mE|Y_{j\ell}|^{2k})^{1/k}\right\}^k.
\label{eq:increment_predictable_bound}
\end{equation}
Equations~\eqref{eq:increment_conditional_jensen}--\eqref{eq:increment_d3_bound} imply
\begin{align*}
\mE|Y_{j1}|^{2k}
&\leq\calK_k\calK_{\max}^{2k}|\tilde{u}_j|^{2k},\\
\mE|Y_{j2}|^{2k}+\mE|Y_{j3}|^{2k}
&\leq\calK_k\calK_{\max}^{2k}n^{-k}\|\tilde{u}\|_2^{2k}.
\end{align*}
Substituting these bounds into \eqref{eq:increment_burkholder} and \eqref{eq:increment_predictable_bound}, and using
\[
\sum_{j=1}^n|\tilde{u}_j|^{2k}
\leq\|\tilde{u}\|_2^{2k},
\qquad n^{1-k}\leq1,
\]
we obtain
\[
\mE\left|\sqrt{p}\{H(\tilde{u},v)-\mE H(\tilde{u},v)\}\right|^{2k}
\leq\calK_k\calK_{\max}^{2k}\|\tilde{u}\|_2^{2k}.
\]
Because $\calK_{\max}$ is uniformly bounded, this proves \eqref{eq:eighth_moment_bound}.

\subsection[Asymptotic Normality of H]{Asymptotic Normality of $\calH$}
\label{subsec:asymptotic_normality_calH}

We now establish asymptotic normality for a fixed finite superposition of the bilinear forms introduced above.

\begin{theorem}
\label{thm:asymptotic_normality_general_form}
Suppose that Conditions~\ref{enum:moment_conditions} and \ref{enum:high_dimensional_regime} hold, and retain the standing truncation and regularization convention in Section~\ref{sec:proof_thm_null}. Let $m$ be fixed, and let $\{u_q\}_{q=1}^m$ and $\{v_q\}_{q=1}^m$ be collections of regular vectors. Define
\begin{align*}
\mathfrak{M}
&=\sum_{q=1}^m u_q^Tv_q,\\
\mathfrak{V}
&=(1-\gamma_n)\sum_{q,q'=1}^m
\left\{
(u_q^Tu_{q'})(v_q^Tv_{q'})
+(u_q^Tv_{q'})(v_q^Tu_{q'})
\right\}.
\end{align*}
If $\liminf_{n\to\infty}\mathfrak{V}>0$, then
\[
\frac{\sqrt{p}(\calH-\mathfrak{M})}{\mathfrak{V}^{1/2}}
\stackrel{D}{\longrightarrow}N(0,1).
\]
\end{theorem}

Subsection~\ref{subsec:convergence_mean} and the fact that $m$ is fixed give
\[
\sqrt{p}\,|\mE\calH-\mathfrak{M}|=o(1).
\]
Moreover, Subsection~\ref{subsec:martingale_construct} gives
\begin{equation}
\label{eq:calH_martingale_for_clt}
\sqrt{p}(\calH-\mE\calH)
=\sum_{j=1}^n X_{nj}+o_{L_2}(1),
\qquad
X_{nj}=\sum_{q=1}^m\beta^{\mE}
\left(\calH_{qj}^{(1)}+\calH_{qj}^{(2)}+\calH_{qj}^{(3)}\right).
\end{equation}
The array $\{X_{nj}:1\leq j\leq n\}$ is a martingale difference array with respect to $\{\sigma(\bz_1,\ldots,\bz_j)\}$.

We use the following standard martingale central limit theorem.
\begin{theorem}[Theorem 35.12 of \citet{billingsley1995probability}]
\label{thm:martingale_central_limit_theorem}
Let $\{Y_{nj}:1\leq j\leq n\}$ be a martingale difference array with respect to filtrations $\{\mathcal F_{nj}\}$. Suppose that, for some constant $\sigma^2>0$,
\begin{align*}
&\sum_{j=1}^n\mE(Y_{nj}^2\mid\mathcal F_{n,j-1}) \stackrel{P}{\longrightarrow}\sigma^2,\\
&\sum_{j=1}^n\mE\{Y_{nj}^2\mathbbm{1}(|Y_{nj}|>\epsilon)\}\longrightarrow0
\end{align*}
for every $\epsilon>0$. Then $\sum_{j=1}^nY_{nj}\stackrel{D}{\longrightarrow}N(0,\sigma^2)$.
\end{theorem}

We first verify the Lindeberg condition. Conditional Jensen's inequality, regularity of the vectors, and Lemma~\ref{lemma:leave_out_moment_bounds} imply, uniformly in $j$, that
\begin{align*}
\mE|\calH_{qj}^{(1)}|^4
&\leq\calK|u_{qj}|^4\|v_q\|_2^4=O(n^{-2}),\\
\mE|\calH_{qj}^{(2)}|^4
&\leq\calK|v_{qj}|^4\|u_q\|_2^4=O(n^{-2}),\\
\mE|\calH_{qj}^{(3)}|^4
&\leq p^{-2}\mE|\varrho_{qj}|^4=O(n^{-2}).
\end{align*}
Since $m$ is fixed and $|\beta^{\mE}|\leq1$,
\[
\sum_{j=1}^n\mE|X_{nj}|^4=O(n^{-1}).
\]
Consequently, for every $\epsilon>0$,
\[
\sum_{j=1}^n\mE\{X_{nj}^2\mathbbm{1}(|X_{nj}|>\epsilon)\}
\leq\epsilon^{-2}\sum_{j=1}^n\mE|X_{nj}|^4
\longrightarrow0.
\]
Because $\mathfrak V$ is bounded away from zero, the same conclusion holds for $X_{nj}/\mathfrak V^{1/2}$.

It remains to evaluate the conditional variance. Let
\[
c_n=1-\gamma_n,
\qquad
d_{nj}=1-\gamma_n\frac{j-1}{n},
\qquad
T_{nj}=p^{-1}\tr\{(\mE_j\bA_j)^2\}.
\]
Recall that $m_n(z)$ denotes the solution of the Mar\v{c}enko--Pastur equation introduced in Subsection~\ref{subsec:convergence_mean}, and define its companion transform by
\[
\underline m_n(z)=\frac{\gamma_n-1}{z}+\gamma_nm_n(z).
\]
The resolvent calculation in equation (2.18) of \citet{bai2004clt}, evaluated at $z=-\rho_n$, gives
\begin{align}
&\tr\{(\mE_j\bA_j)^2\}
\left[
1-\frac{(j-1)p}{n^2}
\left\{\frac{\underline m_n(-\rho_n)}{1+\underline m_n(-\rho_n)}\right\}^2
\right]\nonumber\\
&\qquad=\frac{p}{\rho_n^2\{1+\underline m_n(-\rho_n)\}^2}+l_{nj},
\qquad
\sup_j\mE|l_{nj}|\leq\calK\sqrt n.
\label{eq:conditional_resolvent_equivalent}
\end{align}
Together with the fixed-point expansion in Subsection~\ref{subsec:convergence_mean}, this yields
\begin{equation}
\label{eq:Tnj_equivalent}
\sup_{1\leq j\leq n}
\mE\left|T_{nj}-\frac{1}{c_n^2d_{nj}}\right|
\longrightarrow0,
\qquad
\beta^{\mE}=c_n+o(1).
\end{equation}
The moment bounds in Lemma~\ref{lemma:bound_moment_A} make the sequence $\{T_{nj}^2\}$ uniformly integrable, so the corresponding squared approximation also holds.

For fixed $q,q'$, define
\begin{align*}
S_{q,q'}(j)
&=\left(\sum_{r=1}^j u_{qr}u_{q'r}\right)
  \left(\sum_{r=1}^j v_{qr}v_{q'r}\right)\\
&\quad+\left(\sum_{r=1}^j u_{qr}v_{q'r}\right)
  \left(\sum_{r=1}^j v_{qr}u_{q'r}\right),
\qquad S_{q,q'}(0)=0.
\end{align*}
Regularity implies $\sup_j|S_{q,q'}(j)|=O(1)$. Combining \eqref{eq:limit_H1H1}--\eqref{eq:limit_H3H3}, including the symmetric counterparts obtained by interchanging $u$ and $v$, gives
\begin{align}
&\sum_{g,g'=1}^3\sum_{j=1}^n
(\beta^{\mE})^2
\mE_{j-1}\{\calH_{qj}^{(g)}\calH_{q'j}^{(g')}\}\nonumber\\
&\qquad=\sum_{j=1}^n a_{nj}
\{S_{q,q'}(j)-S_{q,q'}(j-1)\}
+\sum_{j=1}^n b_{nj}S_{q,q'}(j-1)+o_p(1),
\label{eq:conditional_variance_summation}
\end{align}
where the terms involving four entries with index $j$ have total magnitude $O(n^{-1})$, and
\[
a_{nj}=(\beta^{\mE})^4T_{nj},
\qquad
b_{nj}=(\beta^{\mE})^6\frac{p}{n^2}T_{nj}^2.
\]
By \eqref{eq:Tnj_equivalent}, these coefficients may be replaced, uniformly in the weighted sums above, by
\[
\bar a_{nj}=\frac{c_n^2}{d_{nj}},
\qquad
\bar b_{nj}=\frac{c_n^2\gamma_n}{n d_{nj}^2}.
\]
For $j\geq2$,
\[
\bar b_{nj}=\bar a_{nj}-\bar a_{n,j-1}+O(n^{-2}),
\qquad
\bar a_{nn}=c_n+O(n^{-1}).
\]
Discrete summation by parts in \eqref{eq:conditional_variance_summation} therefore gives
\begin{align*}
&\sum_{g,g'=1}^3\sum_{j=1}^n
(\beta^{\mE})^2
\mE_{j-1}\{\calH_{qj}^{(g)}\calH_{q'j}^{(g')}\}\\
&\qquad=c_nS_{q,q'}(n)+o_p(1)\\
&\qquad=(1-\gamma_n)
\left\{
(u_q^Tu_{q'})(v_q^Tv_{q'})
+(u_q^Tv_{q'})(v_q^Tu_{q'})
\right\}+o_p(1).
\end{align*}
Since $m$ is fixed, summing over $q,q'$ yields
\[
\sum_{j=1}^n\mE_{j-1}(X_{nj}^2)
=\mathfrak V+o_p(1).
\]
Thus, the sum of the conditional variances of the normalized array
$\{X_{nj}/\mathfrak V^{1/2}:1\leq j\leq n\}$ converges in probability to
one. Theorem~\ref{thm:martingale_central_limit_theorem},
\eqref{eq:calH_martingale_for_clt}, and Slutsky's theorem complete the proof.

\subsection{Weak Convergence of the Processes}
\label{subsec:weak_convergence_process}

We establish weak convergence of the centered quadratic
process and the uniform stochastic boundedness needed for the bilinear terms
involving the columns of $Q$.

We begin with finite-dimensional distributions. Fix $m\in\mathbb N_+$ and
$t_1,\ldots,t_m\in\calT_{\rm mc}(\varepsilon)$. For arbitrary deterministic
coefficients $a_1,\ldots,a_m$,
\[
\sum_{q=1}^m a_qH(u(t_q))
=\sum_{q=1}^m H(a_qu(t_q),u(t_q)),
\]
which is a superposition of the form considered in
Subsection~\ref{subsec:asymptotic_normality_calH}. More generally, writing
$Q_1,\ldots,Q_M$ for the columns of $Q$, every linear combination of the
finite collection $\{H(u(t_q),Q_\ell):1\leq q\leq m,\ 1\leq\ell\leq M\}$
can be written as
\[
\sum_{\ell=1}^M\sum_{q=1}^m
H(a_{\ell q}u(t_q),Q_\ell).
\]
Theorem~\ref{thm:asymptotic_normality_general_form} and the
Cram\'{e}r--Wold device therefore give finite-dimensional asymptotic
normality whenever the limiting variance is positive. If the limiting
variance is zero, the same conclusion follows in the degenerate sense from
the corresponding second-moment bound and Chebyshev's inequality.

We next establish tightness. The processes under consideration are
\[
\left\{\sqrt p\{H(u_t)-1\}:t\in\calT_{\rm mc}(\varepsilon)\right\}
\quad\text{and}\quad
\left\{\sqrt p\{H(u_t,v)-u_t^Tv\}:t\in\calT_{\rm mc}(\varepsilon)\right\},
\]
where $v$ is a fixed regular vector. Here we have used
$u_t^Tu_t=1$. The eighth-moment bound in
\eqref{eq:eighth_moment_bound}, together with
\eqref{eq:mean_convergence_bound}, implies tightness at each fixed
$t_0\in\calT_{\rm mc}(\varepsilon)$.

It remains to control the increments. Because
$k(t_j)=\lfloor nt_j\rfloor+1$, these processes are constant on cells with
side lengths $n^{-1}$. It is therefore enough to work on the effective
$n^{-1}$-grid of their distinct values. We first record the deterministic
bound
\begin{equation}
\label{eq:u_increment_bound}
\|u_t-u_{t'}\|_2^2
\leq\calK\left(\|t-t'\|_2+\frac1n\right),
\qquad t,t'\in\calT_{\rm mc}(\varepsilon).
\end{equation}
To verify this bound, observe that changing $t$ to $t'$ changes at most
$\calK\{n\|t-t'\|_2+1\}$ columns of $\calX_t$. The uniform entry-wise bound
on $\bX$ and Condition~\ref{enum:full_rank} consequently give
\[
\left\|n^{-1}\calX_t\calX_t^T
-n^{-1}\calX_{t'}\calX_{t'}^T\right\|_2
\leq\calK\left(\|t-t'\|_2+\frac1n\right).
\]
Applying the inverse perturbation identity, using the uniform lower
eigenvalue bound in Condition~\ref{enum:full_rank}, and then accounting for
the columns that enter or leave the three segments yields
\eqref{eq:u_increment_bound} after normalization by $\|\bh_t\|_2$ and
$\|\bh_{t'}\|_2$.

Set $\delta=u_t-u_{t'}$. By bilinearity and the symmetry of $\bA$,
\[
H(u_t)-H(u_{t'})=H(\delta,u_t)+H(\delta,u_{t'}).
\]
Moreover,
\[
\delta^Tu_t+\delta^Tu_{t'}
=\|u_t\|_2^2-\|u_{t'}\|_2^2=0.
\]
Thus, both the quadratic and bilinear increments reduce to terms of the
form
\[
\sqrt p\{H(\delta,v)-\delta^Tv\},
\]
with $v$ regular and $\|v\|_2$ uniformly bounded. Equations
\eqref{eq:mean_convergence_bound} and \eqref{eq:eighth_moment_bound} give,
uniformly over the vectors under consideration,
\[
\mE\left|\sqrt p\{H(\delta,v)-\delta^Tv\}\right|^8
\leq\calK\|\delta\|_2^8.
\]
Consequently, \eqref{eq:u_increment_bound} yields
\begin{align*}
&\mE\left|\sqrt p\{H(u_t)-H(u_{t'})\}\right|^8
\leq\calK\left(\|t-t'\|_2+\frac1n\right)^4,\\
&\mE\left|\sqrt p\{H(u_t,v)-u_t^Tv-H(u_{t'},v)+u_{t'}^Tv\}\right|^8
\leq\calK\left(\|t-t'\|_2+\frac1n\right)^4.
\end{align*}
For two distinct points on the effective grid,
$\|t-t'\|_2\geq n^{-1}$, so the right-hand sides are bounded by
$\calK\|t-t'\|_2^4$. Since $4>3$, the standard chaining argument underlying
the multi-parameter Kolmogorov--Chentsov criterion, applied to the effective
grid, gives asymptotic tightness; see, for example,
\citet[Theorem~2.5.1]{khoshnevisan2002multiparameter}. The maximal jump
between neighboring cells is $o_p(1)$: there are $O(n^3)$ neighboring pairs,
whereas the eighth moment of each corresponding increment is $O(n^{-4})$.
Thus, passing from the effective grid to the original step processes does not
affect tightness.

For the quadratic process, the covariance formula in
Theorem~\ref{thm:asymptotic_normality_general_form} gives
\[
\operatorname{Cov}\left(
\frac{\sqrt p\{H(u_t)-1\}}{\sqrt{2(1-\gamma_n)}},
\frac{\sqrt p\{H(u_{t'})-1\}}{\sqrt{2(1-\gamma_n)}}
\right)
=|u_t^Tu_{t'}|^2+o(1).
\]
Because $u_t$ is the zero-padded normalization of $\bh_t$,
Condition~\ref{enum:stable_cov_kernel} identifies the limiting covariance
with $\calK(t,t')$. The finite-dimensional convergence and tightness proved
above therefore yield
\[
\left\{
\frac{\sqrt p}{\sqrt{2(1-\gamma_n)}}\{H(u_t)-1\}:
t\in\calT_{\rm mc}(\varepsilon)
\right\}
\rightsquigarrow
\left\{G(t):t\in\calT_{\rm mc}(\varepsilon)\right\}
\]
in $\ell^\infty(\calT_{\rm mc}(\varepsilon))$.

For the bilinear terms, define
\[
H(u_t,Q)=\bigl(H(u_t,Q_1),\ldots,H(u_t,Q_M)\bigr).
\]
Indeed, 
\[
\bX u_t=\|\bh_t\|_2^{-1}(\theta-\theta)=0.
\]
It follows from the definition of $Q$ that $u_t^TQ=0$. Since $M$ is fixed,
the finite-dimensional bounds and asymptotic tightness imply
\begin{equation}
\label{eq:H_uQ_uniform_bound}
\sup_{t\in\calT_{\rm mc}(\varepsilon)}
\sqrt p\,\|H(u_t,Q)\|_2=O_p(1).
\end{equation}

Finally, consider the regularized Woodbury expression corresponding to
$p^{-1}V_t$. Let
\[
L_n(t)=\frac{\sqrt p}{\sqrt{2(1-\gamma_n)}}\{p^{-1}V_t-1\}.
\]
The process \(L_n(t)\) is used only as an intermediate linear
normalization; the statistic in the main text is obtained from it by the
square-root transformation. The Woodbury expression is
\begin{align*}
H(u_t)&+\frac{1}{pn}u_t^T\bZ^T\bA\bZ Q
\left(I_M-n^{-1}Q^T\bZ^T\bA\bZ Q\right)^{-1}
Q^T\bZ^T\bA\bZ u_t.
\end{align*}
Applying Theorem~\ref{thm:asymptotic_normality_general_form} to the finitely
many pairs of columns of $Q$ gives, in operator norm,
\[
n^{-1}Q^T\bZ^T\bA\bZ Q
=\gamma_nI_M+O_p(n^{-1/2}).
\]
It follows that
\[
\lambda_{\min}\left(I_M-n^{-1}Q^T\bZ^T\bA\bZ Q\right)
=1-\gamma_n+o_p(1)\longrightarrow1-\gamma>0,
\]
and hence the norm of its inverse is $O_p(1)$. In terms of $H(u_t,Q)$,
the absolute value of the second Woodbury term is bounded uniformly in $t$
by
\[
\frac pn
\left\|\left(I_M-n^{-1}Q^T\bZ^T\bA\bZ Q\right)^{-1}\right\|_2
\|H(u_t,Q)\|_2^2.
\]
By \eqref{eq:H_uQ_uniform_bound}, this bound is
$O_p(p^{-1})=o_p(p^{-1/2})$. Thus, the Woodbury correction is uniformly
negligible on the scale of the limiting process. Finally,
\[
\frac{\sqrt n}{\sqrt{2n/p-2}}
=\frac{\sqrt p}{\sqrt{2(1-\gamma_n)}},
\]
so the preceding weak convergence gives the desired limit for the
regularized and truncated version of \(L_n(t)\). It remains to transfer this
conclusion to the transformed statistic in \eqref{eq:def_stat_normalized}.
Since
\[
\sup_{t\in\calT_{\rm mc}(\varepsilon)}
\left|p^{-1}V_t-1\right|=O_p(p^{-1/2})=o_p(1),
\]
Taylor's theorem for \(g(x)=x^{1/2}\), uniformly on a shrinking neighborhood
of one, gives
\[
\sup_{t\in\calT_{\rm mc}(\varepsilon)}
\left|
2\{(p^{-1}V_t)^{1/2}-1\}-\{p^{-1}V_t-1\}
\right|=O_p(p^{-1}).
\]
Multiplication by \(\sqrt p/\sqrt{2(1-\gamma_n)}\) shows that the
regularized and truncated square-root statistic differs from \(L_n(t)\) by
\(o_p(1)\) uniformly over the scanning set. Hence it has the same Gaussian
weak limit. The truncation and regularization steps are removed in
Subsection~\ref{subsec:truncation}.

\subsection{Effect of truncation and regularization}\label{subsec:truncation}

As in the preceding subsections, $z_{ij}$ denotes the truncated, centered,
and standardized variable. In this subsection, let $\breve z_{ij}$ denote
the corresponding original variable under Condition~\ref{enum:moment_conditions},
and let $\breve{\bZ}=(\breve{\bz}_1,\ldots,\breve{\bz}_n)$ and
$\breve{\bOmega}=n^{-1}\breve{\bZ}\breve{\bZ}^T$. To simplify the
truncation formula, set
\begin{align*}
a_n&=\kappa_n n^{1/2},\\
\mu_n&=\mE\{\breve z_{11}\mathbbm{1}(|\breve z_{11}|\leq a_n)\},\\
\sigma_n^2
&=\mE\left[\breve z_{11}\mathbbm{1}(|\breve z_{11}|\leq a_n)-\mu_n\right]^2.
\end{align*}
Thus,
\[
z_{ij}=\sigma_n^{-1}
\left\{\breve z_{ij}\mathbbm{1}(|\breve z_{ij}|\leq a_n)-\mu_n\right\}.
\]

We first construct an event on which the original and modified matrices can
be compared uniformly. Define
\[
\calE_n=
\left\{\max_{1\leq i\leq p,\ 1\leq j\leq n}|\breve z_{ij}|\leq a_n\right\}
\cap\left\{\lambda_{\min}(\breve{\bOmega})\geq\mathfrak D\right\}
\cap\left\{\lambda_{\min}(\bOmega)\geq\mathfrak D\right\}.
\]
The truncation condition imposed on $\kappa_n$ gives
\begin{align*}
\mP\left(\max_{1\leq i\leq p,\ 1\leq j\leq n}|\breve z_{ij}|>a_n\right)
&\leq np\,\mP(|\breve z_{11}|>a_n)\\
&\leq\frac pn\kappa_n^{-4}
\mE\left\{\breve z_{11}^4\mathbbm{1}(|\breve z_{11}|>a_n)\right\}
\longrightarrow0.
\end{align*}
Moreover, the finite fourth moment and the smallest-eigenvalue theorem of
\citet{bai1993limit} imply
\[
\lambda_{\min}(\breve{\bOmega})\stackrel{a.s.}{\longrightarrow}
(1-\sqrt\gamma)^2,
\]
whereas Lemma~\ref{lemma:poly_bound_tSigma} gives the corresponding
probability bound for $\bOmega$. Therefore,
\begin{equation}
\label{eq:truncation_event_probability}
\mP(\calE_n)\longrightarrow1.
\end{equation}

We next quantify the effects of centering and standardization. Since
$\mE\breve z_{11}=0$,
\begin{align*}
|\mu_n|
&=\left|\mE\{\breve z_{11}\mathbbm{1}(|\breve z_{11}|>a_n)\}\right|\\
&\leq a_n^{-3}
\mE\{\breve z_{11}^4\mathbbm{1}(|\breve z_{11}|>a_n)\}
=o(\kappa_n n^{-3/2}).
\end{align*}
Also,
\begin{align*}
1-\sigma_n^2
&=\mE\{\breve z_{11}^2\mathbbm{1}(|\breve z_{11}|>a_n)\}
+\mu_n^2\\
&\leq a_n^{-2}
\mE\{\breve z_{11}^4\mathbbm{1}(|\breve z_{11}|>a_n)\}
+\mu_n^2
=o(p^{-1}).
\end{align*}
In particular, $\sigma_n\to1$ and
\begin{equation}
\label{eq:truncation_center_scale_rates}
|\sigma_n^{-1}-1|=o(p^{-1}),
\qquad
|\mu_n|\sqrt{np}=o(p^{-1/2}).
\end{equation}

On $\calE_n$, no entry is truncated, and hence
\[
\bZ=\sigma_n^{-1}
\left(\breve{\bZ}-\mu_n\mathbf 1_p\mathbf 1_n^T\right).
\]
The standard largest-eigenvalue bound under a finite fourth moment gives
$\|\breve{\bZ}\|_2=O_p(p^{1/2})$; see \citet{yin1988limit}. It follows from
\eqref{eq:truncation_center_scale_rates} that
\begin{equation}
\label{eq:truncation_matrix_difference}
\|\bZ-\breve{\bZ}\|_2\mathbbm{1}(\calE_n)
=o_p(p^{-1/2}),
\qquad
\|\bZ\|_2\mathbbm{1}(\calE_n)=O_p(p^{1/2}).
\end{equation}
Consequently,
\begin{align}
\|\bOmega-\breve{\bOmega}\|_2\mathbbm{1}(\calE_n)
&\leq n^{-1}(\|\bZ\|_2+\|\breve{\bZ}\|_2)
\|\bZ-\breve{\bZ}\|_2\mathbbm{1}(\calE_n)
=o_p(p^{-1}).
\label{eq:truncation_covariance_difference}
\end{align}

On $\calE_n$, both $\|\bA\|_2$ and
$\|\breve{\bOmega}^{-1}\|_2$ are bounded by $\mathfrak D^{-1}$.
The resolvent identity and \eqref{eq:truncation_covariance_difference}
therefore give
\begin{align}
\|\bA-\breve{\bOmega}^{-1}\|_2\mathbbm{1}(\calE_n)
&\leq\mathfrak D^{-2}
\left(\|\bOmega-\breve{\bOmega}\|_2+\rho_n\right)
\mathbbm{1}(\calE_n)
=o_p(p^{-1}).
\label{eq:truncation_resolvent_difference}
\end{align}

Let $x$ and $y$ be any vectors drawn from
$\{u_t:t\in\calT_{\rm mc}(\varepsilon)\}\cup\{Q_1,\ldots,Q_M\}$.
All these vectors have unit Euclidean norm. Adding and subtracting the
appropriate intermediate terms yields
\begin{align*}
&p^{-1/2}\left|x^T\bZ^T\bA\bZ y
-x^T\breve{\bZ}^T\breve{\bOmega}^{-1}\breve{\bZ}y\right|
\mathbbm{1}(\calE_n)\\
&\quad\leq p^{-1/2}\Bigl\{
\|\bZ-\breve{\bZ}\|_2\|\bA\|_2\|\bZ\|_2
+\|\breve{\bZ}\|_2
\|\bA-\breve{\bOmega}^{-1}\|_2\|\bZ\|_2\\
&\hspace{45mm}
+\|\breve{\bZ}\|_2\|\breve{\bOmega}^{-1}\|_2
\|\bZ-\breve{\bZ}\|_2\Bigr\}\mathbbm{1}(\calE_n)
=o_p(1),
\end{align*}
uniformly over all such $x$ and $y$. In particular,
\begin{align}
&\sup_{t\in\calT_{\rm mc}(\varepsilon)}\sqrt p
\left|p^{-1}u_t^T\bZ^T\bA\bZ u_t
-p^{-1}u_t^T\breve{\bZ}^T\breve{\bOmega}^{-1}
\breve{\bZ}u_t\right|\mathbbm{1}(\calE_n)=o_p(1),
\label{eq:truncation_eq1}\\
&\sup_{t\in\calT_{\rm mc}(\varepsilon)}\max_{1\leq i\leq M}\sqrt p
\left|p^{-1}u_t^T\bZ^T\bA\bZ Q_i
-p^{-1}u_t^T\breve{\bZ}^T\breve{\bOmega}^{-1}
\breve{\bZ}Q_i\right|\mathbbm{1}(\calE_n)=o_p(1),
\label{eq:truncation_eq2}\\
&\max_{1\leq i,j\leq M}\sqrt p
\left|p^{-1}Q_i^T\bZ^T\bA\bZ Q_j
-p^{-1}Q_i^T\breve{\bZ}^T\breve{\bOmega}^{-1}
\breve{\bZ}Q_j\right|\mathbbm{1}(\calE_n)=o_p(1).
\label{eq:truncation_eq3}
\end{align}

Because \eqref{eq:truncation_event_probability} holds, the indicators in
\eqref{eq:truncation_eq1}--\eqref{eq:truncation_eq3} may be removed without
changing the stated stochastic orders. It remains to transfer these
comparisons through the Woodbury inverse.
Equation~\eqref{eq:truncation_eq3} and the result in
Subsection~\ref{subsec:weak_convergence_process} imply
\[
n^{-1}Q^T\breve{\bZ}^T\breve{\bOmega}^{-1}\breve{\bZ}Q
=\gamma_nI_M+O_p(n^{-1/2}).
\]
Hence,
\[
\left\|\left(I_M-n^{-1}Q^T\breve{\bZ}^T
\breve{\bOmega}^{-1}\breve{\bZ}Q\right)^{-1}\right\|_2=O_p(1).
\]
Similarly, \eqref{eq:truncation_eq2} and
\eqref{eq:H_uQ_uniform_bound} give
\[
\sup_{t\in\calT_{\rm mc}(\varepsilon)}\sqrt p
\left\|p^{-1}u_t^T\breve{\bZ}^T
\breve{\bOmega}^{-1}\breve{\bZ}Q\right\|_2=O_p(1).
\]
The Woodbury correction for the original variables is therefore
$O_p(p^{-1})$ uniformly in $t$. Combining this conclusion with
\eqref{eq:truncation_eq1} and \eqref{eq:truncation_event_probability} yields
\[
\sup_{t\in\calT_{\rm mc}(\varepsilon)}\sqrt p
\left|p^{-1}V_t-H(u_t)\right|=o_p(1).
\]
It follows in particular that
\(\sup_t|p^{-1}V_t-1|=O_p(p^{-1/2})\). Applying the same uniform Taylor
expansion of \(x^{1/2}\) around one shows that replacing the linear
normalization by
\[
\frac{2\sqrt n}{\sqrt{2n/p-2}}\{(p^{-1}V_t)^{1/2}-1\}
\]
changes the process by \(o_p(1)\) uniformly in \(t\). Thus, truncation,
centering, standardization, regularization, and the square-root transformation
do not alter the weak limit. Slutsky's theorem and
Subsection~\ref{subsec:weak_convergence_process} complete the proof of
Theorem~\ref{thm:main_null}.

\subsection{Discussion of the finite-moment conditions}\label{subsec:discussion_moment_condition}

We outline the modifications underlying Remark~\ref{remark:moment_condition}, retaining Conditions~\ref{enum:high_dimensional_regime}--\ref{enum:stable_cov_kernel} and a fixed $\varepsilon>0$. First of all, the finite-dimensional asymptotic normality of $\{H(u_t)\}$ after appropriate scaling only requires the existence of the 4th moment of $z_{11}$. Closely related arguments can be found in \citet{li2020high} and \citet{Pan2011central}; we therefore omit the details.  

The distinction between the scanning sets concerns tightness: the three-dimensional full scan $\calT_{\rm mc}(\varepsilon)$ requires a spatial exponent greater than three in the moment bound used in the grid-based Kolmogorov--Chentsov argument, whereas an exponent greater than one suffices for the single-change-point scan $\calT_{\rm sc}(\varepsilon)$, which is one-dimensional. In particular, we only need to show 
\[
\mE\Big|\sqrt p\{H(u_t)-\mE H(u_t)\} - \sqrt p\{H(u_{t'})-\mE H(u_{t'})\}\Big|^4
\leq\calK\|u_t-u_{t'}\|_2^4
\leq\calK\left(|r-r'|+n^{-1}\right)^2.
\]
Careful examination of the arguments in Subsection \ref{subsec:increment_bound} reveals that this bound can be established under the finite eighth moment condition. 



For $\calT_{\rm disc}(\varepsilon)$, the index set is finite and independent of $n$ because $\varepsilon$ is fixed. Thus, joint convergence of the finite vector of scan statistics suffices; no increment bound is needed. 

\section{Proof of Theorem \ref{thm:power_single} and Theorem \ref{thm:power_multiple}}\label{sec:proof_thm_power}

Because Theorem~\ref{thm:power_single} is a special case of
Theorem~\ref{thm:power_multiple}, it suffices to prove the latter.
We first remove the unrestricted common coefficient level. For any
$B\in\mathbb R^{M\times p}$, set $\bY^*=\bY-B^T\bX$ and
$P_0=I_n-QQ^T$. The zero-padded contrast vector satisfies
$\bX\bh_t=\theta-\theta=0$, and $\bX P_0=0$. Hence
\[
\bY^*\bh_t=\bY\bh_t,
\qquad
\bY^*P_0(\bY^*)^T=\bY P_0\bY^T.
\]
Thus, subtracting a common regression component leaves every scan
statistic unchanged. It also leaves the drift unchanged: for either
the actual or the limiting block-diagonal design, subtracting $B$
from every regime subtracts $B^T\bX\bh_t=0$ from the signal.

Take $B=\beta_{(1)}$ and define
$\Delta_{(i)}=\beta_{(i)}-\beta_{(1)}$, so that $\Delta_{(1)}=0$.
For the remainder of the proof, write $\bY$ for the centered response
$\bY-\beta_{(1)}^T\bX$. This exact reduction places no restriction
on the original $\beta_{(1)}$. We now work with the truncated and standardized variables satisfying
\eqref{eq:truncated_variable_condition} and retain the regularization
$\rho_n=n^{-3/2}$. The final step removes these modifications.

Let
\[
\calU=\Sigma_p^{-1/2}
\bigl(0_{p\times M},\Delta_{(2)}^T,\ldots,\Delta_{(s+1)}^T\bigr),
\qquad
\calV=\operatorname{BlockDiag}
\bigl(\bX_{(1)},\ldots,\bX_{(s+1)}\bigr),
\]
where $\calV$ is formed using the actual break points. Thus,
\[
\bY=\Sigma_p^{1/2}(\calU\calV+\bZ).
\]
As in the proof of Theorem~\ref{thm:main_null}, congruence invariance allows
us to set $\Sigma_p=I_p$. Set
\[
P_0=I_n-QQ^T,
\qquad
\widetilde{\bS}=n^{-1}\bZ P_0\bZ^T.
\]
The residual covariance matrix under the alternative can then be written as
\[
\bS=n^{-1}(\bZ+\calU\calV)P_0(\bZ+\calU\calV)^T
=\widetilde{\bS}+R_n,
\]
where
\begin{align}
R_n&=C_n+C_n^T+D_n,\nonumber\\
C_n&=n^{-1}\bZ P_0\calV^T\calU^T,
\qquad
D_n=n^{-1}\calU\calV P_0\calV^T\calU^T.
\label{eq:power_R_decomposition}
\end{align}

For $t\in\calT_{\rm mc}(\varepsilon)$, write
\[
z_t=\bZ u_t,
\qquad
m_t=\calU\calV u_t.
\]
Because $Q^Tu_t=0$, we also have $P_0u_t=u_t$. The regularized version of
the scan statistic is therefore
\begin{equation}
\label{eq:power_regularized_statistic}
p^{-1}V_{\rho}(t)
=p^{-1}(z_t+m_t)^T(\bS+\rho_nI_p)^{-1}(z_t+m_t).
\end{equation}

We next relate $m_t$ to the drift process in
\eqref{eq:def_U_multiple}. Let $\widetilde{\calV}$ be the block-diagonal
design matrix formed using the limiting break fractions. The boundedness of
the design entries and $B_i/n\to\widetilde B_i$ imply
$\|\calV-\widetilde{\calV}\|_2=o(n^{1/2})$. The upper bound in
\eqref{eq:local_alternative_mc} gives
\begin{equation}
\label{eq:power_U_rate}
\|\calU\|_2
=\left\|\Sigma_p^{-1/2}
\bigl(\Delta_{(2)}^T,\ldots,\Delta_{(s+1)}^T\bigr)\right\|_2
=O(n^{-1/4}).
\end{equation}
Consequently, uniformly over $t$,
\begin{equation}
\label{eq:power_signal_relation}
n^{-1}\|m_t\|_2^2
=\calS_n(t)+o(n^{-1/2}).
\end{equation}
Indeed, $\|\calV\|_2+\|\widetilde{\calV}\|_2=O(n^{1/2})$, so the
difference between the two sides is bounded by
\[
n^{-1}\|\calU\|_2^2
\|\calV-\widetilde{\calV}\|_2
(\|\calV\|_2+\|\widetilde{\calV}\|_2)=o(n^{-1/2}).
\]
In particular,
\begin{equation}
\label{eq:power_m_rate}
\sup_{t\in\calT_{\rm mc}(\varepsilon)}\|m_t\|_2=O(n^{1/4}).
\end{equation}

The null analysis applies directly to the noise-only residual covariance
$\widetilde{\bS}$. With
\[
\widetilde{\bG}=(\widetilde{\bS}+\rho_nI_p)^{-1},
\]
it gives
\begin{equation}
\label{eq:proof_alternative1}
\left\{
\frac{\sqrt p}{\sqrt{2(1-\gamma_n)}}
\left[p^{-1}z_t^T\widetilde{\bG}z_t-1\right]:
t\in\calT_{\rm mc}(\varepsilon)
\right\}
\rightsquigarrow\{G(t):t\in\calT_{\rm mc}(\varepsilon)\}.
\end{equation}

We record the auxiliary bounds needed for the perturbation expansion. Since
$s$ and $M$ are fixed, $\|\calV\|_2=O(n^{1/2})$ and
$\|\bZ\|_2=O_p(n^{1/2})$. Equations
\eqref{eq:power_R_decomposition} and \eqref{eq:power_U_rate} therefore yield
\begin{equation}
\label{eq:power_R_rate}
\|C_n\|_2=O_p(n^{-1/4}),
\qquad
\|D_n\|_2=O(n^{-1/2}),
\qquad
\|R_n\|_2=O_p(n^{-1/4}).
\end{equation}
The results of Section~\ref{sec:proof_thm_null} show that
$\|\widetilde{\bG}\|_2=O_p(1)$. Hence, by
\eqref{eq:power_R_rate} and the resolvent identity,
\begin{equation}
\label{eq:power_resolvent_bounds}
\|(\bS+\rho_nI_p)^{-1}\|_2=O_p(1).
\end{equation}

We next collect three consequences of the leave-one-out bounds used in the
null proof. Every column of $p^{-1/2}P_0\calV^T$ is regular. Indeed, the
entries of $\calV^T$ are uniformly bounded, while the entries of
$Q(Q^T\calV^T)$ are uniformly bounded because the columns of $Q$ are
regular and $M$ and $s$ are fixed. The moment and increment arguments in
Subsections~\ref{subsec:increment_bound} and
\ref{subsec:weak_convergence_process}, applied to these finitely many
columns, therefore give
\begin{equation}
\label{eq:power_projection_approximation}
\sup_t\left\|
n^{-1}\calV P_0\bZ^T\widetilde{\bG}z_t
-\gamma_n\calV u_t\right\|_2=O_p(1).
\end{equation}
Here and below, suprema are over $\calT_{\rm mc}(\varepsilon)$. The
companion anisotropic resolvent bounds for deterministic $p$-vectors \citep{knowles2017anisotropic},
obtained from the same leave-one-out expansion, are
\begin{align}
&\sup_t\|\calU^T\widetilde{\bG}z_t\|_2
=O_p(n^{-1/4}),
\label{eq:power_anisotropic_cross}\\
&\left\|\calU^T\widetilde{\bG}\calU
-\frac{1}{1-\gamma_n}\calU^T\calU\right\|_2
=o_p(n^{-1/2}).
\label{eq:power_anisotropic_quadratic}
\end{align}
The  projection $P_0$ is handled by the same Woodbury argument
as in Subsection~\ref{subsec:weak_convergence_process}. Since
$\sup_t\|\calV u_t\|_2=O(n^{1/2})$, equations
\eqref{eq:power_anisotropic_cross}--\eqref{eq:power_anisotropic_quadratic}
imply
\begin{equation}
\label{eq:power_cross_bounds}
\sup_t p^{-1}|m_t^T\widetilde{\bG}z_t|=o_p(p^{-1/2})
\end{equation}
and the isotropic deterministic equivalent \citep{knowles2017anisotropic}
\begin{equation}
\label{eq:power_signal_deterministic_equivalent}
\sup_t\left|
m_t^T\widetilde{\bG}m_t
-\frac{\|m_t\|_2^2}{1-\gamma_n}
\right|=o_p(p^{1/2}).
\end{equation}
The fixed dimensions of $\calU$ and $\calV$ make all of these bounds uniform
over their columns.

Set
\[
\bG=(\bS+\rho_nI_p)^{-1}.
\]
The iterated resolvent identity gives the exact expansion
\begin{equation}
\label{eq:power_resolvent_expansion}
\bG
=\widetilde{\bG}
-\widetilde{\bG}R_n\widetilde{\bG}
+\widetilde{\bG}R_n\widetilde{\bG}R_n\widetilde{\bG}
-\widetilde{\bG}R_n\widetilde{\bG}R_n
 \widetilde{\bG}R_n\bG.
\end{equation}
Equations~\eqref{eq:power_projection_approximation} and
\eqref{eq:power_anisotropic_cross} imply, uniformly in \(t\),
\begin{align}
C_n^T\widetilde{\bG}z_t
&=\gamma_nm_t+O_p(n^{-1/4}),\nonumber\\
C_n\widetilde{\bG}z_t&=O_p(n^{-1/4}),
\qquad
D_n\widetilde{\bG}z_t=O_p(n^{-1/2}).
\label{eq:power_R_action}
\end{align}
For example, the first identity follows by left-multiplying
\eqref{eq:power_projection_approximation} by \(\calU\), and the second uses
\(\|n^{-1}\bZ P_0\calV^T\|_2=O_p(1)\) together with
\eqref{eq:power_anisotropic_cross}. Hence,
\[
R_n\widetilde{\bG}z_t
=\gamma_nm_t+O_p(n^{-1/4}),
\qquad
R_n\widetilde{\bG}m_t=O_p(1),
\]
uniformly in \(t\).

We now apply \eqref{eq:power_resolvent_expansion} to
\eqref{eq:power_regularized_statistic}. First,
\begin{align}
p^{-1}(z_t+m_t)^T\widetilde{\bG}(z_t+m_t)
&=p^{-1}z_t^T\widetilde{\bG}z_t
+p^{-1}m_t^T\widetilde{\bG}m_t
+o_p(p^{-1/2}),
\label{eq:power_expansion_zero}
\end{align}
where the cross term is negligible by \eqref{eq:power_cross_bounds}. Next,
\eqref{eq:power_R_action} gives
\begin{align}
&-p^{-1}(z_t+m_t)^T
\widetilde{\bG}R_n\widetilde{\bG}(z_t+m_t)\nonumber\\
&\qquad=-2\gamma_np^{-1}
m_t^T\widetilde{\bG}m_t+o_p(p^{-1/2}),
\label{eq:power_expansion_first}\\
&p^{-1}(z_t+m_t)^T
\widetilde{\bG}R_n\widetilde{\bG}R_n
\widetilde{\bG}(z_t+m_t)\nonumber\\
&\qquad=\gamma_n^2p^{-1}
m_t^T\widetilde{\bG}m_t+o_p(p^{-1/2}).
\label{eq:power_expansion_second}
\end{align}
All residuals in \eqref{eq:power_expansion_zero}--\eqref{eq:power_expansion_second}
are uniform over the scanning set. In
\eqref{eq:power_expansion_first}, the two leading contributions are the
cross-products involving $m_t$ and $R_n\widetilde{\bG}z_t$; each equals
$\gamma_nm_t^T\widetilde{\bG}m_t$ to the required order. In
\eqref{eq:power_expansion_second}, the leading contribution is
$(R_n\widetilde{\bG}z_t)^T\widetilde{\bG}
(R_n\widetilde{\bG}z_t)$, which explains the factor $\gamma_n^2$.

It remains to bound the last term in
\eqref{eq:power_resolvent_expansion}. The moment and increment bound used in
Subsection~\ref{subsec:weak_convergence_process} also gives
\[
\sup_t p^{-1}\|z_t\|_2^2=O_p(1).
\]
Together with \eqref{eq:power_m_rate}, \eqref{eq:power_R_rate}, and
\eqref{eq:power_resolvent_bounds}, this yields
\begin{align*}
&\sup_t p^{-1}\left|
(z_t+m_t)^T\widetilde{\bG}R_n\widetilde{\bG}R_n
\widetilde{\bG}R_n\bG(z_t+m_t)\right|\\
&\qquad\leq
O_p(1)\|R_n\|_2^3
=O_p(n^{-3/4})
=o_p(p^{-1/2}).
\end{align*}
Combining the four terms in the resolvent expansion, we obtain
\begin{equation}
\label{eq:power_main_expansion}
p^{-1}V_\rho(t)
=p^{-1}z_t^T\widetilde{\bG}z_t
+(1-\gamma_n)^2p^{-1}
m_t^T\widetilde{\bG}m_t
+o_p(p^{-1/2})
\end{equation}
uniformly over \(t\in\calT_{\rm mc}(\varepsilon)\).

Substituting \eqref{eq:power_signal_deterministic_equivalent} into
\eqref{eq:power_main_expansion} gives
\begin{align}
p^{-1}V_\rho(t)
&=p^{-1}z_t^T\widetilde{\bG}z_t
+(1-\gamma_n)p^{-1}\|m_t\|_2^2
+o_p(p^{-1/2})\nonumber\\
&=p^{-1}z_t^T\widetilde{\bG}z_t
+\frac{1-\gamma_n}{\gamma_n}\calS_n(t)
+o_p(p^{-1/2}),
\label{eq:power_drift_expansion}
\end{align}
uniformly over the scanning set, where the second equality uses
\eqref{eq:power_signal_relation} and \(p/n=\gamma_n\). Therefore,
\begin{align}
&\frac{\sqrt p}{\sqrt{2(1-\gamma_n)}}
\{p^{-1}V_\rho(t)-1\}\nonumber\\
&\quad=
\frac{\sqrt p}{\sqrt{2(1-\gamma_n)}}
\{p^{-1}z_t^T\widetilde{\bG}z_t-1\}
+{\frac{\sqrt n}{\sqrt{2\gamma_n}}
(1-\gamma_n)^{1/2}}\calS_n(t)+o_p(1),
\label{eq:power_process_expansion}
\end{align}
where the remainder is uniform in \(t\).

We now remove regularization. The null analysis shows that the smallest
eigenvalue of \(\widetilde{\bS}\) is bounded away from zero with probability
tending to one. By \eqref{eq:power_R_rate} and Weyl's inequality, the same is
true of \(\bS\). On this event,
\[
\|\bS^{-1}-(\bS+\rho_nI_p)^{-1}\|_2=O_p(\rho_n).
\]
Since
\[
\sup_t p^{-1}\|z_t+m_t\|_2^2=O_p(1)
\quad\text{and}\quad
\sqrt p\,\rho_n\longrightarrow0,
\]
replacing \(V_\rho(t)\) by \(V(t)\) changes the left-hand side of
\eqref{eq:power_process_expansion} by \(o_p(1)\) uniformly in \(t\).
The coupling argument in Subsection~\ref{subsec:truncation} applies with
\(\bZ\) replaced by \(\bZ+\calU\calV\). In view of
\eqref{eq:power_U_rate} and \(\|\calV\|_2=O(n^{1/2})\), the additional
deterministic component has norm \(O(n^{1/4})\). More explicitly,
\eqref{eq:truncation_matrix_difference} and
$\|\bZ+\calU\calV\|_2+\|\breve{\bZ}+\calU\calV\|_2=O_p(p^{1/2})$ imply
\begin{align*}
&\left\|n^{-1}(\bZ+\calU\calV)P_0(\bZ+\calU\calV)^T
-n^{-1}(\breve{\bZ}+\calU\calV)P_0
(\breve{\bZ}+\calU\calV)^T\right\|_2
=o_p(p^{-1}),\\
&\sup_t\|(\bZ-\breve{\bZ})u_t\|_2=o_p(p^{-1/2}).
\end{align*}
The inverse perturbation identity and the eigenvalue bounds used above then
show that the corresponding standardized statistics differ by $o_p(1)$
uniformly over the scanning set. Thus,
\eqref{eq:power_process_expansion} also holds for the original variables
under Condition~\ref{enum:moment_conditions}.

It remains to pass from the intermediate linear normalization to the
square-root statistic \(D(t)\) used in the main text. The preceding expansion,
the tightness in \eqref{eq:proof_alternative1}, and the local-signal bound
imply
\[
\sup_{t\in\calT_{\rm mc}(\varepsilon)}
\left|p^{-1}V(t)-1\right|=O_p(p^{-1/2}).
\]
Therefore, uniformly over the scanning set,
\[
2\{(p^{-1}V(t))^{1/2}-1\}
=\{p^{-1}V(t)-1\}+O_p(p^{-1}).
\]
After multiplication by \(\sqrt p/\sqrt{2(1-\gamma_n)}\), the additional
remainder is \(o_p(1)\). Consequently, \eqref{eq:power_process_expansion}
remains valid with its left-hand side replaced by the transformed statistic
\[
D(t)=
\frac{2\sqrt n}{\sqrt{2n/p-2}}
\{(p^{-1}V(t))^{1/2}-1\}.
\]
The deterministic drift is unchanged, because the derivative of
\(x^{1/2}\) at \(x=1\) is \(1/2\), which is canceled by the leading factor
2 in the definition of \(D(t)\).

To finish the probability comparison, define
\[
d_n(t)=
\frac{\sqrt n}{\sqrt{2\gamma_n}}(1-\gamma_n)^{1/2}\calS_n(t).
\]
The upper local-signal bound, the boundedness of the design, and
\eqref{eq:u_increment_bound} imply
\[
\sup_t|d_n(t)|\leq\calK,
\qquad
|d_n(t)-d_n(t')|
\leq\calK\left(\|t-t'\|_2+n^{-1}\right)^{1/2},
\]
uniformly in \(t,t'\). A diagonal extraction on successively finer finite
meshes of the compact scanning set, followed by the preceding increment
bound, shows that every subsequence contains a further subsequence along
which $d_n$ converges uniformly to a continuous function $d$. Along such a
subsequence, equations~\eqref{eq:proof_alternative1} and
\eqref{eq:power_process_expansion} and the continuous mapping theorem give
\[
\sup_t D(t)\ \stackrel{D}{\longrightarrow}\ \sup_t\{G(t)+d(t)\},
\qquad
\sup_t\{G(t)+d_n(t)\}\ \stackrel{D}{\longrightarrow}\
\ \sup_t\{G(t)+d(t)\}.
\]
Since $\calK(t,t)=1$, a finite-mesh approximation together with the standard
Gaussian anti-concentration argument implies that the limiting random
variable has no atoms; see \citet{chernozhukov2014anti}. Therefore, at the
fixed critical value
$\widetilde\xi_{\rm mc}(1-\alpha)$,
\begin{align*}
&\mP\left\{\sup_{t\in\calT_{\rm mc}(\varepsilon)}D(t)>
\widetilde\xi_{\rm mc}(1-\alpha)\right\}\\
&\quad-
\mP\left\{\sup_{t\in\calT_{\rm mc}(\varepsilon)}
[G(t)+d_n(t)]>\widetilde\xi_{\rm mc}(1-\alpha)\right\}
\longrightarrow0.
\end{align*}
Because the argument applies to every subsequence, it proves
Theorem~\ref{thm:power_multiple}. Restricting the index set to
\(\calT_{\rm sc}(\varepsilon)\) proves Theorem~\ref{thm:power_single}; the
same argument also applies to \(\calT_{\rm disc}(\varepsilon)\).
Only the upper bound on the centered coefficient matrix is used in this
argument; no positive lower signal bound is required. The initial
baseline reduction preserves both the statistic and the drift, so the
conclusion holds for arbitrary common coefficient levels in the original model.

\end{document}